\def \withParikhbyBlocks {}
\def \VersionarXiv {}
	\documentclass[10pt,a4paper,twocolumn]{article}
	\usepackage[a4paper,margin=1.9cm]{geometry}
	\usepackage{microtype}          %
	\usepackage{parskip}            %
	\usepackage{setspace}           %

	\usepackage{titlesec}
	\titleformat{\section}{\large\bfseries}{\thesection}{1em}{}
	\titleformat{\subsection}{\normalsize\bfseries}{\thesubsection}{0.75em}{}

	\usepackage[labelfont=bf]{caption}
	\usepackage{booktabs}

	\usepackage{fontawesome}
	\newcommand{\homepage}[1]{\href{#1}{\color{gray}\faHome}}

	\usepackage{amssymb} %

	\newenvironment{acks}
		{\section*{Acknowledgments}}
		{\medskip}
\makeatletter
\AtBeginDocument{%
	\@ifpackageloaded{hyperref}
	{\def\@doi#1{\href{https://doi.org/#1}
			{\ttfamily https://doi.org/#1}\egroup}}
	{\def\@doi#1{\ttfamily https://doi.org/#1\egroup}}
	\def\doi{\bgroup\catcode`\_=12\relax\@doi}}
\makeatother

\usepackage[utf8]{inputenc}
\usepackage{multirow}
\usepackage{array}
\usepackage[ruled,vlined,linesnumbered]{algorithm2e}
\SetKwInOut{Input}{input}
\SetKwInOut{Output}{output}

\usepackage{subcaption}

\usepackage{amsmath} %
\usepackage{stmaryrd} %

\usepackage{graphicx}
\graphicspath{{figures/}}

\usepackage{thm-restate}

\usepackage{paralist} %

\newenvironment{oneenumerate}
{\ifdefined\VersionarXiv\begin{enumerate}\else\begin{inparaenum}[1)]\fi}
	{\ifdefined\VersionarXiv\end{enumerate}\else\end{inparaenum}\fi}

\usepackage[svgnames,table]{xcolor}

\definecolor{LSdeuxNrichblack}{HTML}{020A13}
\definecolor{LSdeuxNbazul}{HTML}{1273CE}
\definecolor{LSdeuxNcornflower}{HTML}{5297FF}
\definecolor{LSdeuxNAntiflashwhite}{HTML}{F8FBFF}

	\usepackage[
		pdfauthor={Etienne Andre and Lydia Bakiri},%
		pdftitle={Opacity problems in multi-energy timed automata},
		breaklinks  = true,
		colorlinks  = true,
		citecolor   = LSdeuxNcornflower,
		linkcolor   = LSdeuxNbazul,
		urlcolor    = LSdeuxNbazul,
	]{hyperref}
\usepackage[capitalise,english,nameinlink]{cleveref} %
\crefname{line}{\text{line}}{\text{lines}} %

	\usepackage{orcidlink}
	\newcommand{\orcidID}[1]{\orcidlink{#1}}

\newcommand{\defProblemGen}[4]
{%
	\noindent\fcolorbox{LSdeuxNrichblack}{LSdeuxNAntiflashwhite}{
		\begin{minipage}{.95\columnwidth}
			\textbf{#1#4:}\\
			\textsc{Input}: #2\\
			\textsc{Problem}: #3
		\end{minipage}
	}
	
	\smallskip
	
}

\newcommand{\defProblem}[3]
{%
	\defProblemGen{#1}{#2}{#3}{~problem}
}

\usepackage{tikz}
\usetikzlibrary{arrows,automata,calc,fit,positioning,shapes}

\tikzstyle{every node}     = [initial text=]
\tikzstyle{PTA}            = [auto, ->, >=stealth']
\tikzstyle{NFA}            = [auto, ->, >=stealth']
\tikzstyle{location}       = [rectangle, rounded corners, minimum size= 12pt, draw=black, fill=blue!15, inner sep=2pt, align=center]
\tikzstyle{location2CM}    = [location,thick,fill=green!30]
\tikzstyle{locbad}         = [location,fill=orange!50]
\tikzstyle{locunreachable} = [location, gray, fill=white]
\tikzstyle{sink}           = [fill=gray!50]
\tikzstyle{final}          = [double,fill=green!50]
\tikzstyle{private}        = [thick,fill=red!50]
\tikzstyle{urgent}         = [semithick,fill=yellow]
\tikzstyle{invariant}      = [draw=black, dotted, fill=purple!10, inner sep= 1pt, node distance=0] %
\tikzstyle{invariantNorth} = [invariant, yshift=1em]
\tikzstyle{invariantSouth} = [invariant, yshift=-1em]
\tikzstyle{invariantWest}  = [invariant, xshift=-1.5em]
\tikzstyle{invariantEast}  = [invariant, xshift=+1.5em]
\tikzstyle{edge}           = [->,  >=stealth']
\tikzstyle{state}          = [rectangle, minimum size= 12pt, draw=black, fill=black!10, inner sep=2pt, align=center]

\definecolor{coloract}{rgb}{0, 0.3, 0}
\definecolor{colorclock}{rgb}{0.7, 0, 0}
\definecolor{colordisc}{rgb}{1, 0, 1}
\definecolor{colorenergy}{rgb}{1, 0, 1}
\definecolor{colorloc}{rgb}{0.4, 0.4, 0.65}
\definecolor{colorparam}{rgb}{1, 0.6, 0.0}

\newcommand{\rowHeader}{\rowcolor{LSdeuxNcornflower}}

\newcommand{\cellDec}{\cellcolor{green!40}}
\newcommand{\cellUndec}{\cellcolor{red!40}}
\newcommand{\cellOpen}{\cellcolor{yellow!40}}

\usepackage{amsthm}
\theoremstyle{plain}
\newtheorem{lemma}{Lemma}
\newtheorem{proposition}{Proposition}
\newtheorem{theorem}{Theorem}
\newtheorem{corollary}{Corollary}

\theoremstyle{definition}
\newtheorem{definition}{Definition}
\newtheorem{example}{Example}

\theoremstyle{remark}
\newtheorem{remark}{Remark}

\newcommand{\setN}{\ensuremath{\mathbb N}}
\newcommand{\setQ}{\ensuremath{{\mathbb Q}}}
\newcommand{\setQplus}{\ensuremath{\setQ_{+}}} %
\newcommand{\setR}{\ensuremath{\mathbb R}}
\newcommand{\setRgeqzero}{\ensuremath{\setR_{\geq 0}}}
\newcommand{\setZ}{\ensuremath{\mathbb Z}}

\newcommand{\init}{\ensuremath{_0}}
\newcommand{\generalFigsScaleFactor}{1}

\newif\ifinfigure
\let\infigure\iffalse %

\let\oldtikzpicture\tikzpicture
\let\endoldtikzpicture\endtikzpicture
\renewenvironment{tikzpicture}
  {\let\ifinfigure\iftrue\oldtikzpicture}
  {\endoldtikzpicture\let\ifinfigure\iffalse}

\newcommand{\styleclock}[1]{\ensuremath{
	\ifinfigure%
		\textcolor{colorclock}{{#1}}%
	\else%
		#1%
	\fi%
	}%
}

\newcommand{\styleenergy}[1]{\ensuremath{
	\ifinfigure%
		\textcolor{colorenergy}{{#1}}%
	\else%
		#1%
	\fi%
	}%
}

\newcommand{\styleact}[1]{\ensuremath{
	\ifinfigure%
		\textcolor{coloract}{{#1}}%
	\else%
		#1%
	\fi%
	}%
}

\newcommand{\styleloc}[1]{\ensuremath{
	\ifinfigure%
		\textcolor{colorloc}{{#1}}%
	\else%
		#1%
	\fi%
	}%
}

\newcommand{\TA}{\ensuremath{\mathcal{A}}}
\newcommand{\TN}{\ensuremath{\mathcal{N}}}
\newcommand{\TM}{\ensuremath{\mathcal{M}}}
\newcommand{\absTimeRun}{\ensuremath{\mathit{absT}}}
\newcommand{\abstime}{\ensuremath{\tau}}
\newcommand{\destutter}{\ensuremath{\mathit{destutter}}}
\newcommand{\subseqAndProject}{\ensuremath{\mathit{subseqProj}}}
\newcommand{\seqrun}{\ensuremath{\mathit{seq}}}
\newcommand{\Apriv}[1]{\ensuremath{{#1}_\mathit{priv}}}
\newcommand{\Apub}[1]{\ensuremath{{#1}_\mathit{pub}}}
\newcommand{\Actions}{\ensuremath{\Sigma}}
\newcommand{\action}{\ensuremath{\styleact{a}}}

\newcommand{\silentaction}{\ensuremath{\styleact{\varepsilon}}}

\newcommand{\assign}{\leftarrow}

\newcommand{\Clock}{\ensuremath{\styleclock{\mathcal{X}}}} %
\newcommand{\ClockCard}{H} %
\newcommand{\clock}{\ensuremath{\styleclock{x}}}
\newcommand{\clocki}[1]{\ensuremath{\styleclock{\clock_{#1}}}}
\newcommand{\clocky}{\styleclock{y}} %
\newcommand{\clockval}{w} %
\newcommand{\clockvalZero}{\ensuremath{\vec{0}_\Clock}}
\newcommand{\clockZeroTime}{\ensuremath{\styleclock{z_0}}} %
\newcommand{\clockTickOne}{\ensuremath{\styleclock{z_t}}} %
\newcommand{\duration}{\ensuremath{\mathit{dur}}} %
\newcommand{\runduration}[1]{\ensuremath{\duration}(#1)}
\newcommand{\edge}{\ensuremath{e}}
\newcommand{\edgei}[1]{\ensuremath{\edge_{#1}}}
\newcommand{\Edges}{E}
\newcommand{\fleche}[1]{\stackrel{#1}{\rightarrow}}
\newcommand{\flow}{\ensuremath{\mathit{fl}}}
\newcommand{\TransConcrete}{\ensuremath{{\rightarrow}}}
\newcommand{\FlecheConcrete}[1]{\stackrel{#1}{\rightarrowtail}}
\newcommand{\TransConcreteEdge}{\ensuremath{{\rightarrowtail}}}

\newcommand{\subR}{{\mathbb D}}
\newcommand{\guard}{\ensuremath{g}}

\newcommand{\subinterval}{\ensuremath{\nu}}
\newcommand{\invariant}{\ensuremath{I}}
\newcommand{\Language}{\ensuremath{\mathcal{L}}}
\newcommand{\LanguageIncrDecr}{\ensuremath{\Language_{\geq 0}}}
\newcommand{\loc}{\ensuremath{\ell}} %
\newcommand{\loci}[1]{\ensuremath{\loc_{#1}}}
\newcommand{\lochalt}{\ensuremath{\cmspta_{\cmshaltname}}}
\newcommand{\locinit}{\loc\init}
\newcommand{\LocSet}{\ensuremath{L}} %
\newcommand{\locpriv}{\ensuremath{\loc_{\mathit{priv}}}}
\newcommand{\LocsPriv}{\ensuremath{\LocSet_{\mathit{priv}}}}
\newcommand{\locfinal}{\ensuremath{\loc_f}}
\newcommand{\LocsFinal}{\ensuremath{\LocSet_{f}}}

\newcommand{\maxConstantTReachability}{\ensuremath{C}} %
\newcommand{\maxConstantEnergy}{\ensuremath{M}} %
\newcommand{\paramd}{\ensuremath{d}}
\newcommand{\ParikbB}{\ensuremath{\mathit{PbB}}}
\newcommand{\Q}{\ensuremath{S}}
\newcommand{\q}{\ensuremath{s}}
\newcommand{\qinit}{\ensuremath{\q\init}}
\newcommand{\qfinal}{\ensuremath{\q_f}}
\newcommand{\qi}[1]{\ensuremath{\q_{#1}}}
\newcommand{\resets}{\ensuremath{R}}
\newcommand{\concstate}{\ensuremath{\q}} %
\newcommand{\tickIncr}{\ensuremath{\styleact{\mathit{inc}}}}
\newcommand{\tickDecr}{\ensuremath{\styleact{\mathit{dec}}}}
\newcommand{\tickFin}{\ensuremath{\styleact{f}}}
\newcommand{\tickTime}{\ensuremath{\styleact{t}}}
\newcommand{\tickLastTimeRational}{\ensuremath{\styleact{t_{>0}}}}

\newcommand{\timelapseFlow}{\ensuremath{\mathit{te}}}
\newcommand{\Variables}{\ensuremath{V}}
\newcommand{\variable}{\ensuremath{v}}
\newcommand{\varrun}{\ensuremath{\rho}} %
\newcommand{\updateConstant}{\ensuremath{N}}
\newcommand{\wv}[2]{#1|#2} %

\newcommand{\compop}{\bowtie}

\newcommand{\reset}[2]{\ensuremath{[#1]_{#2}}}

\newcommand{\cardinality}[1]{\ensuremath{\lvert #1 \rvert}}

\newcommand{\PrivDurVisit}[1]{\ensuremath{\mathit{DVisit}^\mathit{priv}(#1)}}
\newcommand{\PrivEnerVisit}[1]{\ensuremath{\mathit{EVisit}^\mathit{priv}(#1)}}
\newcommand{\PubDurVisit}[1]{\ensuremath{\mathit{DVisit}^{\mathit{pub}}}(#1)}
\newcommand{\PubEnerVisit}[1]{\ensuremath{\mathit{EVisit}^{\mathit{pub}}}(#1)}

\newcommand{\existsWeakFull}{\ensuremath{\delta}}
\newcommand{\ENorETEN}{\ensuremath{\sigma}}

\newcommand{\PrivVisit}[1]{\ensuremath{\mathit{Visit}^{\mathit{priv}}(#1)}}
\newcommand{\PubVisit}[1]{\ensuremath{\mathit{Visit}^{\mathit{pub}}(#1)}}

\newcommand{\buffDiscreteEnergyObs}{\ensuremath{\mathit{bDEO}}}
\newcommand{\DiscreteEnergyObs}{\ensuremath{\mathit{DEO}}}
\newcommand{\Energies}{\ensuremath{\mathcal{E}}}
\newcommand{\EnergiesCard}{\ensuremath{M}}
\newcommand{\enervar}{\ensuremath{\styleenergy{\eta}}}
\newcommand{\enervari}[1]{\ensuremath{\styleenergy{\enervar_{#1}}}}
\newcommand{\enerval}{\ensuremath{v}}
\newcommand{\enervalZero}{\ensuremath{\vec{0}_\Energies}}
\newcommand{\EnerLevelT}{\ensuremath{\mathit{EL}}}
\newcommand{\EnerLevelBufferT}{\ensuremath{\mathit{bEL}}}
\newcommand{\enerupdates}{\ensuremath{U}}
\newcommand{\rates}{\ensuremath{\phi}}
\newcommand{\finalEnergy}{\ensuremath{\mathit{FE}}}
\newcommand{\emptyseq}{\ensuremath{\epsilon}}
\newcommand{\emptyword}{\ensuremath{\epsilon}}

\newcommand{\battery}{\ensuremath{\styleenergy{b}}}
\newcommand{\temperature}{\ensuremath{\styleenergy{t}}}

\newcommand{\LargestConstant}{\ensuremath{M}}
\newcommand{\region}{\ensuremath{r}}
\newcommand{\regioni}[1]{\ensuremath{\region_{#1}}}
\newcommand{\Regions}[1]{\ensuremath{\mathcal{R}_{#1}}}
\newcommand{\RegionAutomaton}[1]{\ensuremath{\mathcal{RA}({#1})}}

\newcommand{\intpart}[1]{\ensuremath{\lfloor#1\rfloor}}
\newcommand{\ceiling}[1]{\ensuremath{\lceil #1 \rceil}}
\newcommand{\fract}[1]{\ensuremath{\text{frac}(#1)}}

\newcommand{\stylePDA}[1]{\ensuremath{
	\ifinfigure%
		\textcolor{blue!70!green}{{#1}}%
	\else%
		#1%
	\fi%
	}%
}

\newcommand{\fontPDA}[1]{\ensuremath{\mathfrak{#1}}}
\newcommand{\PDAstackEnergy}{\ensuremath{\stylePDA{\fontPDA{e}}}}
\newcommand{\PDAstackAlphabet}{\ensuremath{\Gamma}}
\newcommand{\initialStackSymbol}{\ensuremath{\stylePDA{\bot}}}
\newcommand{\stackSymbol}{\ensuremath{\stylePDA{\fontPDA{a}}}}
\newcommand{\PDAstate}{\ensuremath{\stylePDA{\fontPDA{q}}}}
\newcommand{\PDAstateinit}{\ensuremath{\stylePDA{\PDAstate_0}}}
\newcommand{\PDAStates}{\ensuremath{\stylePDA{\fontPDA{Q}}}}
\newcommand{\PDAStatesFinal}{\ensuremath{\stylePDA{\PDAStates_F}}}
\newcommand{\PDAEdges}{\ensuremath{\stylePDA{\fontPDA{E}}}}
\newcommand{\PDAedge}{\ensuremath{\stylePDA{\fontPDA{e}}}}

\newcommand{\AtomicProp}{\ensuremath{AP}}

\newcommand{\LabelFunc}{\ensuremath{\gamma}}
\newcommand{\droneActCharge}{\ensuremath{\styleact{\mathit{charge}}}}
\newcommand{\droneActCool}{\ensuremath{\styleact{\mathit{cool}}}}
\newcommand{\droneActExplode}{\ensuremath{\styleact{\mathit{explode}}}}
\newcommand{\droneActFlash}{\ensuremath{\styleact{\mathit{flash}}}}
\newcommand{\droneActFly}{\ensuremath{\styleact{\mathit{fly}}}}
\newcommand{\droneActLand}{\ensuremath{\styleact{\mathit{land}}}}
\newcommand{\droneActOff}{\ensuremath{\styleact{\mathit{off}}}}
\newcommand{\droneActRest}{\ensuremath{\styleact{\mathit{rest}}}}

\newcommand{\droneLocCharging}{\ensuremath{\styleloc{\text{charging}}}}
\newcommand{\droneLocCooling}{\ensuremath{\styleloc{\text{cooling}}}}
\newcommand{\droneLocExploded}{\ensuremath{\styleloc{\text{exploded}}}}
\newcommand{\droneLocFlying}{\ensuremath{\styleloc{\text{flying}}}}
\newcommand{\droneLocOutOfBattery}{\ensuremath{\styleloc{\text{out of battery}}}}
\newcommand{\droneLocFlashed}{\ensuremath{\styleloc{\text{flashed}}}}
\newcommand{\droneLocStandby}{\ensuremath{\styleloc{\text{standby}}}}
\newcommand{\droneLocSuccess}{\ensuremath{\styleloc{\text{mission success}}}}

\newcommand{\calM}{\ensuremath{\mathcal{M}}}

\newcommand{\Counter}{\ensuremath{\mathcal{C}}}
\newcommand{\counterGenericIndex}{\ensuremath{l}}

\newcommand{\cmshaltname}{\ensuremath{\textrm{halt}}} %
\newcommand{\cms}{\ensuremath{\mathtt{q}}} %
\newcommand{\cmshalt}{\ensuremath{\cms_{\cmshaltname}}} %
\newcommand{\cmspta}{\ensuremath{q}} %

\newcommand{\clockCMt}{\ensuremath{\styleclock{t}}}

\newcommand{\ComplexityFont}[1]{{\sffamily\upshape #1}}

\newcommand{\PSPACE}{\ComplexityFont{PSPACE}\xspace}

\newcommand{\coNP}{\ComplexityFont{co-NP}\xspace}

\newcommand{\EXPSPACE}{\ComplexityFont{EXPSPACE}\xspace}
\newcommand{\twoEXPSPACE}{\ComplexityFont{2EXPSPACE}\xspace}
\newcommand{\threeEXPSPACE}{\ComplexityFont{3EXPSPACE}\xspace}

\newcommand{\eg}{e.g.,\xspace}

\newcommand{\ie}{i.e.,\xspace}

\newcommand{\wrt}{w.r.t.\@}

\newcommand{\ourAbstract}{%
	Cyber-physical systems can be subject to information leakage; in the presence of continuous variables such as time and energy, these leaks can be subtle to detect.
	We study here the verification of opacity problems over systems with observation over both timing and energy information.
	We introduce guarded multi-energy timed automata as an extension of timed automata with multiple energy variables and guards over such variables.
	Despite undecidability of this general formalism, we establish positive results over a number of subclasses, notably when the attacker observes the final energy and/or the execution time, but also when they have access to the value of the energy variables every time unit.
}

\newcommand{\ourKeywords}{verification, security properties, opacity, timed systems}

	\author{}
	\date{}

\begin{document}

\title{Opacity problems in multi-energy timed automata\thanks}
\twocolumn[
  \begin{@twocolumnfalse}
	\maketitle

	\noindent{}\textbf{Étienne André\homepage{https://lipn.univ-paris13.fr/~andre/}%
		\orcidID{0000-0001-8473-9555}}
	\\
	{\em\small{}Nantes Université, CNRS, LS2N, Nantes, France}
	\\
	{\em\small{}Institut universitaire de France (IUF)}

	\smallskip

	\noindent{}\textbf{Lydia Bakiri\orcidID{0009-0003-8655-6455}}
	\\
	{\em\small{}Université Sorbonne Paris Nord, LIPN, CNRS UMR 7030, F-93430 Villetaneuse, France}
	\\
	{\em\small{}LIX, CNRS, École polytechnique, Institut Polytechnique de Paris, Palaiseau, France}

	\bigskip
	\bigskip

	{\color{gray}\hrule}

	\medskip

	\begin{abstract}
		\ourAbstract{}

		\medskip
		\noindent{}\textbf{Keywords:} \ourKeywords{}
	\end{abstract}

	\medskip

	{\color{gray}\hrule}

	\bigskip
	\bigskip

  \end{@twocolumnfalse}
]

\footnotetext{This is the author version (extended with all proofs) of the manuscript of the same name published in the proceedings of the 41st~ACM/SIGAPP Symposium on Applied Computing (SAC 2026).}

	\section{Introduction}\label{section:introduction}

	Modelling and verifying real-time systems require formal models that can accurately represent timing constraints.
	\emph{Timed automata} (TAs)~\cite{AD94} became a foundational model for specifying and analysing such systems.
	TAs~extend finite automata with real-valued clocks, enabling the representation of time-dependent behaviours.

	However, in addition to timing, many cyber-physical systems are also subject to quantitative constraints, particularly related to the consumption or accumulation of resources like energy, memory, or cost.
	To incorporate such quantitative considerations, various extensions of timed automata have been proposed.
	\emph{Priced timed automata} (also known as \emph{weighted} or \emph{cost timed automata}) augment the model with cost variables that accumulate according to rates and discrete updates~\cite{BFHLPRV01,BLR05}.
	More specifically, \emph{energy timed automata} model the dynamic consumption and replenishment of energy in systems (\eg{} \cite{BFLMS08,BBFLMR21,CHL24}).
	These extensions enable the specification and verification of energy-aware properties, such as ``the system can always recharge before the battery is depleted'' or ``the total energy cost never exceeds a budget''.

	\paragraph{Opacity}
The concept of \emph{opacity}~\cite{Mazare04,BKMR08,LLH18} pertains to the potential leakage of information from a system to an attacker. Specifically, it captures the attacker's ability to infer secret information based on publicly observable behaviours.
A system is considered opaque if an attacker with access to only a subset of the system's observable actions cannot determine whether a specific sequence of actions has taken place.
Time particularly influences the deductive capabilities of the attacker.
It has been shown in~\cite{GMR07} that it is possible for models that are opaque when timing constraints are omitted, to become non-opaque when those constraints are added to the model.
Franck~Cassez proposed in~\cite{Cassez09} a first definition of \emph{timed} opacity for TAs: the system is opaque when an attacker can never deduce whether some secret sequence of actions (possibly with timestamps) was performed, by only observing a given set of observable actions together with their timestamp.
It is then proved in~\cite{Cassez09} that it is undecidable whether a TA is opaque.
The aforementioned negative result leaves hope only if the definition or the setting is changed.
First, in~\cite{WZ18,WZA18}, the input model is simplified to \emph{real-time automata}.
\cite{LLHL22}~exhibits decidability results for constant-time labelled automata.
In~\cite{Zhang24}, Zhang studies decidability for labelled real-timed automata.

Second, in~\cite{AGWZH24,ADL24,KKG24}, decidability results are exhibited in the setting of Cassez' definition, but with restrictions in the model: one-clock automata, one-action automata, or over discrete time.
Third, in~\cite{AEYM21}, the authors consider a \emph{time-bounded} notion of the opacity of~\cite{Cassez09}, where the attacker has to disclose the secret before a deadline, using a partial observability.
A somehow similar framework is considered in~\cite{SLR23}, in which the attacker has a bounded memory, and a finite duration between distinct observations is required, in which case the problem is decidable and \PSPACE{}-complete.

Fourth, in~\cite{ALLMS23}, an alternative definition is proposed, by studying execution-time opacity: the attacker has only access to the \emph{execution time} of the system, as opposed to Cassez' partial observations where some events (with their timestamps) are observable.
The goal for the attacker is to deduce whether a special secret location was visited, by observing only the execution time.
In that case, most problems for TAs become decidable.
Control is studied in~\cite{ABLM22,ADLL24}.

Regarding non-interference for TAs, some decidability results are proved in~\cite{BDST02,BT03}, while control was considered in~\cite{BCLR15} and a parametric timed extension in~\cite{AK20}.
General security problems for TAs are surveyed in~\cite{AA23survey}.

\paragraph{Contributions}
Here, we extend execution-time opacity~\cite{ALLMS23} by adding the observation of energy.
Our attacker model is as follows: the attacker knows the input model, and observes only final energy and/or execution time; from this observation, they aim at deducing whether a special secret location was visited.
Our first contribution is to define a general framework for opacity in timed automata extended with multiple energy variables and energy constraints---a formalism called \emph{guarded multi-energy timed automaton} (guarded META).
In this general setting, energy can continuously increase or decrease, and be subject to discrete updates, thus encoding concepts such as battery or temperature; energy constraints can also be part of our model, encoding behaviours such as ``when the battery level decreases below some threshold, the system enters degraded mode''.
Our second and presumably main contribution is to study opacity (with an observation of either final energy, or final time and energy) in \emph{discrete} guarded METAs, in which energy is updated only via discrete updates (in an entirely continuous, real-time setting).
The general subclass is undecidable without surprise (as discrete updates can simulate a 2-counter machine); we therefore exhibit several decidable subclasses.
Our third contribution is to study a setting with a stronger attacker, able to observe the energy level at each time unit, and for which we also exhibit decidable subclasses.
We finally identify a first decidable subclass with continuous energy rates in addition to the discrete updates.
Our proofs rely on different techniques, notably by defining modified versions of the input META and/or of the region graph, and subsequently studying (untimed) regular or context-free language problems.
To the best of our knowledge, this work is the first to define and study opacity problems over TAs extended with multiple energy variables.

\paragraph{Outline}
\cref{section:preliminaries} introduces the necessary material.
\cref{section:problems} introduces opacity problems in guarded METAs.
\cref{section:generic} introduces common constructions for our subsequent proofs.
\cref{section:discrete} studies opacity in discrete guarded METAs.
\cref{section:DE} studies opacity with an observation of the energy level every time unit.
\cref{section:ISMETA} exhibits a preliminary decidable subclass for non-necessarily discrete models, \ie{} with multiple energy rates in addition to the discrete updates.
\cref{section:conclusion} concludes and highlights perspectives.
	\section{Preliminaries}\label{section:preliminaries}

	Let $\setN$,
	$\setZ$,
	$\setQplus$,
	$\setRgeqzero$,
	$\setR$
	denote the sets of non-negative integers, integers, non-negative rationals, non-negative reals, %
	and reals respectively.
	Let $\subR \subseteq \setR$.
	Given a finite variables set~$\Variables$, a $\subR$-valuation over~$\Variables$ is a function from~$\Variables$ to~$\subR$.

	Throughout this paper, we assume a set~$\Clock = \{ \clocki{1}, \dots, \clocki{\ClockCard} \} $ of \emph{clocks}, \ie{} non-negative real-valued variables that evolve at the same rate.
	A \emph{clock valuation} is an $\setRgeqzero$-valuation over~$\Clock$.
	
	We assume a set~$\Energies = \{ \enervar, \dots, \enervar_\EnergiesCard \} $ of \emph{energy variables}, \ie{} non-negative real-valued variables that evolve at potentially different rates.
	An \emph{energy valuation} is an $\setRgeqzero$-valuation over~$\Energies$.

	We write $\clockvalZero$ for the valuation such that for all $\clock \in \Clock$, $\clockvalZero(\clock) = 0$;
	similarly, $\enervalZero$ denotes the valuation s.t.\ for all $\enervar \in \Energies$, $\enervalZero(\enervar) = 0$.

	Given $\resets \subseteq \Clock$, we define the \emph{reset} of a clock valuation~$\clockval$, denoted by $\reset{\clockval}{\resets}$, as follows: $\reset{\clockval}{\resets}(\clock) = 0$ if $\clock \in \resets$, and $\reset{\clockval}{\resets}(\clock)=\clockval(\clock)$ otherwise.
	Given $d \in \setRgeqzero$, $\clockval + d$ denotes the clock valuation such that $(\clockval + d)(\clock) = \clockval(\clock) + d$, for all $\clock \in \Clock$.
	Similarly, given $d \in \setRgeqzero$, and a flow function $ \flow : \Energies \to \setZ$ assigning each variable with a rate (\ie{} the value of its derivative), we define the time elapsing function $\timelapseFlow$ as follows:
$\timelapseFlow(\enerval, \flow, d)$ is the valuation such that
\(\forall \enervar \in \Energies : \timelapseFlow(\enerval, \flow, d)(\variable) = \enerval(\enervar) + \flow(\enervar) \times d \).

	We denote by $\enerupdates$ an energy update function, which assigns to each $\enervar\in\Energies$ an offset in~$\setZ$.
	We define the \emph{update} of an energy valuation $\enerval$, denoted by $\reset{\enerval}{\enerupdates}$, as follows: $\reset{\enerval}{\enerupdates}(\enervar) = \enerval(\enervar) + \enerupdates (\enervar)$.

	In the following, we assume ${\bowtie} \in \{<, \leq, \geq, >\}$.
	A \emph{simple constraint}~$\guard$ is a constraint over $\Clock \cup \Energies$ defined by a conjunction of inequalities of the form $\variable \compop d$,
	with $\variable \in \Clock \cup \Energies$, %
	and~$d \in \setZ$.
	A simple clock constraint (resp.\ simple energy constraint) is a simple constraint over~$\Clock$ (resp.~$\Energies$).
	
	A clock valuation~$\clockval$ \emph{satisfies} a simple clock constraint~$\guard$, denoted by $\clockval \models \guard$, if replacing each $\clock \in \Clock$ with $\clockval(\clock)$ within~$\guard$ evaluates to true.
	We define similarly the satisfaction of simple energy constraints.
	Given a clock valuation~$\clockval$, an energy valuation~$\enerval$ and a simple constraint~$\guard$, we write $\wv{\clockval}{\enerval} \models \guard$, if replacing each $\clock \in \Clock$ with $\clockval(\clock)$ and each $\enervar \in \Energies$ with $\enerval(\enervar)$ within~$\guard$ evaluates to true.

	\subsection{Guarded multi-energy timed automata}
	We first define a very general formalism, extending TAs with multiple energy variables, and guards and invariants involving both energy and clock variables.
	Our formalism is seemingly more expressive than any formalism from the energy timed automata literature; we will restrict this formalism to various subclasses later.
	It can also be seen as a subclass of the highly undecidable formalism of \emph{hybrid automata}~\cite{ACHH92}.
	
	\begin{definition}\label{def:gMETA}
		A \emph{guarded multi-energy timed automaton (guarded META)} $\TA$ over a set~$\AtomicProp$ of atomic propositions is a tuple \mbox{$\TA = (\Actions, \LocSet, \locinit, \LocsPriv, \LocsFinal, \LabelFunc, \Clock, \Energies, \invariant, \rates, \Edges)$}, where:
		\begin{enumerate}%
			\item $\Actions$ is a finite set of actions,
			\item $\LocSet$ is a finite set of locations,
			 $\locinit \in \LocSet$ is the initial location,
			\item $\LocsPriv \subseteq \LocSet$ is the non-empty set of private locations,
			\item $\LocsFinal \subseteq \LocSet \setminus \LocsPriv$ is the non-empty set of final locations,
			\item $\LabelFunc$ is a label function $\LabelFunc : \LocSet \fleche{} 2^{\AtomicProp}$,
			\item $\Clock$ is a finite set of clocks,
			 $\Energies$ is a finite set of energy variables,
			\item $\invariant$ is the invariant, assigning to each $\loc \in \LocSet$ a simple constraint,
			\item $\rates$ is the energy rate function, assigning to every $\loc \in \LocSet$ a flow function assigning to each energy variable a rate in~$\setZ$,
			and
			\item $\Edges$ is a set of edges  $\edge = (\loc, \guard, \action, \resets, \enerupdates, \loc')$
			where
			$\loc,\loc'\in \LocSet$ are the source and target locations,
			$\guard$ (the ``guard'') is a simple constraint,
			$\action \in \Actions \cup \{ \silentaction \}$,
			$\resets\subseteq \Clock$ is a set of clocks to be reset to~$0$,
			and
			$\enerupdates$ is the energy update function.
		\end{enumerate}
	\end{definition}
	W.l.o.g., we assume that final locations do not have outgoing transitions.

	\begin{figure*}[tb]
			\centering
			\scalebox{.9}{
			\begin{tikzpicture}[PTA, node distance=1.25cm and 2.5cm]
				
				\node[location, initial, initial above] (standby) {\droneLocStandby{}};
				\node[location, left = of standby, align=center](cool) {\droneLocCooling{} \\ $\dot{\temperature} = -1$};
				
				\node[location, right = of standby, align=center] (fly) {\droneLocFlying{}\\$\dot{\temperature} = 2$\\$\dot{\battery} = -2$};
				\node[location, private, below = of fly] (shoot) {\droneLocFlashed{}};
				\node[locbad, below = of standby] (crash) {\droneLocExploded{}};
				\node[location, below = of cool, align=center] (charge) {\droneLocCharging{} \\ $\dot{\temperature} = 1$\\$\dot{\battery} = 2$};
				\node[location, final,  right = of fly] (end) {\droneLocSuccess{}};
				
				\node[locbad, above = of end] (oob) {\droneLocOutOfBattery{}};
				
				\node[invariantNorth] at (cool.north) {$\temperature \geq 0$};
				\node[invariantNorth] at (fly.north) {$\temperature \leq 100 \land \battery \geq 0$};
				\node[invariantWest, align=center, xshift=-1em] at (charge.west) {$\temperature \leq 100$ \\ $ \land \battery \leq 100$};
				\node[invariantSouth] at (shoot.south) {$\clock = 0$}; %

				\path (standby) edge[sloped] node[align=center] {\droneActCharge} (charge);
				\path (charge) edge node[left, align=center] {\droneActCool} (cool);
				\path (cool) edge[sloped] node[align=center]{$\temperature = 0$} node[below, align=center] {\droneActRest} (standby);
				
				\path (standby) edge node[above, align=center] {\droneActFly}(fly);
				\path (charge) edge[sloped] node[align=center] {\droneActFly} (fly);
				\path (cool) edge[bend left] node[above, align=center] {\droneActFly} (fly);
				
				\path (fly) edge[bend left] node[right, align=center, yshift=-0.5em] {$\battery \geq 5$\\ \droneActFlash \\ $\battery: -5$ ; $\temperature : +5$ \\ $\clock \assign 0$}(shoot);
				\path (shoot) edge[bend left] node[below left, align=center] {$\silentaction$} (fly);
				
				\path (charge) edge[sloped] node[below, align=center]{$\temperature \geq 100$ \\ \droneActExplode{}} (crash); %
				\path (fly) edge[sloped] node[below left, align=center]{$\temperature \geq 100$ \\ \droneActExplode{}} (crash); %
				\path (shoot) edge[bend left] node[below, align=center]{\droneActExplode{}}  node[above, align=center] {$\temperature \geq 100$} (crash);
				
				\path (fly) edge node[above, align=center] {$\battery > 0$} node[below, align=center] {\droneActLand} (end);
				
				\path (fly) edge[sloped]  node[align=center]{$\battery = 0$} node[below, align=center] {\droneActOff{}} (oob);
			\end{tikzpicture}
		}

		\caption{Example of a guarded-META: drone emitting light flashes}
		\label{figure:example-drone}
	\end{figure*}
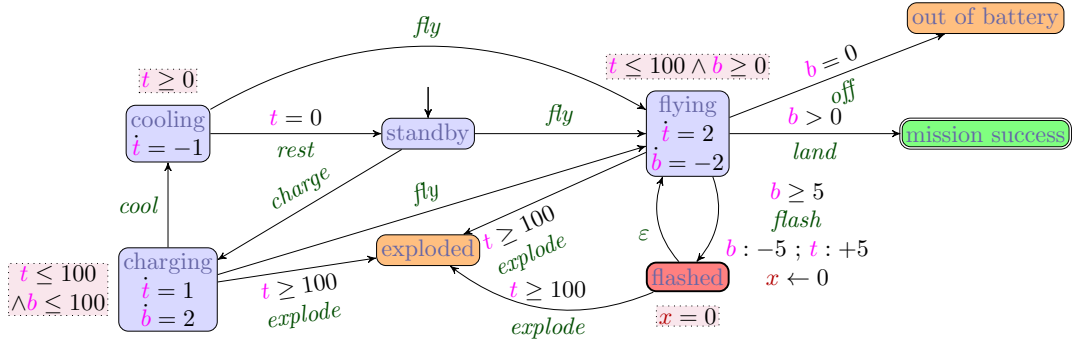
	
	\begin{example}
	\cref{figure:example-drone} models a fictitious drone that can charge and then fly (a single time).
	When it flies, it can emit light flashes (as in an entertainment drone show).
	The guarded META in \cref{figure:example-drone} has one clock~$\clock$ and two energy variables (initially~0): $\temperature$ denotes the temperature of the drone while $\battery$ denotes the battery percentage.
	Initially, the drone is on standby.
	When it starts charging (location ``\droneLocCharging''), both its battery and its temperature increase, as seen in the location rates (we use $\dot{\temperature} = 1$ for $\rates(\loc)(\temperature) = 1$); no rate depicted denotes a 0-rate.
	The drone can also cool down (location ``\droneLocCooling{}'').
	When it flies (location ``\droneLocFlying{}''), it consumes battery and the temperature rises continuously.
	Each time the drone emits a light flash, it consumes 5 battery units and gains 5 temperature units in 0-time, denoted by the transition from ``\droneLocFlying{}'' to ``\droneLocFlashed{}''.
	This latter urgent location (in which no time can elapse) serves to encode opacity: it is reachable iff the ``\droneActFlash{}'' action occurs.

	Energy guards are used to model that
	\begin{oneenumerate}%
		\item whenever the temperature goes higher than 100 degrees, the drone explodes (encoded by the transitions from most locations to ``\droneLocExploded{}''); and
		\item the battery level must remain between 0 and~100.
		If it reaches~0 during the drone's flight, the mission ends in location ``\droneLocOutOfBattery{}''.
	\end{oneenumerate}%
	In this model, we will be interested in deciding whether an attacker, by looking at time and energy (\eg{} at the end of the execution), can detect that the drone has emitted at least one light flash, \ie{} whether the ``\droneLocFlashed{}'' location was visited.
	\qed
	\end{example}

	Let us now define the concrete semantics of~guarded METAs; this is a straightforward extension of the semantics of timed automata~\cite{AD94} and multiple-resource energy timed automata (\eg{} \cite{CHL24}).

	\begin{definition}[Semantics of a guarded META]\label{definition:semanticsTA}
		Given a guarded META $\TA = (\Actions, \LocSet, \locinit, \LocsPriv, \LocsFinal, \LabelFunc, \Clock, \Energies, \invariant, \rates, \Edges)$,
		the \emph{concrete semantics} of $\TA$ is given by the timed transition system $(\Q, \qinit, \TransConcrete)$, with
		\begin{itemize}
			\item
			$\Q = \big\{ (\loc, \clockval, \enerval) \in \LocSet \times \setRgeqzero^{\cardinality{\Clock}} \times \setRgeqzero^{\cardinality{\Energies}} \mid \wv{\clockval}{\enerval} \models \invariant(\loc) \big\}$,
			\item $\qinit = (\locinit, \clockvalZero, \enervalZero) $,
			\item delay transitions:
			$\big((\loc, \clockval, \enerval), d, (\loc, \clockval+d, \timelapseFlow(\enerval, \rates(\loc), d)) \big) \in {\TransConcrete}$, with $d \in \setRgeqzero$, if $\forall d' \in [0, d], \big(\loc, \clockval+d', \timelapseFlow(\enerval, \rates(\loc), d') \big) \in \Q$;
			\item discrete transitions:
			$\big((\loc, \clockval, \enerval), \edge, (\loc',\clockval', \enerval') \big) \in {\TransConcrete}$ %
			if $(\loc, \clockval, \enerval) , (\loc',\clockval', \enerval') \in \Q$, there exists $\edge = (\loc, \guard, \action, \resets, \enerupdates, \loc') \in \Edges$, $\clockval'= \reset{\clockval}{\resets}$, $\enerval' = \reset{\enerval}{\enerupdates} $, and $\wv{\clockval}{\enerval} \models \guard$.
		\end{itemize}
	\end{definition}

	We assume that $\qinit \in \Q$, \ie{} $\wv{\clockvalZero}{\enervalZero} \models \invariant(\locinit)$, \ie{} the initial clock valuation satisfies the invariant of the initial location.
	
	We refer to the states of the concrete semantics of a guarded META as its \emph{concrete states}.
	As usual, given two concrete states $\concstate,\concstate'$, we write $\concstate\fleche{a}\concstate'$, for $a \in \setRgeqzero \cup \Edges$, instead of $(\concstate, a,\concstate') \in {\TransConcrete}$.
	We also further define relation $\TransConcreteEdge$ by $\big((\loc, \clockval, \enerval), (d, \edge), (\loc', \clockval', , \enerval')\big) \in \TransConcreteEdge$ if $\exists \clockval'' \in \setRgeqzero^{\cardinality{\Clock}},  \enerval'' \in \setRgeqzero^{\cardinality{\Energies}}, \edge \in \Edges, d \in \setRgeqzero:  (\loc, \clockval, \enerval) \fleche{d} (\loc,\clockval'', \enerval'') \fleche{\edge} (\loc', \clockval',  \enerval')$.
	A concrete run~$\varrun$ of a guarded META is an alternating sequence of concrete states of~$\Q$ and edges of the form
	$(\loc_0, \clockval_0, \enerval_0) \FlecheConcrete{(d_0, \edge_0)} (\loc_1, \clockval_1, \enerval_1) \FlecheConcrete {(d_1, \edge_1)} \cdots \FlecheConcrete{(d_{m-1}, \edge_{m-1})} (\loc_m, \clockval_m, \enerval_m)$,
	such that for all $i = 0, \dots, m-1$, $d_i \in \setRgeqzero$, $\edge_i \in \Edges$, and $\big((\loc_i, \clockval_i, \enerval_i) , (d_i, \edge_i) , (\loc_{i+1}, \clockval_{i+1}, \enerval_{i+1}) \big) \in \TransConcreteEdge$.
	We say that states $(\loc_i, \clockval_i, \enerval_i)$, for $i = 0, \dots, m$, \emph{belong} to~$\varrun$.
	The duration of such a run $\varrun$ is $\runduration{\varrun} = \sum_{0 \leq i \leq m-1} d_i$.

	The \emph{final energies} of~$\varrun$, denoted by $\finalEnergy(\varrun)$, is the energy valuation~$\enerval_m$.
	Given a concrete state~$\concstate = (\loc, \clockval, \enerval)$, we say that $\concstate$ is reachable (or that $\TA$ reaches~$\concstate$) if $\concstate$ belongs to a run of~$\TA$.
	By extension, a location~$\loc$ is reachable if there exists $\clockval$ and~$\enerval$ such that $(\loc, \clockval, \enerval)$ is reachable.
	We extend the notion of timed language of~TAs to guarded METAs in a straightforward manner, as the set of accepting timed word, \ie{} runs whose last state is a final location then projected onto actions and time durations, \ie{} keeping only the sequences $(\sum_{j=0}^{i}{d_j}, \action_i)$ for $0 \leq i \leq m-1$, where $\action_i$ is the action of edge~$\edge_i$, and removing $\silentaction$ transitions in an appropriate manner.
	Two guarded METAs are \emph{equivalent} if their timed language is the same.
	\begin{example}\label{example:run}
		\cref{figure:example-drone} accepts the following run, where $(\droneLocFlying{}, 15, [\vec{5,20}])$ is a shorthand for $(\droneLocStandby{}, \clockval, \enerval)$ such that $\clockval(\clock) = 15$, $\enerval(t) = 5$ and $\enerval(b) = 20$:
			$(\droneLocStandby{}, 0, [\vec{0,0}]) \FlecheConcrete{(0, \droneActCharge)} (\droneLocCharging, 0, [\vec{0,0}]) \FlecheConcrete{(10, \droneActCool)} (\droneLocCooling{}, 10, [\vec{10,20}]) \FlecheConcrete{(5, \droneActFly)} (\droneLocFlying{}, 15, [\vec{5,20}]) \FlecheConcrete{(2, \droneActFlash)} (\droneLocFlashed{} , 0, [\vec{14,11}]) \FlecheConcrete{(0, \silentaction)} (\droneLocFlying{}, 0, [\vec{14,11}]) \FlecheConcrete{(0, \droneActLand)} (\droneLocSuccess{}, 0, [\vec{14,11}])$.
		The duration of this run is $10 + 5 + 2 = 17$, and the final energy values are $14$ for the temperature and $11$ for the battery.
		The associated timed word is $(0, \droneActCharge) (10, \droneActCool) (15, \droneActFly) (17, \droneActFlash) (17, \droneActLand)$.
	\end{example}
	\subsubsection*{Subclasses}\label{sss:subclasses}
	A guarded META is \emph{positive} whenever $\forall \loc \in \LocSet, \forall \enervar \in \Energies : \rates(\loc)(\enervar) \in \setN$ and $\forall (\loc, \guard, \action, \resets, \enerupdates, \loc') \in \Edges,  \forall \enervar \in \Energies : \enerupdates(\enervar) \in \setN$.
	A guarded META is \emph{discrete} whenever $\forall (\loc, \guard, \action, \resets, \enerupdates, \loc') \in \Edges,  \forall \enervar \in \Energies : \enerupdates(\enervar) = 0$, \ie{} energy variables are only updated along discrete edges.
	A \emph{META}\footnote{%
			In~\cite{CHL24}, a formalism close to our METAs is called ``multi-variable energy~TA''.
		} is a guarded META without energy guards and invariants, \ie{} all guards and invariants are simple clock constraints.
	An \emph{energy timed automaton (ETA)}~\cite{CHL24} is a META with $|\Energies| = 1$.
	A \emph{timed automaton} (TA)~\cite{AD94} is a META such that $\Energies = \emptyset$.

	We assume familiarity with the region automaton of TAs~\cite{AD94} (see \cref{appendix:regions} for a formal definition).
	We denote by $\RegionAutomaton{\TA}$ the region automaton of a TA~$\TA$, which is a non-deterministic finite automaton~(NFA).
	Its (untimed) language is denoted by $\Language(\RegionAutomaton{\TA})$.

	\section{Defining opacity in guarded METAs}\label{section:problems}
	Given a guarded META~$\TA$ and a run~$\varrun$, we say that $\LocsPriv$ is \emph{visited on the way to a final location in~$\varrun$} when $\varrun$ is of the form~$(\loci{0}, \clockval_0, \enerval_0), (\paramd_0, \edgei{0}), (\loci{1}, \clockval_1, \enerval_1),  \ldots,  (\loci{m}, \clockval_m, \enerval_m), (\paramd_m, \edgei{m})$, $ \ldots,  (\loci{n}, \clockval_n, \enerval_n)$
	\noindent{}for some~$m,n \in \setN$ such that $\loci{m} \in \LocsPriv$ and $\loci{n} \in \LocsFinal$.
	We denote by $\PrivVisit{\TA}$ the set of those runs, and refer to them as \emph{private} runs.
	We denote by $\PrivDurVisit{\TA}$ the set of all the durations of these runs.
	We denote by $\PrivEnerVisit{\TA}$ the set of all the final energies of these runs.

	Conversely, we say that
	$\LocsPriv$ is \emph{avoided on the way to a final location in~$\varrun$}
	when $\varrun$ is of the form
	\((\loci{0}, \clockval_0, \enerval_0), (\paramd_0, \edgei{0}), (\loci{1}, \clockval_1, \enerval_1), \ldots, (\loci{n}, \clockval_n, \enerval_n )\)
	\noindent{}with $\loci{n} \in \LocsFinal$ and $\forall 0 \leq i \leq n, \loci{i} \notin \LocsPriv$.
	We denote the set of those runs by~$\PubVisit{\TA}$, referring to them as \emph{public} runs,
	by $\PubDurVisit{\TA}$ the set of all the durations of these public runs,
	and
	by $\PubEnerVisit{\TA}$ the set of all the final energies of these public runs.

	$\PrivDurVisit{\TA}$ and $\PubDurVisit{\TA}$ can be seen as the set of execution times from the initial location~$\locinit$ to a final location while
	visiting (resp.\ not visiting) the private locations~$\LocsPriv$.

	We extend to guarded METAs the notion of ET-opacity~\cite{ALLMS23,ADLL24}.

	\begin{definition}[ET-opacity]\label{def:full-weak-opacity}
		A guarded META~$\TA$ is \emph{fully ET-opaque} when $\PrivDurVisit{\TA} = \PubDurVisit{\TA}$.
		It is \emph{weakly ET-opaque} when $\PrivDurVisit{\TA} \subseteq \PubDurVisit{\TA}$.
		It is \emph{$\exists$-ET-opaque} when $\PrivDurVisit{\TA} \cap \PubDurVisit{\TA} \neq \emptyset$.
	\end{definition}

	That is, if for any run of duration~$\paramd$ reaching a final location after visiting~$\LocsPriv$, there exists another run of the same duration reaching a final location but not visiting~$\LocsPriv$, and vice versa, then the TA is fully ET-opaque.
	Weak ET-opacity requires the existence of a public run only if there exists a private run of same duration.
	Finally, when there exist at least one private and one public run of the same duration, the system is $\exists$-ET-opaque.

	We now define a new notion of opacity \wrt{} the final energy.

	\begin{definition}[EN-opacity]\label{def:full-energy-opacity}
		A guarded META~$\TA$ is fully opaque \wrt{} energy (\emph{fully EN-opaque}) when $\PrivEnerVisit{\TA} = \PubEnerVisit{\TA}$.
		It is weakly opaque \wrt{} energy (\emph{weakly EN-opaque}) when $\PrivEnerVisit{\TA} \subseteq \PubEnerVisit{\TA}$.
		It is $\exists$-opaque \wrt{} energy (\emph{$\exists$-EN-opaque}) when $\PrivEnerVisit{\TA} \cap \PubEnerVisit{\TA} \neq \emptyset$.
	\end{definition}

	We define similarly \emph{ET-EN-opacity} when the runs meet both execution time and energy conditions (see \cref{def:ET-EN-opacity} in \cref{appendix:definitions:ETEN} for a formal definition).

	\subsection{Problems}

	For each $\existsWeakFull \in \{ \exists, \text{weak}, \text{full} \}$ and each $\ENorETEN \in \{ \text{EN}, \text{ET-EN} \}$, we define:

	\defProblem{$\existsWeakFull$-$\ENorETEN$-opacity}
		{A guarded META~$\TA$}
		{Is $\TA$ $\existsWeakFull$-$\ENorETEN$-opaque?}
	\begin{example}\label{example:opacity}
		Let us check some opacity problems over the guarded META in \cref{figure:example-drone}.
		The private run in \cref{example:run} gives final energy values of $[\vec{14,11}]$.
		There exists a public run with these same final values:
		$(\droneLocStandby{}, 0, [\vec{0,0}]) \FlecheConcrete{(0, \droneActCharge)} (\droneLocCharging, 0, [\vec{0,0}]) \FlecheConcrete{(10, \droneActCool)} (\droneLocCooling{}, 10, [\vec{10,20}]) \FlecheConcrete{(5, \droneActFly{})} (\droneLocFlying{}, 15, [\vec{5,20}]) \FlecheConcrete{(4.5, \droneActLand{})} (\droneLocSuccess{}, 19.5, [\vec{14,11}])$.
		Therefore, the guarded META is $\exists$-EN-opaque.
		Actually, it is also weakly-EN-opaque: one can always wait 2.5 time units in the ``\droneLocFlying{}'' location instead of flashing, as both actions have the same energy consumption.
		However, we can easily check that it is not fully-ET-opaque, as the private runs require at least 2.5 time unit of charging.
	\end{example}
	\section{Generic ingredients and proofs}\label{section:generic}

	Given a guarded META~$\TA$, we define two further guarded METAs encoding all public runs (resp.\ all private runs) of~$\TA$.
	The \emph{public runs guarded META}~$\Apub{\TA}$ is obtained very easily from~$\TA$ by simply removing the private locations~$\LocsPriv$.
	Therefore, all accepting runs are public, as they do not visit~$\LocsPriv$.

	Inspired by techniques from~\cite{ADL24}, given a guarded META~$\TA$, the \emph{private runs guarded META}~$\Apriv{\TA}$ is obtained by duplicating all locations and transitions of~$\TA$: one copy $\TA^{V}$ corresponds to the paths that already visited a private location, while the other copy $\TA^{\bar{V}}$ corresponds to the paths that did not (this is a usual way to encode a Boolean, here ``$\LocsPriv$ was visited'', in the locations of a guarded~META).
	That is, for each location $\loc$ of~$\TA$, one location~$\loc^{V}$ is defined in~$\TA^{V}$,
	and one location~$\loc^{\bar{V}}$ is defined in~$\TA^{\bar{V}}$.
	Then, we redirect all transitions leading to a location~$\locpriv^{\bar{V}}$ to the copy $\locpriv^{V}$ from~$\TA^{V}$.
	The initial location is~$\locinit^{\bar{V}}$ and the final locations are $\locfinal^{V}$.
	Hence all runs need to go from $\TA^{\bar{V}}$ to $\TA^{V}$ before reaching a final location---which requires visiting a private location.
	See \cref{definition:Apub,definition:Apriv} in \cref{appendix:def:ApubApriv} for a formal definition.

	\begin{example}
		\cref{figure:example-METApubpriv} gives an example of a guarded META~$\TA$ with its public and private runs guarded METAs.
	\end{example}
		\begin{figure}[tb]
		\centering

			\begin{subfigure}[b]{\linewidth}
				\centering
				\scalebox{.9}{
				\begin{tikzpicture}[PTA, node distance=.75cm and 2cm]
					\node[location, initial] (l0) {$\locinit$};
					\node[location, right = of l0] (l1) {$\loc_1$};
					\node[location, private, below = of l1] (lp) {$\locpriv$};
					\node[location, final, right = of l1] (lF) {$\locfinal$};
					
					\node[invariantNorth] at (l0.north) {$\clock \leq 3$};
					\node[invariantNorth] at (lp.north) {$\enervar \leq 5$};
					\node[invariantNorth] at (l1.north) {$\clock \leq 2$};
					
					\path (l0) edge[bend right] node[ below left, align=center]{$\styleact{a}$ \\ $\styleenergy{\enervar:+1}$ \\ $\clock \assign 0$} (lp);
					
					\path (l0) edge node[above, align=center]{$\clock > 1$ \\ $\styleact{c}$} node[below, align=center]{$\clock \assign 0$} (l1);
					
					\path (l1) edge node[above, align=center]{$\clock > 1$ \\ $\styleact{d}$} node[below, align=center]{$\styleenergy{\enervar:+2}$} (lF);
					
					\path (lp) edge[loop below] node[right, align=center, xshift=.5em, yshift=-.5em]{$\clock \leq 1$ \\ $\styleact{b}$ \\ $\styleenergy{\enervar:+2}$ \\ $\clock \assign 0$} (lp);
					
					\path (lp) edge[bend right] node[below right, align=center]{$\clock > 1$ \\ $\styleact{b}$} (lF);
					
				\end{tikzpicture}
				}
				
				\caption{\TA}
				\label{figure:example-META-A}
			\end{subfigure}

			\begin{subfigure}[b]{\linewidth}
				\centering
				\scalebox{.9}{
				\begin{tikzpicture}[PTA, node distance=1cm and 1.5cm]
					\node[location, initial] (l0) {$\locinit$};
					\node[location, right = of l0] (l1) {$\loc_1$};
					\node[location, final, right = of l1] (lF) {$\locfinal$};
					
					\node[invariantNorth] at (l0.north) {$\clock \leq 3$};
					\node[invariantSouth] at (l1.south) {$\clock \leq 2$};

					\path (l0) edge node[above, align=center]{$\clock > 1$ \\ $\styleact{c}$} node[below, align=center]{$\clock \assign 0$} (l1);
					
					\path (l1) edge node[above, align=center]{$\clock > 1$ \\ $\styleact{d}$} node[below, align=center]{$\styleenergy{\enervar:+2}$} (lF);
					
				\end{tikzpicture}
				}
				
				\caption{$\Apub{\TA}$}
				\label{figure:example-META-Apub}
			\end{subfigure}
			\begin{subfigure}[b]{\linewidth}
				\centering
				\scalebox{.85}{
				\begin{tikzpicture}[PTA, node distance=.75cm and 2cm]
					\node[location, initial] (l0) {$\locinit^{\bar{V}}$};
					\node[location, right = of l0] (l1) {$\loc_1^{\bar{V}}$};
					\node[locunreachable, below = of l1] (lp) {$\locpriv^{\bar{V}}$};
					\node[location, right = of l1] (lF) {$\locfinal^{\bar{V}}$};

					\node[invariantNorth] at (l0.north) {$\clock \leq 3$};
					\node[invariantNorth] at (lp.north) {$\enervar \leq 5$};
					\node[invariantNorth] at (l1.north) {$\clock \leq 2$};

					\node[locunreachable, below of=l0, yshift=-8em] (l0') {$\locinit^{V}$};
					\node[locunreachable, right = of l0'] (l1') {$\loc_1^{V}$};
					\node[location, private, below = of l1'] (lp') {$\locpriv^{V}$};
					\node[location, final, right = of l1'] (lF') {$\locfinal^{V}$};
					\node[invariantNorth] at (l0'.north) {$\clock \leq 3$};
					\node[invariantNorth] at (lp'.north) {$\enervar \leq 5$};
					\node[invariantNorth] at (l1'.north) {$\clock \leq 2$};

					\path (l0) edge[bend right, thick, black!50!red] node[above left, xshift=-1em, yshift=2em, align=center]{$\styleact{a}$ \\ $\styleenergy{\enervar:+1}$ \\ $\clock \assign 0$} (lp');

					\path (l0) edge node[above, align=center]{$\clock > 1$ \\ $\styleact{c}$} node[below, align=center]{$\clock \assign 0$} (l1);

					\path (l1) edge node[above, align=center]{$\clock > 1$ \\ $\styleact{d}$} node[below, align=center]{$\styleenergy{\enervar:+2}$} (lF);

					\path (lp) edge[loop left] node[left, align=center]{$\clock \leq 1$ \\ $\styleact{b}$ \\ $\styleenergy{\enervar:+2}$ \\ $\clock \assign 0$} (lp);

					\path (lp) edge[bend right] node[below right, align=center]{$\clock > 1$ \\ $\styleact{b}$} (lF);

					\path (l0') edge[bend right] node[ below left, align=center]{$\styleact{a}$ \\ $\styleenergy{\enervar:+1}$ \\ $\clock \assign 0$} (lp');

					\path (l0') edge node[above, align=center]{$\clock > 1$ \\ $\styleact{c}$} node[below, align=center]{$\clock \assign 0$} (l1');

					\path (l1') edge node[above, align=center]{$\clock > 1$ \\ $\styleact{d}$} node[below, align=center]{$\styleenergy{\enervar:+2}$} (lF');

					\path (lp') edge[loop below] node[right, align=center, xshift=1em]{$\clock \leq 1$ \\ $\styleact{b}$ \\ $\styleenergy{\enervar:+2}$ \\ $\clock \assign 0$} (lp');

					\path (lp') edge[bend right] node[below right, align=center]{$\clock > 1$ \\ $\styleact{b}$} (lF');

					\end{tikzpicture}
					}

				\caption{$\Apriv{\TA}$}
				\label{figure:example-META-Apriv}
			\end{subfigure}

			\caption{Public and private runs guarded METAs}
			\label{figure:example-METApubpriv}
		
		\end{figure}
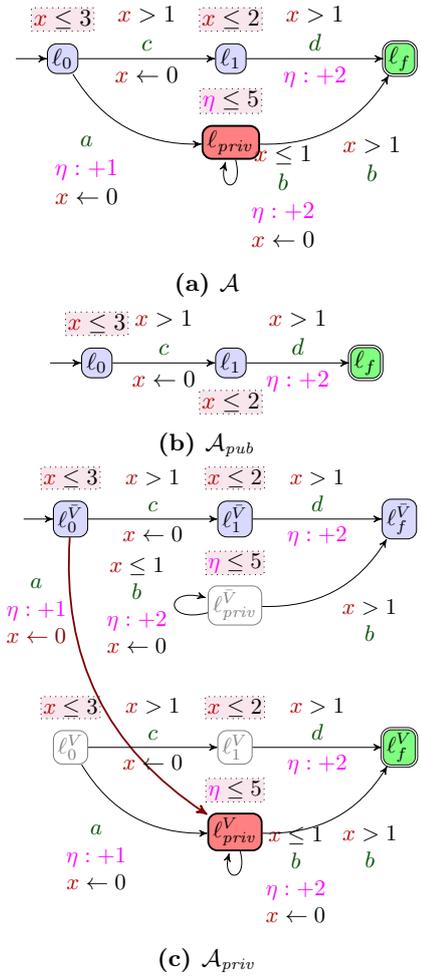

	In some cases, a positive guarded META can be transformed into an equivalent META (without energy guards).

	\begin{restatable}{proposition}{propRemovingGuardsDiscrete}\label{proposition:removing-guards-discrete}
		Given be a discrete positive guarded META~$\TA$,
		$\exists$ a discrete positive META exponentially larger than~$\TA$ equivalent to~$\TA$.
	\end{restatable}
	\begin{proof}[Proof sketch]
		From the fact that, in a discrete guarded META, the value of energy can only increase; so once it becomes $> \maxConstantEnergy$ (with $\maxConstantEnergy$ denoting the maximum constant in energy guards), its value is not relevant anymore.
		In an ETA, considering $\maxConstantEnergy+2$ copies of the automaton ($0$ to~$\maxConstantEnergy$ encoding the exact current energy valuation while the last one encodes $> \maxConstantEnergy$) correctly encodes the energy; in METAs, this construction leads to an exponential blowup.
		See \cref{appendix:proof:proposition:removing-guards-discrete}.
	\end{proof}
	\section{Opacity problems for discrete METAs}\label{section:discrete}

	In this section, we consider \emph{discrete} (guarded) METAs, hence in which the evolution of the energy is limited to discrete increments or decrements---while the model remains real-time.

	We first rule out guarded METAs in general, because a discrete guarded META with only 2~energies and 1~clock is sufficiently expressive to simulate a 2-counter machine.

	\begin{restatable}{proposition}{undecTwoEnergies}\label{proposition:undecidability-2-energies}
		For each $\existsWeakFull \in \{ \exists, \text{weak}, \text{full} \}$,
		and $\ENorETEN \in \{ \text{EN}, \text{ET-EN} \}$,
		$\existsWeakFull$-$\ENorETEN$-opacity for discrete guarded METAs with 2~energy variables and 1~clock is undecidable.
	\end{restatable}
	\begin{proof}[Proof sketch]
		We first prove that (constant-time) reachability is undecidable in discrete guarded METAs with only 2~energies and a single clock (always~0) (\cref{lemma:undecidability-META-2-1} in \cref{appendix:lemma:undecidability-META-2-1}).
		We then prove that for each $\existsWeakFull \in \{ \exists, \text{weak}, \text{full} \}$, and $\ENorETEN \in \{ \text{EN}, \text{ET-EN} \}$, $\existsWeakFull$-$\ENorETEN$-opacity is harder than reachability
		(see \cref{appendix:hardness:opacity-reachability-META}).
	\end{proof}

	Thus, we restrict ourselves to a single energy or to only positive rate energies.

	\subsection{EN-opacity in discrete positive (guarded) METAs}

	Our first setting is whenever the attacker observes only the final energy, without hint on the actual execution time.

	\begin{figure*}[tb]
		\begin{subfigure}[b]{.45\columnwidth}
		\centering
			\scalebox{.8}{
			\begin{tikzpicture}[PTA, node distance=.8cm and 1.4cm]
				\node[location, initial] (l0) {$\locinit$};
				\node[location, private, below =of l0] (lpriv) {$\locpriv$};
				\node[location, final, right =of l0] (lF) {$\locfinal$};

				\node[invariantNorth] at (l0.north) {$\clock \leq 3$};
				\node[invariantSouth] at (lpriv.south) {$\clock \leq 3$};

				\path (l0) edge node[right, align=center]{$\styleact{a}$} (lpriv);
				\path (l0) edge node[above, align=center]{$\clock \geq 1$ \\ $\styleact{c}$} node[below, align=center]{$\styleenergy{\enervar_1 : +2}$} (lF);
				\path (lpriv) edge[loop left] node[above, align=center]{%
					$\styleact{b}$ \\ $\styleenergy{\enervar_1 : +1}$} (lpriv);
				\path (lpriv) edge[bend right] node[below right, align=center]{$\clock \geq 1$ \\ $\styleact{b}$} (lF);
			\end{tikzpicture}
			}

		\caption{A META~$\TA$}
		\label{figure:example-DiscPosGuarETA}
		\end{subfigure}
		\begin{subfigure}[b]{.45\columnwidth}
		\centering
			\scalebox{.8}{
			\begin{tikzpicture}[PTA, node distance=.8cm and 1.4cm]
				\node[location, initial] (l0) {$\locinit$};
				\node[location, private, below = of l0] (lpriv) {$\locpriv$};
				\node[location, urgent, right = of l0] (l0') {$\locinit'$};
				\node[location, final, right = of l0'] (lF) {$\locfinal$};

				\node[invariantNorth] at (l0.north) {$\clock \leq 3$};
				\node[invariantSouth] at (lpriv.south) {$\clock \leq 3$};

				\path (l0) edge node[right, align=center]{$\styleact{a}$} (lpriv);
				\path (l0) edge node[above, align=center]{$\clock \geq 1$ \\ $\styleact{c}$} node[below, align=center]{$\styleenergy{\enervar_1 : +1}$ \\ $\clockZeroTime \assign 0$} (l0');
				\path (l0') edge node[above, align=center]{$\clockZeroTime = 0$ \\ $\styleact{c}$} node[below, align=center]{$\styleenergy{\enervar_1 : +1}$} (lF);
				\path (lpriv) edge[loop left] node[above, align=center]{%
					$\styleact{b}$ \\ $\styleenergy{\enervar_1 : +1}$} (lpriv);
				\path (lpriv) edge[bend right] node[below right, align=center]{$\clock \geq 1$ \\ $\styleact{b}$} (lF);
			\end{tikzpicture}
			}

		\caption{Removing $+\updateConstant$ updates}
		\label{figure:example-DiscPosGuarETA-N}
		\end{subfigure}
		\begin{subfigure}[b]{.35\textwidth}
		\centering
			\scalebox{.8}{
			\begin{tikzpicture}[PTA, node distance=.8cm and 1.4cm]
				\node[location, initial] (l0) {$\locinit$};
				\node[location, private, below = of l0] (lpriv) {$\locpriv$};
				\node[location, urgent, right = of l0] (l0') {$\locinit'$};
				\node[location, final, right = of l0'] (lF) {$\locfinal$};

				\node[invariantNorth] at (l0.north) {$\clock \leq 3$};
				\node[invariantSouth] at (lpriv.south) {$\clock \leq 3$};

				\path (l0) edge node[right, align=center]{$\silentaction$} (lpriv);
				\path (l0) edge node[above, align=center]{$\clock \geq 1$ \\ $\styleact{\tickIncr_1}$} node[below, align=center]{%
					$\clockZeroTime \assign 0$} (l0');
				\path (l0') edge node[above, align=center]{$\clockZeroTime = 0$ \\ $\styleact{\tickIncr_1}$} %
				(lF);
				\path (lpriv) edge[loop left] node[above, align=center]{%
					$\styleact{\tickIncr_1}$ %
					} (lpriv);
				\path (lpriv) edge[bend right] node[below right, align=center]{$\clock \geq 1$ \\ $\silentaction$} (lF);
			\end{tikzpicture}
			}

		\caption{Adding energy ticks}
		\label{figure:example-DiscPosGuarETA-N-tic}
		\end{subfigure}
		\caption{Example and transformations in the proof of \cref{theorem:fullweak-discrete-positive-guarded-ETAs}}
	\end{figure*}
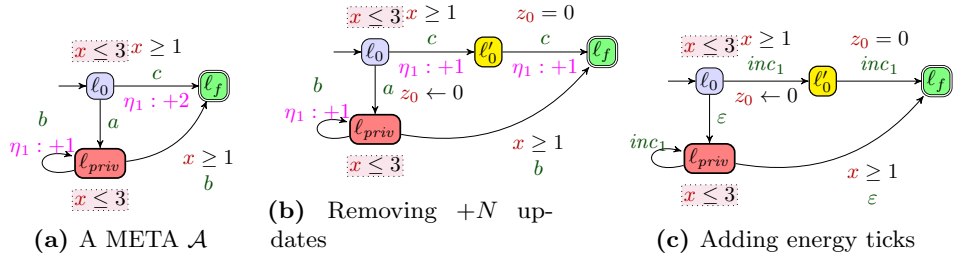
	\begin{theorem}\label{theorem:fullweak-discrete-positive-guarded-ETAs}
		For each $\existsWeakFull \in \{ \exists, \text{weak}, \text{full} \}$, $\existsWeakFull$-EN-opacity can be decided in \twoEXPSPACE{} for discrete positive METAs, and \threeEXPSPACE{} for discrete positive guarded METAs.
	\end{theorem}
	\begin{proof}
		Let us first consider discrete positive METAs without energy guards.
		We first transform the META into an equivalent model with only energy updates such that $\enerupdates(\loc)(\enervari{i}) = 1$ for each energy $\enervari{i} \in \Energies$:
		this can be achieved by duplicating any edge such that $\enerupdates(\loc)(\enervari{i}) = \updateConstant$ (with $\updateConstant > 1$) into $\updateConstant$ consecutive edges in 0-time, thanks to a new clock~$\clockZeroTime$.
		Then, we relabel actions as follows: any edge such that $\enerupdates(\loc)(\enervari{i}) = 1$ is labelled by action~$\tickIncr_i$, while any other edge is labelled with~$\silentaction$.
		We thus obtain a $|\Energies|$-action non-deterministic TA~$\TA'$.
		See \cref{figure:example-DiscPosGuarETA,figure:example-DiscPosGuarETA-N,figure:example-DiscPosGuarETA-N-tic} for an example.

		We then build $\Apub{\TA'}$ and $\Apriv{\TA'}$.
		Note that, up to action relabelling and edge duplication, all runs of $\Apub{\TA'}$ (resp.\ $\Apriv{\TA'}$) correspond to the public (resp.\ private) runs of~$\TA$, and the number of $\tickIncr$ actions along a run directly encodes the amount of energy spent.

		Our third step is to build the region automaton of these~TAs, \ie{} $\RegionAutomaton{\Apub{\TA'}}$ and $\RegionAutomaton{\Apriv{\TA'}}$.
		These two region automata are non-deterministic finite automata over $\{ \tickIncr_i, 1 \leq i \leq |\Energies| \}$.

		For example, given the private run~$\varrun_1 = (\locinit, 0, 0) \FlecheConcrete{(0.8, \styleact{a})} (\locpriv, 0.8, 0) \FlecheConcrete{(0.3, \styleact{b})} (\locpriv, 1.1, 1) \FlecheConcrete{(0.5, \styleact{b})} (\locpriv, 1.6, 2) \FlecheConcrete{(0.7, \styleact{b})} (\locfinal, 2.3, 2) $
		of~$\TA$ in \cref{figure:example-DiscPosGuarETA},
		its counterpart in the TA in \cref{figure:example-DiscPosGuarETA-N-tic} is the run
		$(\locinit, 0) \FlecheConcrete{(0.8, \silentaction)} (\locpriv, 0.8) \FlecheConcrete{(0.3, \styleact{\tickIncr_1})} (\locpriv, 1.1) \FlecheConcrete{(0.5, \styleact{\tickIncr_1})} (\locpriv, 1.6) \FlecheConcrete{(0.7, \silentaction)} (\locfinal, 2.3) $.
		This run gives in $\RegionAutomaton{\Apriv{\TA'}}$ the untimed word $\styleact{\tickIncr_1} \styleact{\tickIncr_1}$, correctly encoding the final value of~$\enervari{1}$.
		The public run $\varrun_2 = (\locinit, 0, 0) \FlecheConcrete{(2.3, \styleact{c})} (\locfinal, 2.3, 2) $ also gives $\styleact{\tickIncr_1} \styleact{\tickIncr_1}$.

		Verifying EN-opacity amounts to comparing $\Language(\RegionAutomaton{\Apriv{\TA'}})$ and $\Language(\RegionAutomaton{\Apub{\TA'}})$, by only focusing on the \emph{number} of symbols $\tickIncr_i$ in each word, disregarding their \emph{order}.
		To do so, we use the Parikh image~\cite{Parikh66} of these languages (which can be represented using a semilinear set).

		Verifying $\exists$-EN-opacity amounts to checking the emptiness of the intersection of these Parikh images.
		Checking intersection of the Parikh image of two regular languages is polynomial~\cite{KT10}.
		As the region automaton is exponential in the size of~$\Apub{\TA'}$ and~$\Apriv{\TA'}$, themselves linear in~$\TA$, this gives a \twoEXPSPACE{} procedure, since the subset construction of the region automata (required by the determinisation) adds another exponential.

		Verifying weak-EN-opacity amounts to checking the inclusion of the Parikh images of $\Language(\RegionAutomaton{\Apriv{\TA'}})$ and $\Language(\RegionAutomaton{\Apub{\TA'}})$.
		This is \coNP{}-complete~\cite{KT10}, again yielding a \twoEXPSPACE{} procedure.
		A similar reasoning holds for full-EN-opacity.

		Finally, the case of discrete positive \emph{guarded} METAs follows from \cref{proposition:removing-guards-discrete}, adding another exponential.
	\end{proof}
	\begin{example}
		Let $\TA$ be the discrete META given in \cref{figure:example-DiscPosGuarETA-N-tic}.
		After transforming the updates into $\tickIncr_i$ actions and constructing $\Apriv{\TA'}$ and~$\Apub{\TA'}$, we construct their region automata, see \cref{figure:example-ngMETA-Apub,figure:example-ngMETA-Apriv} (we omit some states and transitions due to space constraints).
		As the maximum constant for $\clock$ is~$3$, we have 8 clock regions:
		$\regioni{\{0\}}$, $\regioni{(0,1)}$, \dots, $\regioni{\{3\}}$, $\regioni{(3,\infty)}$, where $\region_i$ denotes the clock region such that $\clock \in i$.
		We can then test the emptiness of the intersection of the Parikh images of these two NFAs.
		In our example, we can directly see that the Parikh image of $\Language(\RegionAutomaton{\Apub{\TA'}})$ is $\{ 2 \}$ while that of $\Language(\RegionAutomaton{\Apriv{\TA'}})$ is~$\setN$; this corresponds indeed to $\PubEnerVisit{\TA}$ and~$\PrivEnerVisit{\TA}$ respectively.
		Therefore, $\TA$ is $\exists$-EN-opaque.
		However, it is neither weakly nor fully-EN-opaque, as we do not have inclusion.
	\end{example}

	\begin{figure}[tb]

			\begin{subfigure}[b]{\linewidth}
			\centering
			\scalebox{.8}{
			\begin{tikzpicture}[NFA, node distance=.3cm and .9cm]
				\node[state, initial] (l0s0') {$\locinit,\regioni{\{0\}}$};

				\node[state, right = of l0s0', xshift=-2em] (l0s1') {$\locinit,\regioni{(0,1)}$};

				\node[state, right = of l0s1', xshift=-2em] (l0s2') {$\locinit,\regioni{\{1\}}$};
				\node[state, right = of l0s2'] (l1s2') {$\loc_1,\regioni{\{1\}}$};
				\node[state, final, right = of l1s2'] (lFs2') {$\locfinal,\regioni{\{1\}}$};
				
				\node[state, below = of l0s2', yshift=-1.5em] (l0s6') {$\locinit,\regioni{\{3\}}$};
				\node[state, right = of l0s6'] (l1s6') {$\loc_1,\regioni{\{3\}}$};
				\node[state, final, right = of l1s6'] (lFs6') {$\locfinal,\regioni{\{3\}}$};

				\node[state, final, right = of lFs6', xshift=-2em] (lFs3') {$\locfinal,\regioni{(3,\infty)}$};

				\path (l0s2') edge node[above, align=center]{$\styleact{\tickIncr_1}$} (l1s2');
				\path (l1s2') edge node[above, align=center]{$\styleact{\tickIncr_1}$} (lFs2');
				
				\path (l0s6') edge node[above, align=center]{$\styleact{\tickIncr_1}$} (l1s6');
				\path (l1s6') edge node[above, align=center]{$\styleact{\tickIncr_1}$} (lFs6');

				\draw[-,line width=1pt, dotted, dash pattern=on 0pt off 4pt, line cap=round] ([yshift=-5pt]l1s2'.south) -- ++(0,-0.5);
				
				\path (l0s0') edge node[above, align=center]{$\silentaction$} (l0s1');
				\path (l0s1') edge node[above, align=center]{$\silentaction$} (l0s2');
				\path ([yshift=-1.75em]l0s2'.south) edge node[left, align=center]{$\silentaction$} (l0s6');
				\path (l0s2') edge node[left, align=center]{$\silentaction$} ([yshift=1.75em]l0s6'.north);

				\path ([yshift=-1.75em]lFs2'.south) edge node[right, align=center]{$\silentaction$} (lFs6');
				\path (lFs2') edge node[right, align=center]{$\silentaction$} ([yshift=1.75em]lFs6'.north);
				\path (lFs6') edge node[above, align=center]{$\silentaction$} (lFs3');

			\end{tikzpicture}
			}
			
			\caption{$\RegionAutomaton{\Apub{\TA'}}$}
			\label{figure:example-ngMETA-Apub}
		\end{subfigure}
		\begin{subfigure}[b]{\linewidth}
			\centering
			\scalebox{.8}{
			\begin{tikzpicture}[NFA, node distance=.4cm and .9cm]
				\node[state, initial] (l0s0) {$\locinit^{\bar{V}},\regioni{\{0\}}$};
				\node[state, private, below = of l0s0] (lps0) {$\locpriv^V,\regioni{\{0\}}$};

				\node[state,  right = of l0s0, xshift=-1em] (l0s1) {$\locinit^{\bar{V}},\regioni{(0,1)}$};
				\node[state, private, below = of l0s1] (lps1) {$\locpriv^V,\regioni{(0,1)}$};

				\node[state, right = of l0s1, xshift=0em] (l0s2) {$\locinit^{\bar{V}},\regioni{\{1\}}$};
				\node[state, right = of l0s2] (l1s2) {$\loc_1^{\bar{V}},\regioni{\{1\}}$};
				\node[state, private, below = of l0s2] (lps2) {$\locpriv^V,\regioni{\{1\}}$};
				\node[state, right = of l1s2] (lFs2) {$\locfinal^{\bar{V}},\regioni{\{1\}}$};
				\node[state, final, below = of lFs2] (lFps2) {$\locfinal^V,\regioni{\{1\}}$};
				
				\node[state, below = of l0s2, yshift=-3em] (l0s3) {$\locinit^{\bar{V}},\regioni{(1,2)}$};
				\node[state, below = of l1s2, yshift=-3em] (l1s3) {$\loc_1^{\bar{V}},\regioni{(1,2)}$};
				\node[state, private, below = of l0s3] (lps3) {$\locpriv^V,\regioni{(1,2)}$};
				\node[state, below = of lFps2, yshift=0.4em] (lFs3) {$\locfinal^{\bar{V}},\regioni{(1,2)}$};
				\node[state, final, below = of lFs3] (lFps3) {$\locfinal^V,\regioni{(1,2)}$};
				
				\node[state, below = of lps3, yshift=-1.5em] (l0s6) {$\locinit^{\bar{V}},\regioni{\{3\}}$};
				\node[state, right = of l0s6] (l1s6) {$\loc_1^{\bar{V}},\regioni{\{3\}}$};
				\node[state, private, below = of l0s6, align=center] (lps6) {$\locpriv^V,\regioni{\{3\}}$};
				\node[state, right = of l1s6] (lFs6) {$\locfinal^{\bar{V}},\regioni{\{3\}}$};
				\node[state, final, below = of lFs6] (lFps6) {$\locfinal^V,\regioni{\{3\}}$};

				\path (l0s0) edge node[right, align=center]{$\silentaction$} (lps0);
				\path (lps0) edge[loop below] node[left, align=center]{ $\styleact{\tickIncr_1}$ } (lps0);

				\path (l0s1) edge node[right, align=center]{$\silentaction$} (lps1);
				\path (lps1) edge[loop below] node[left, align=center]{ $\styleact{\tickIncr_1}$ } (lps1);

				\path (l0s2) edge node[right, align=center]{$\silentaction$} (lps2);
				\path (l0s2) edge node[above, align=center]{$\styleact{\tickIncr_1}$} (l1s2);
				\path (l1s2) edge node[above, align=center]{$\styleact{\tickIncr_1}$} (lFs2);
				\path (lps2) edge[loop left, looseness=3] node[below, yshift = -0.5em, align=center]{ $\styleact{\tickIncr_1}$ } (lps2);
				\path (lps2) edge node[align=center]{ $\silentaction$} (lFps2);

				\path (l0s3) edge node[right, align=center]{$\silentaction$} (lps3);
				\path (l0s3) edge node[above, align=center]{$\styleact{\tickIncr_1}$} (l1s3);
				\path (lps3) edge[loop left, looseness=3] node[left, align=center]{ $\styleact{\tickIncr_1}$ } (lps3);
				\path (lps3) edge node[align=center]{ $\silentaction$} (lFps3);
				\path (l1s3) edge node[above, align=center]{$\styleact{\tickIncr_1}$} (lFs3);
				
				\path (l0s6) edge node[right, align=center]{$\silentaction$} (lps6);
				\path (l0s6) edge node[above, align=center]{$\styleact{\tickIncr_1}$} (l1s6);
				\path (lps6) edge[loop left, looseness=3] node[left, align=center]{ $\styleact{\tickIncr_1}$ } (lps6);
				\path (lps6) edge node[align=center]{ $\silentaction$} (lFps6);
				\path (l1s6) edge node[above, align=center]{$\styleact{\tickIncr_1}$} (lFs6);

				\path (l0s0) edge node[above, align=center]{$\silentaction$} (l0s1);
				\path (l0s1) edge node[above, align=center]{$\silentaction$} (l0s2);
				\path (l0s2) edge[bend left=55] node[below right, align=center]{$\silentaction$} (l0s3);

				\path (lFs2) edge[bend left=55] node[right, align=center]{$\silentaction$} (lFs3);
				\path (lFps2) edge[bend left=55] node[right, align=center]{$\silentaction$} (lFps3);

				\path (lps0) edge node[above, align=center]{$\silentaction$} (lps1);
				\path (lps1) edge[bend left=45] node[below, align=center]{$\silentaction$} (lps2);
				\path (lps2) edge[bend right=55]  node[left, align=center]{$\silentaction$}(lps3);
				
				\draw[-,line width=1pt, dotted, dash pattern=on 0pt off 4pt, line cap=round] ([yshift=2.5em]l1s6.north) -- ++(0,-0.5);
			\end{tikzpicture}
			}
			
			\caption{$\RegionAutomaton{\Apriv{\TA'}}$}
			\label{figure:example-ngMETA-Apriv}
		\end{subfigure}

		\caption{Simplified region graphs of $\Apub{\TA}$ and $\Apriv{\TA}$ of \cref{figure:example-DiscPosGuarETA-N-tic}}
	\end{figure}
	\subsection{ET-EN-opacity in discrete positive (guarded) METAs}
	\begin{theorem}\label{theorem:ET-EN-fullweak-discrete-positive-guarded-METAs}
		For each $\existsWeakFull \in \{ \exists, \text{weak}, \text{full} \}$, $\existsWeakFull$-ET-EN-opacity can be decided in \twoEXPSPACE{} for discrete positive METAs, and \threeEXPSPACE{} for discrete positive guarded METAs.
	\end{theorem}
	\begin{proof}
		Given~$\TA$, we split increments of the form ``$+ \updateConstant$'' with $\updateConstant > 1$ into $\updateConstant$ consecutive $+1$ increments in 0-time following the construction of~$\TA'$ in the proof of \cref{theorem:fullweak-discrete-positive-guarded-ETAs}.
		Then, we modify~$\TA'$:
		\begin{enumerate}
			\item we add a global ``ticking clock'' $\clockTickOne$ reset every time unit via action~$\tickTime$ on any location, and a global invariant $\clockTickOne \leq 1$ (not depicted in our figures);
			\item we add to any guard the additional clock constraint $0 < \clockTickOne \leq 1$, enforcing that action $\tickTime$ occurs after the updates performed in 0-time on the same time unit;\footnote{%
				For sake of readability, we do not consider in our transformation the case of updates at time~0, made impossible by the guard $0 < \clockTickOne$; a simple fix is to start in a copy of the automaton with a global invariant $\clockTickOne = 0$ to allow for updates at time~0 and, as soon as time elapses, moving to the ``normal'' automaton with the guards $0 < \clockTickOne \leq 1$.
			}
			\item for each final location $\locfinal$ (made urgent), we add a new location~$\locfinal'$ (the only ones to be final) reachable from $\locfinal$ via action~$\tickTime$ if $\locfinal$ is reachable on an integer time (which can be tested thanks to~$\clockTickOne$), and via the fresh action~$\tickLastTimeRational$ otherwise.
		\end{enumerate}
		The last step of the construction is used to differentiate when the execution time is an integer, or when it is a non-integer (recall from~\cite{AAL24} that, if a location is reachable for a non-integer time, then it is for any other time in the open interval of length~1 with integer bounds to which this time belongs).
		We exemplify our construction on the discrete positive META of \cref{figure:example-DiscPosGuarETA} in \cref{figure:example-DiscPosGuarETA:ETEN}.

		This transformation, say~$\TA''$, is linear in the size of~$\TA'$.
		We then build $\Apub{\TA''}$ and $\Apriv{\TA''}$.
		The region automata $\RegionAutomaton{\Apub{\TA''}}$ and $\RegionAutomaton{\Apriv{\TA''}}$ over $\Actions = \{ \tickIncr_i \mid 1 \leq i \leq |\Energies| \} \cup \{ \tickTime , \tickLastTimeRational\}$ %
			now contain not only the increments of energy thanks to the $\tickIncr_i$ actions (as in the proof of \cref{theorem:fullweak-discrete-positive-guarded-ETAs}), but also the discrete time elapsing thanks to the $\tickTime$ actions; finally, the occurrence of one (resp.\ the absence of) $\tickLastTimeRational$ denotes a non-integer (resp.\ integer) execution time.
		Verifying $\exists$, weak and full opacity using Parikh image intersection and inclusion concludes the proof.
	\end{proof}
	\begin{figure}[tb]
		{\centering
			\scalebox{.77}{
			\begin{tikzpicture}[PTA, node distance=1cm and 2.5cm]
				\node[location, initial] (l0) {$\locinit$};
				\node[location, private, below = of l0] (lpriv) {$\locpriv$};
				\node[location, urgent, right = of l0] (l0') {$\locinit'$};
				\node[location, urgent, right = of l0'] (lF) {$\locfinal$};
				\node[location, final, right = of lF] (lFprime) {$\locfinal'$};

				\node[invariantNorth] at (l0.north) {$\clock \leq 3$};
				\node[invariantNorth] at (l0'.north) {$\clockZeroTime = 0$};
				\node[invariantSouth] at (lpriv.south) {$\clock \leq 3$};
				\node[invariantNorth] at (lF.north) {$\clockZeroTime = 0$};

				\path (l0) edge node[right, align=center]{$\silentaction$} (lpriv);
				\path (l0) edge node[above, align=center]{$\clock \geq 1$ \\ $\land 0 < \clockTickOne$ \\ $\styleact{\tickIncr_1}$} node[below, align=center]{$\clockZeroTime \assign 0$} (l0');
				\path (l0) edge[out=240,in=210,looseness=15] node[left, align=center]{$\clockTickOne = 1$ \\ $\tickTime$ \\ $\clockTickOne \assign 0$} (l0);

				\path (l0') edge node[above, align=center]{$\clockZeroTime = 0$ \\ $\styleact{\tickIncr_1}$} node[below]{$\clockZeroTime \assign 0$} (lF);

				\path (lpriv) edge[loop right] node[above right, align=center]{$0 < \clockTickOne$ \\ $\styleact{\tickIncr_1}$} (lpriv);
				\path (lpriv) edge[out=200,in=180,looseness=15] node[left, align=center]{$\clockTickOne = 1$ \\ $\tickTime$ \\ $\clockTickOne \assign 0$} (lpriv);

				\path (lpriv) edge[bend right] node[right, xshift=3.5em, align=center]{$\clock \geq 1 \land 0 < \clockTickOne$ \\ $\silentaction$ \\ $\clockZeroTime \assign 0$} (lF);

				\path (lF) edge[bend right] node[below, align=center]{$\clockTickOne = 1$ \\ $\tickTime$} (lFprime);
				\path (lF) edge[bend left] node[align=center]{$0 < \clockTickOne < 1$ \\ $\styleact{\tickLastTimeRational}$} (lFprime);
			\end{tikzpicture}
			}

			}

		\caption{Exemplifying the construction in the proof of \cref{theorem:ET-EN-fullweak-discrete-positive-guarded-METAs}}
		\label{figure:example-DiscPosGuarETA:ETEN}
	\end{figure}
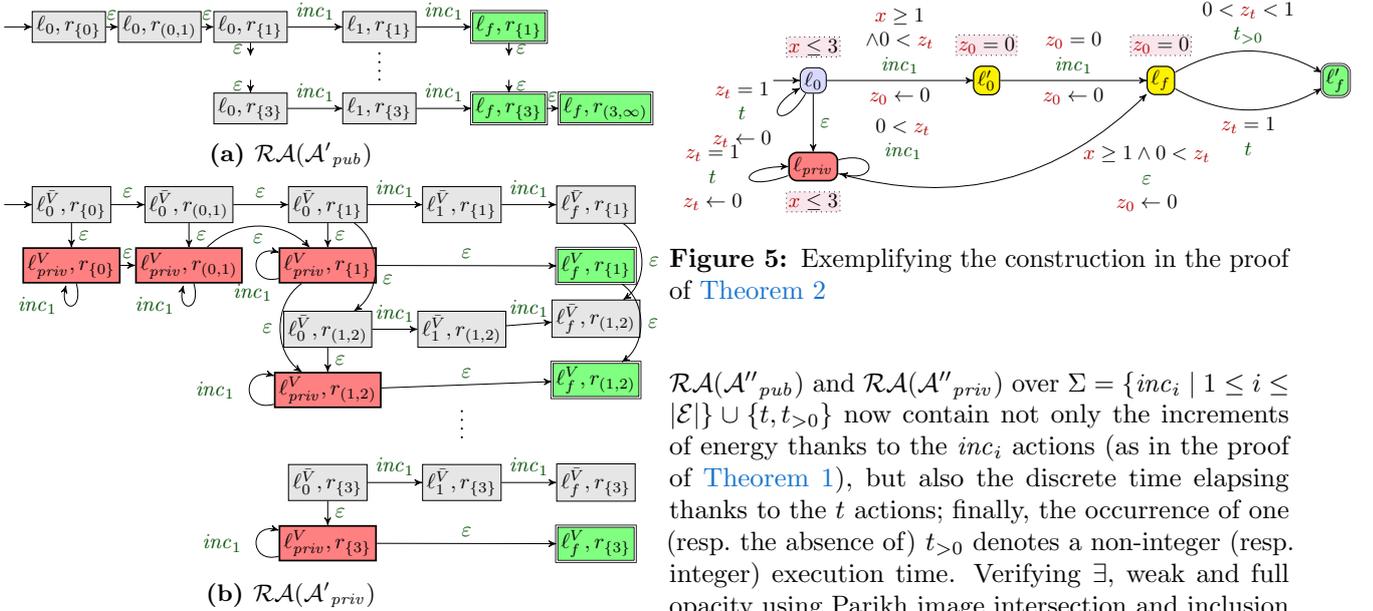
	\begin{example}
		The private run $\varrun_1$ in the proof of \cref{theorem:fullweak-discrete-positive-guarded-ETAs} gives $\tickTime \tickIncr_1 \tickIncr_1 \tickTime \tickLastTimeRational$, while the public run $\varrun_2$ gives $\tickTime \tickTime \tickIncr_1 \tickIncr_1 \tickLastTimeRational$.
		Even though these words differ, their Parikh image is $\{(2,2,1)\}$ (denoting 2, 2, 1 occurrences of $\tickIncr_1$, $\tickTime$, $\tickLastTimeRational$ respectively), and therefore the META in \cref{figure:example-DiscPosGuarETA} is $\exists$-ET-EN-opaque.
	\end{example}
	\subsection{(ET-)EN-opacity in discrete guarded ETAs}

	We now consider discrete guarded ETAs (with a single energy variable) with both increment and decrement.
	Recall that two energy variables make our problems undecidable (\cref{proposition:undecidability-2-energies}).

	\begin{theorem}\label{theorem:EN:discrete-ETAs}
		For each $\existsWeakFull \in \{ \exists, \text{weak}, \text{full} \}$
		and $\ENorETEN \in \{ \text{EN}, \text{ET-EN} \}$,
		$\existsWeakFull$-$\ENorETEN$-opacity is decidable %
			for discrete ETAs.
	\end{theorem}
	\begin{proof}
		We start from the construction~$\TA''$ in the proof of \cref{theorem:ET-EN-fullweak-discrete-positive-guarded-METAs}, and extend it to decrements, \ie{} using $\tickDecr$ instead of~$\tickIncr$ when appropriate.
		As before, we then construct the NFAs $\Language(\RegionAutomaton{\Apriv{\TA''}}) $ and $ \Language(\RegionAutomaton{\Apub{\TA''}})$.
		However, adding $\tickDecr$ symbols to the construction in the proof of \cref{theorem:ET-EN-fullweak-discrete-positive-guarded-METAs} may allow some impossible runs, typically causing the energy to drop below~0.
		Therefore, we turn each of these two NFAs into a pushdown automaton~(PDA) (see \cref{definition:pushdown} recalled in \cref{appendix:definition:pushdown}), as follows:
		\begin{enumerate}
			\item
			We define a stack~$\PDAstackAlphabet$ encoding the value of the energy, over a single symbol~$\PDAstackEnergy$ (in addition to the empty stack symbol~$\initialStackSymbol$);
			\item We turn each transition of these NFAs with action~$\tickIncr$ (resp.~$\tickDecr$) into a silent transition pushing a symbol~$\PDAstackEnergy$ to (resp.\ popping a symbol from) the stack;
			\item For each accepting region~$\qfinal$,
				we add a self-loop over~$\qfinal$ via a fresh action~$\action$ popping an element~$\PDAstackEnergy$ of the stack;
				we add a new state~$\qfinal'$ connected from~$\qfinal$ via a silent transition testing for emptiness (\ie{} the popped symbol is~$\initialStackSymbol$);
			\item we make all states~$\qfinal$ non-accepting, and all~$\qfinal'$ accepting.
		\end{enumerate}
		An example of this transformation applied to a fictitious NFA is shown in \cref{figure:example-EN-NFA-tf}.
		The languages of these PDAs are made of words of the form $\{ \tickTime , \action, \tickLastTimeRational \}^*$, where the number of ``$\action$'' (resp.\ of~``$\tickTime$'') encodes the final energy (resp.\ the duration) of the run, and $\tickLastTimeRational$ encodes the non-integerness of the duration.
		Now, we can compute the intersection (resp.\ inclusion) of the Parikh images of the languages of these NPDAs, which is \coNP{}-complete (resp.\ $\Pi^P_2$-complete)~\cite{KT10}, thus solving $\exists$-ET-EN opacity (resp.\ weak- and full ET-EN-opacity).
		Solving EN-opacity follows exactly the same reasoning, by making all $\tickTime, \tickLastTimeRational$ actions silent.
	\end{proof}
	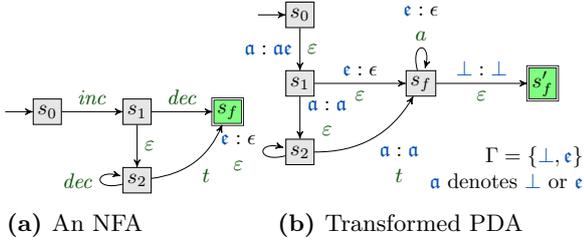
\begin{figure}[tb]
				\begin{subfigure}[b]{.3\columnwidth}
				\centering
				\scalebox{.8}{
					\begin{tikzpicture}[NFA, node distance=.7cm and 1cm]
						\node[state, initial] (l0) {$\qinit$};
						\node[state, right = of l0] (l1) {$\qi{1}$};
						\node[state, below = of l1] (l2) {$\qi{2}$};
						\node[state, final, right = of l1] (lF) {$\qfinal$};
						
						\path (l0) edge node[above, align=center]{$\styleenergy{\tickIncr}$} (l1);
						\path (l1) edge node[right, align=center]{$\silentaction$} (l2);
						\path (l1) edge node[above, align=center]{$\styleenergy{\tickDecr}$} (lF);
						\path (l2) edge[loop left] node[left, align=center]{ $\styleenergy{\tickDecr}$} (l2);
						\path (l2) edge[bend right] node[below right, align=center]{$\tickTime$} (lF);
					\end{tikzpicture}
				}
				
				\caption{An NFA}
				\label{figure:example-EN-NFA}
			\end{subfigure}
			\hfill{}
			\begin{subfigure}[b]{.65\columnwidth}
				\centering
				\scalebox{.8}{
					\begin{tikzpicture}[NFA, node distance=.7cm and 1.5cm]
						\node[state, initial] (l0) {$\qinit$};
						\node[state, below = of l0] (l1) {$\qi{1}$};
						\node[state, below = of l1] (l2) {$\qi{2}$};
						\node[state, right = of l1] (lF) {$\qfinal$};
						\node[state, final,  right = of lF] (F) {$\qfinal'$};
						
						\path (l0) edge node[left, align=center]{$\stackSymbol : \stackSymbol \PDAstackEnergy$}node[right, align=center] {$\silentaction$} (l1);
						\path (l1) edge node[right, align=center]{$\stackSymbol : \stackSymbol$ \\ $\silentaction$} (l2);
						\path (l1) edge node[above, align=center]{$\PDAstackEnergy : \emptyword$} node[below, align=center] {$\silentaction$} (lF);
						\path (l2) edge[loop left] node[left, align=center]{ $\PDAstackEnergy : \emptyword$ \\ $\silentaction$} (l2);
						\path (l2) edge[bend right] node[below right, align=center] {$\stackSymbol : \stackSymbol$ \\ $\tickTime$} (lF);
						\path (lF) edge node[above, align=center]{$\initialStackSymbol : \initialStackSymbol$} node[below, align=center]{$\silentaction$} (F);
						\path (lF) edge[loop above] node[above, align=center]{$\PDAstackEnergy : \emptyword$ \\ $\action$} (lF);

						\node[below = of lF, align=right, xshift=4em]{$\PDAstackAlphabet = \{ \initialStackSymbol , \PDAstackEnergy \}$
							\\
		$\stackSymbol$ denotes $\initialStackSymbol$ or $\PDAstackEnergy$};
					\end{tikzpicture}
				}
				\caption{Transformed PDA}
				\label{figure:example-EN-PDA-tf}
			\end{subfigure}
		\caption{Construction in the proof of \cref{theorem:EN:discrete-ETAs}}
		\label{figure:example-EN-NFA-tf}
	\end{figure}

	For \emph{guarded} ETAs, we cannot reuse the reasoning removing energy guards from \cref{proposition:removing-guards-discrete}, as the energy can also decrease.
	However, we can combine this mechanism with the stack, leading to the following result (proof is in \cref{appendix:theorem:EN:discrete-guarded-ETAs}).

	\begin{restatable}{theorem}{theoremENDiscreteGuardedETAs}\label{theorem:EN:discrete-guarded-ETAs}
		For each $\existsWeakFull \in \{ \exists, \text{weak}, \text{full} \}$
		and $\ENorETEN \in \{ \text{EN}, \text{ET-EN} \}$,
		$\existsWeakFull$-$\ENorETEN$-opacity is decidable %
			for discrete guarded ETAs.
	\end{restatable}
	\section{Opacity with energy observation every time unit}\label{section:DE}

	We now consider the case when the attacker can observe the current energy level every time unit.
	The model remains over dense time, but the attacker has only a discrete observation power.

	\subsection{Observing the energy every time unit in discrete positive guarded ETAs}

	We first consider the setting when the attacker can observe the exact energy valuation every time unit.

	Given a run \(\varrun = (\loci{0}, \clockval_0, \enerval_0), (\paramd_0, \edgei{0}), (\loci{1}, \clockval_1, \enerval_1), \ldots, (\loci{n}, \clockval_n, \enerval_n)\) of a discrete positive guarded META,
	we define the \emph{energy level at time~$t$} as follows:

	\[
	\EnerLevelT(\varrun, t) = \left\{
		\begin{array}{ll}
			\enervalZero & \mbox{if } t < d_0 \\
			\enerval_k \text{ where } k = \max_i\big(( \sum_{j=0}^i d_j) \leq t \big) & \mbox{otherwise.}
		\end{array}
	\right.
	\]

	Then, the discrete energy observation of~$\varrun$ is:

		\[ \DiscreteEnergyObs(\varrun) = \EnerLevelT(\varrun, 1) , \dots, \EnerLevelT(\varrun, k) \text { with } k = \ceiling{\sum_{j=0}^n d_i} \text{.} \]

	A guarded META~$\TA$ is \emph{weakly discrete-energy opaque (weakly DE-opaque)} whenever, for each run $\varrun \in \PrivVisit{\TA}$, there exists a run $\varrun' \in \PubVisit{\TA}$ such that $\DiscreteEnergyObs(\varrun) = \DiscreteEnergyObs(\varrun')$.
	Full and existential counterparts are defined similarly.
	We define the discrete-energy opacity problem as follows (for each $\existsWeakFull \in \{ \exists, \text{weak}, \text{full} \}$):

	\defProblem{$\existsWeakFull$-DE-opacity}
		{A guarded META~$\TA$}
		{Is $\TA$ $\existsWeakFull$-DE-opaque?}
	\begin{example}\label{example:DEO}
		Consider again the discrete positive META in \cref{figure:example-META-A}.
		Let $\varrun$ be the private run $(\locinit, 0, 0) \FlecheConcrete{(1.8, \styleact{a})} (\locpriv, 0, 1) \FlecheConcrete{(0.4, \styleact{b})} (\locpriv, 0, 3)\FlecheConcrete{(0.8, \styleact{b})} (\locpriv, 0, 5) \FlecheConcrete{(0.7, \styleact{b})} (\locfinal, 0.7, 5)$.
		We have $\DiscreteEnergyObs(\varrun) = 0 , 1, 5 , 5$.

	\end{example}

	We first rule out non-necessarily positive guarded METAs; the proof is a simple adaptation from \cref{proposition:undecidability-2-energies}.

	\begin{restatable}{proposition}{undecTwoEnergiesDE}\label{proposition:undecidability-2-energies-DE}
		For each $\existsWeakFull \in \{ \exists, \text{weak}, \text{full} \}$,
		$\existsWeakFull$-DE-opacity for discrete guarded METAs with 2~energies and 1~clock is undecidable.
	\end{restatable}

	The following positive result only holds for positive ETAs, \ie{} with a single energy.

	\begin{restatable}{theorem}{theoremDEopacity}\label{theorem:DE-Opacity}
		For each $\existsWeakFull \in \{ \exists, \text{weak}, \text{full} \}$, $\existsWeakFull$-DE-opacity can be decided in \EXPSPACE{} for discrete positive (guarded) ETAs.
	\end{restatable}

	The following result extends DE-opacity to METAs, but only for the $\exists$ variant.
	We use a novel construction of ``Parikh by block'' automaton.
	\begin{restatable}{theorem}{theoremDEopacityMETA}\label{theorem:DE-Opacity-META}
		 $\exists$-DE-opacity
			is decidable
		for discrete positive (guarded) METAs.
	\end{restatable}
	\begin{proof}[Proof (sketch)]
		We first build $\TA'$ similarly to the proof of \cref{theorem:ET-EN-fullweak-discrete-positive-guarded-METAs}, and then construct the region automata of $\Apriv{\TA'}$ and~$\Apub{\TA'}$, where $\TA'$ embeds the $\tickIncr_i$ actions (encoding discrete increments for the various energy variables), and time ticks (actions~$\tickTime$, and at most one final action~$\tickFin$ (which was denoted by~$\tickLastTimeRational$ in \cref{theorem:ET-EN-fullweak-discrete-positive-guarded-METAs})).
		Then, verifying $\exists$-DE-opacity amounts to checking whether there exist a run in $\Apriv{\TA'}$ and a run of $\Apub{\TA'}$ with the same execution time (same number of~$\tickTime$), and the same energy valuations at each tick, \ie{} the same number of increments~$\tickIncr_i$ in between two consecutive $\tickTime$~actions, regardless of their order---which is indeed not visible to the attacker.

		To this end, we define and use a novel construction, which we call ``Parikh by block'' ($\ParikbB$) automaton: the idea is to count the occurrences of different actions (here the~$\tickIncr_i$), between two consecutive ticks (actions~$\tickTime$ or~$\tickFin$)---hence the notion of ``block''.
		Counting such numbers of actions between two ticks will be represented using a Parikh image, \ie{} a semilinear set~\cite{Parikh66}.
		The resulting structure is a finite automaton over actions~$\tickTime$ and~$\tickFin$, where each transition is additionally labelled by a Parikh image over the~$\tickIncr_i$, encoding the number of $\tickIncr_i$ since the last tick.
		We then perform the synchronised product of $\ParikbB(\Apriv{\TA'})$ and $\ParikbB(\Apub{\TA'})$ over $\{ \tickTime , \tickFin \}$; each transition in the resulting product (which is a finite automaton) is labelled by \emph{two} Parikh images, one coming from $\ParikbB(\Apriv{\TA'})$ and one from $\ParikbB(\Apub{\TA'})$.
		Finally, we check the existence of an accepting path in this automaton such that, for each transition, the intersection of the two Parikh images is non-empty.
		All these operations on Parikh images can be computed using the (decidable) Presburger arithmetics~\cite{BHK17}.
		This indeed denotes the existence of two paths, one private and one public, with the same execution time, and with the same energy level at each time unit (correctly encoded by the number of $\tickIncr_i$ between consecutive $\tickTime$ actions)---which decides $\exists$-DE-opacity.
		See additional details in \cref{appendix:proof:theorem:DE-Opacity-META}.
	\end{proof}

		\begin{figure}[tb]
				\begin{subfigure}[b]{.45\columnwidth}
				\centering
				\scalebox{.8}{
					\begin{tikzpicture}[NFA, node distance=.5cm and 1cm]
						\node[state, initial] (l0) {$\qinit$};
						\node[state, right = of l0] (l1) {$\qi{1}$};
						\node[state, right = of l1] (l2) {$\qi{2}$};
						\node[state, private, below = of l1] (l3) {$\qi{3}$};
						\node[state, final, right = of l3] (lf) {$\qfinal$};

						\path (l0) edge[bend left] node[above, align=center]{$\styleact{a}$} (l1);
						\path (l1) edge[bend left] node[above, align=center]{$\styleact{b}$} (l0);
						\path (l0) edge[bend right] node[above, align=center]{$\styleact{b}$} (l3);
						\path (l1) edge node[above, align=center]{$\styleact{t}$} (l2);
						\path (l2) edge[loop right] node[right, align=center]{	$\styleact{t}$} (l2);
						\path (l2) edge node[above, align=center]{$\styleact{b}$} (l3);
						\path (l3) edge node[right, align=center]{$\styleact{t}$} (l1);
						\path (l3) edge node[above, align=center]{$\styleact{a}$} (lf);

					\end{tikzpicture}
				}
				
				\caption{An NFA~$\TN$}
				\label{figure:example-DE-NFA-N}
			\end{subfigure}
			\hfill{}
			\begin{subfigure}[b]{.45\columnwidth}
				\centering
				\scalebox{.8}{
					\begin{tikzpicture}[NFA, node distance=.5cm and 1cm]
						\node[state, initial] (l0) {$\qinit'$};
						\node[state, right = of l0] (l1) {$\qi{1}'$};
						\node[state, final, right = of l1] (lf) {$\qfinal'$};
						
						\path (l0) edge[loop above] node[above, align=center]{$\styleact{a}$} (l0);
						\path (l1) edge[loop above] node[above, align=center]{$\styleact{b}$} (l1);
						\path (l1) edge[bend left] node[above, align=center]{$\styleact{t}$} (l0);
						\path (l0) edge[bend left] node[above, align=center]{$\styleact{t}$} (l1);
						\path (l1) edge node[above, align=center]{$\styleact{a}$} (lf);
					\end{tikzpicture}
				}
				
				\caption{An NFA $\TM$}
				\label{figure:example-DE-NFA-M}
			\end{subfigure}

			\begin{subfigure}[b]{.7\linewidth}
				\centering
				\scalebox{.8}{
					\begin{tikzpicture}[NFA, node distance=1cm and 2.5cm]
						\node[state, initial] (l0) {$\qinit$};
						\node[state, right = of l0] (l1) {$\qi{1}$};
						\node[state, right = of l1] (l2) {$\qi{2}$};
						\node[state, final, below = of l1] (lf) {$\qfinal$};
						
						\path (l0) edge[bend left=65] node[above, align=center]{$\styleact{t}$\\$\styleact{(1,0) + \alpha(1,1)}$} (l2);
						\path (l0) edge node[above, align=center]{$\styleact{t}$\\$\styleact{(0,1) + \alpha(1,1)}$} (l1);
						\path (l1) edge[bend left] node[above, align=center]{$\styleact{t}$ \\ $\styleact{\alpha(1,1)}$} (l2);
						\path (l2) edge[loop right] node[below, align=center]{$\styleact{t}$\\$\styleact{(0,0)}$} (l2);
						\path (l2) edge[bend left] node[above, align=center]{$\styleact{t}$\\$\styleact{(0,1)}$} (l1);
						\path (l1) edge[loop above] node[ align=center]{$\styleact{t}$\\$\styleact{(0,2) + \alpha(1,1)}$} (l2);
						
						\path (l1) edge node[left, align=center]{$\tickFin$\\$\styleact{(1,2) + \alpha(1,1)}$} (lf);
						\path (l0) edge[bend right=45] node[below, align=center]{$\tickFin$\\$\styleact{(1,1) + \alpha(1,1)}$} (lf);
						\path (l2) edge[bend left=45] node[below, align=center]{$\tickFin$\\$\styleact{(1,1)}$} (lf);

					\end{tikzpicture}
				}
				
				\caption{$\ParikbB(\TN)$}
				\label{figure:example-DE-by-block-N}
			\end{subfigure}
			\hfill{}
			\begin{subfigure}[b]{.28\columnwidth}
				\centering
				\scalebox{.8}{
					\begin{tikzpicture}[NFA, node distance=1cm and 1cm]
						\node[state, initial] (l0) {$\qinit'$};
						\node[state, below = of l0] (l1) {$\qi{1}'$};
						\node[state, final, below = of l1] (lf) {$\qfinal'$};

						\path (l0) edge[bend left] node[right, align=center]{$\styleact{t}$\\$\styleact{\alpha(1,0)}$} (l1);
						\path (l1) edge[bend left] node[left, align=center]{$\styleact{t}$\\$\styleact{\alpha(0,1)}$} (l0);
						\path (l1) edge node[right, align=center]{$\tickFin$\\$\styleact{(1,0) + \alpha(0,1)}$} (lf);
					\end{tikzpicture}
				}
				
				\caption{$\ParikbB(\TM)$}
				\label{figure:example-DE-by-block-M}
			\end{subfigure}
			
			\begin{subfigure}[b]{\linewidth}
				\centering
				\scalebox{.8}{
					\begin{tikzpicture}[NFA, node distance=1cm and 1.5cm]
						\node[state, initial] (l0s0) {$\qinit,\qinit'$};
						\node[state, right = of l0s0] (l1s1) {$\qi{1},\qi{1}'$};
						\node[state, right = of l1s1] (l2s0) {$\qi{2},\qi{0}'$};
						\node[state, below = of l1s1] (l1s0) {$\qi{1},\qi{0}'$};
						\node[state, below = of l2s0] (l2s1) {$\qi{2},\qi{1}'$};
						\node[state, final, right = of l2s0] (lfsf) {$\qfinal,\qfinal'$};
						
						\path (l0s0) edge[dotted] node[above, align=center]{$\styleact{t}$} (l1s1);
						\path (l0s0) edge[bend right=65] node[left, align=center]{$\styleact{t}$} (l2s1);
						
						\path (l1s0) edge[bend left, dotted] node[below, align=center]{$\styleact{t}$} (l2s1);
						\path (l1s1) edge[bend left, dotted] node[below, align=center]{$\styleact{t}$} (l2s0);
						\path (l2s0) edge[bend left, dotted] node[below, align=center]{$\styleact{t}$} (l1s1);
						\path (l2s1) edge[bend left] node[above, align=center]{$\styleact{t}$} (l1s0);
						
						\path (l1s0) edge[bend left, dotted] node[left, align=center]{$\styleact{t}$} (l1s1);
						\path (l1s1) edge[bend left] node[right, align=center]{$\styleact{t}$} (l1s0);
						\path (l2s0) edge[bend left] node[right, align=center]{$\styleact{t}$} (l2s1);
						\path (l2s1) edge[bend left] node[left, align=center]{$\styleact{t}$} (l2s0);
						
						\path (l1s1) edge[bend left] node[below right, align=center]{$\tickFin$} (lfsf);
						\path (l2s1) edge[bend right] node[below, align=center]{$\tickFin$} (lfsf);
					\end{tikzpicture}
				}
				
				\caption{Synchronised product of $\ParikbB(\TN)$ and $\ParikbB(\TM)$ }
				\label{figure:example-DE-sync}
			\end{subfigure}
			
			\caption{Illustrating the Parikh by block construction}
			\label{figure:example-NFA-par_bloc}
			\end{figure}
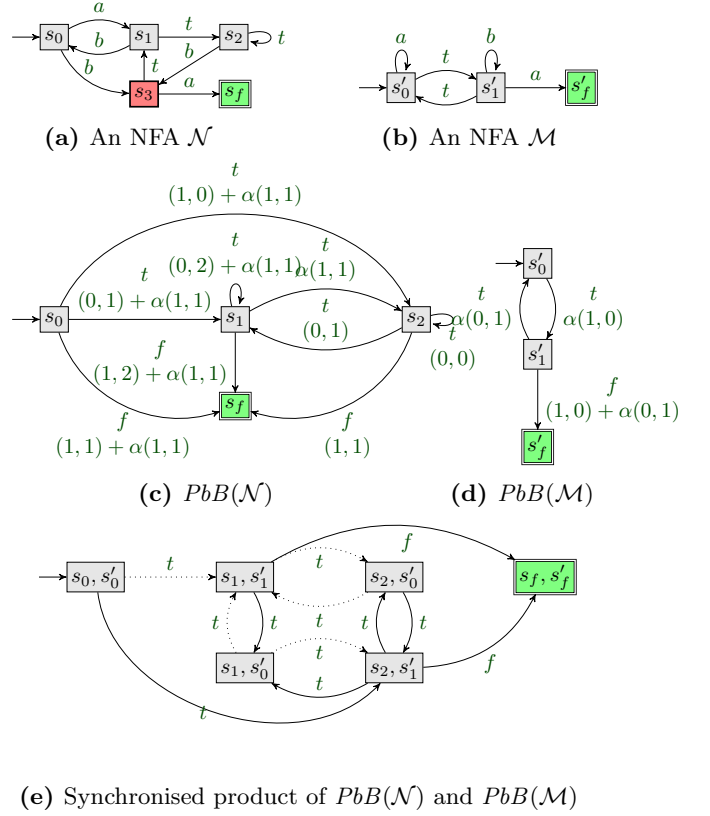
	\begin{example}\label{example-Parikh-by-blocks}
		Let us illustrate our Parikh by block construction.
		Consider the NFAs $\TN,\TM$ in \cref{figure:example-DE-NFA-N,figure:example-DE-NFA-M} (which could be the region automata of $\Apriv{\TA'}$ and~$\Apub{\TA'}$ for some discrete positive META~$\TA$, with $a$ and $b$ denoting for example~$\tickIncr_1$ and~$\tickIncr_2$).
		Considering a Parikh by block construction over $\{ \tickTime, \tickFin \}$, we give $\ParikbB(\TN)$ and $\ParikbB(\TM)$ in \cref{figure:example-DE-by-block-N,figure:example-DE-by-block-M} respectively, where the semilinear set $(0,2) + \alpha (1,1)$ denotes ``exactly 2~$b$s plus an arbitrarily (yet identical) number of $a$s and~$b$s''.
		Their synchronised product is given in \cref{figure:example-DE-sync}, where labels (Parikh images) are omitted;
		a dashed line denotes an empty intersection of the two associated Parikh images, while a plain line denotes a non-empty intersection.
		For example, the edge from $(\qinit, \qinit')$ to $(\qi{1}, \qi{1}')$ is labelled with ``$(0,1) + \alpha(1,1)$'' and ``$\alpha(1, 0)$''; the intersection of these two semilinear sets is empty.
		In contrast, the edge from $(\qinit, \qinit')$ to $(\qi{2}, \qi{1}')$ is labelled with ``$(1,0) + \alpha(1,1)$'' and ``$\alpha(1, 0)$''; the intersection is clearly non-empty, for example $(1,0)$ is a solution (denoting one~$a$, and no~$b$).
		Finally, we can exhibit a path over edges with a non-empty intersection---denoting that the (fictitious) original META~$\TA$ is $\exists$-DE-opaque.
	\end{example}
	\subsection{Observing the buffered energy changes every time unit in discrete METAs}

	We now consider an attacker able to see all the evolution of energy and their order, but with no knowledge of the exact timestamps between two consecutive time units.
	That is, the attacker can make a difference between an increase of~2 energy units followed by a decrease of~3 (within the same time unit) compared to a decrease of~3 followed by an increase of~2---yet, they will not be able to know the exact timestamp of these events.
	This can be seen as the natural setting in which a buffer storing all energy changes is read every time unit by the attacker; this is also reminiscent of the recent setting of the weakened attacker able to observe absolute time only at integer times (``intruder with discrete-time precision'') in~\cite{DQY25}.

	Let $\emptyseq$ denote the empty sequence.
	Given $\abstime \in \setN \setminus \{ 0 \}$, let $\EnerLevelBufferT(\varrun, \abstime)$ be the sequence of energy valuations defined as follows:
	when $\abstime > 1$, this is the sequence of energy valuations obtained from the subrun of~$\varrun$ by keeping only timestamps in~$(\abstime-1 , \abstime]$ which modify at least one energy variable;
	the case $\abstime = 1$ is obtained by keeping the timestamps in~$[0,1]$.
	(See \cref{definition:bEL} in \cref{appendix:definition:bEL} for a formal definition.)
	Then, the buffered discrete energy observation of~$\varrun$ is:

	\[ \buffDiscreteEnergyObs(\varrun) = \EnerLevelBufferT(\varrun, 1) , \dots, \EnerLevelBufferT(\varrun, k) \text { with } k = \ceiling{\sum_{j=0}^{n-1} d_i} \text{.} \]

	A discrete guarded META~$\TA$ is \emph{weakly buffered-discrete-energy opaque (weakly bDE-opaque)} when, for each run $\varrun \in \PrivVisit{\TA}$, there exists a run $\varrun' \in \PubVisit{\TA}$ such that $\buffDiscreteEnergyObs(\varrun) = \buffDiscreteEnergyObs(\varrun')$.
	Full and existential counterparts are defined similarly.
	We define the buffered-discrete-energy opacity problem as expected.

	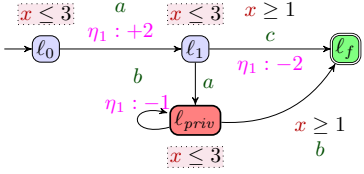
\begin{figure}[tb]
		\centering
			\scalebox{.8}{
			\begin{tikzpicture}[PTA, node distance=.7cm and 2cm]
				\node[location, initial] (l0) {$\locinit$};
				\node[location, right = of l0] (l1) {$\loci{1}$};
				\node[location, private, below = of l1] (lpriv) {$\locpriv$};
				\node[location, final, right = of l1] (lF) {$\locfinal$};

				\node[invariantNorth] at (l0.north) {$\clock \leq 3$};
				\node[invariantNorth] at (l1.north) {$\clock \leq 3$};
				\node[invariantSouth] at (lpriv.south) {$\clock \leq 3$};

				\path (l0) edge node[above, align=center]{$\styleact{a}$\\$\styleenergy{\enervar_1 : +2}$} (l1);
				\path (l1) edge node[right, align=center]{$\styleact{a}$} (lpriv);
				\path (l1) edge node[above, align=center]{$\clock \geq 1$ \\ $\styleact{c}$} node[below, align=center]{$\styleenergy{\enervar_1 : -2}$} (lF);
				\path (lpriv) edge[loop left] node[above, align=center]{%
					$\styleact{b}$ \\ $\styleenergy{\enervar_1 : -1}$} (lpriv);
				\path (lpriv) edge[bend right] node[below right, align=center]{$\clock \geq 1$ \\ $\styleact{b}$} (lF);
			\end{tikzpicture}
			}

		\caption{A discrete ETA}
		\label{figure:example-bDEO-A}
	\end{figure}
	\begin{example}\label{example-bDEO}
		Consider the following runs of the discrete ETA~$\TA$ in \cref{figure:example-bDEO-A}.
		Let $\varrun_1$ be the private run $(\locinit, 0, 0) \FlecheConcrete{(1.5, \styleact{a})} (\loci{1}, 1.5, 2) \FlecheConcrete{(1, \styleact{a})} (\locpriv, 2.5, 2)\FlecheConcrete{(0.1, \styleact{b})} (\locpriv, 2.6, 1) \FlecheConcrete{(0.1, \styleact{b})} (\locpriv, 2.7, 0)  \FlecheConcrete{(0.1, \styleact{b})} (\locfinal, 2.8, 0)$.
		Let $\varrun_2$ be the public run $(\locinit, 0, 0) \FlecheConcrete{(1.5, \styleact{a})} (\loci{1}, 1.5, 2) \FlecheConcrete{(1.3, \styleact{c})} (\locfinal, 2.8, 0)$.
		We have $\DiscreteEnergyObs(\varrun_1) = \DiscreteEnergyObs(\varrun_2) = 0 , 2, 0$, and therefore $\TA$ is $\exists$-DE-opaque.
		Now, let us address bDE-opacity: we have $\buffDiscreteEnergyObs(\varrun_1) = \emptyseq , (2) , (1,0) $
		while
		$\buffDiscreteEnergyObs(\varrun_2) = \emptyseq , (2) , (0) $;
		it can actually be shown that $\TA$ is \emph{not} $\exists$-bDE-opaque.
	\end{example}
	\begin{figure}[tb]
		{\centering
			\scalebox{.8}{
			\begin{tikzpicture}[PTA, node distance=1cm and 1.7cm]
				\node[location, initial] (l1) {$\locinit$};
				\node[location, urgent, right = of l1] (l1') {$\locinit'$};
				\node[location, urgent, right = of l1'] (l1'') {$\locinit''$};

				\node[location, private, below = of l1'] (lpriv) {$\locpriv$};
				\node[location, urgent, below = of lpriv] (lpriv') {$\locpriv'$};

				\node[location, right = of l1''] (lF) {$\locfinal$};
				\node[location, final, right = of lF] (lF') {$\locfinal'$};

				\node[invariantNorth] at (l1.north) {$\clock \leq 3$};
				\node[invariantNorth] at (l1'.north) {$\clockZeroTime = 0$};
				\node[invariantNorth] at (l1''.north) {$\clockZeroTime = 0$};
				\node[invariantNorth] at (lpriv.north) {$\clock \leq 3$};
				\node[invariantEast] at (lpriv'.east) {$\clockZeroTime = 0$};
				\node[invariantNorth] at (lF.north) {$\clockZeroTime = 0$};

				\path (l1) edge[loop below] node[below left, align=center]{$\clockTickOne = 1$ \\ $\tickTime$ \\ $\clockTickOne \assign 0$} (l1);

				\path (l1) edge[bend right] node[right, align=center]{$0 < \clockTickOne $ \\ $\silentaction$} (lpriv);
				\path (l1) edge node[above, align=center]{$\clock \geq 1$ \\ $\land 0 < \clockTickOne $ \\ $\styleact{\tickIncr_1}$} node[below, align=center]{$\clockZeroTime \assign 0$} (l1');

				\path (l1') edge node[above, align=center]{$\styleact{\tickIncr_1}$} node[below]{$\clockZeroTime \assign 0$} (l1'');
				\path (l1'') edge node[above, align=center]{$\styleact{\tickFin}$} node[below]{$\clockZeroTime \assign 0$} (lF);

				\path (lpriv) edge[bend left] node[align=center]{$0 < \clockTickOne $ \\ $\styleact{\tickIncr_1}$ \\ $\clockZeroTime \assign 0$} (lpriv');
				\path (lpriv') edge[bend left] node[align=center]{$\styleact{\tickFin}$} (lpriv);
				\path (lpriv) edge[out=220,in=200,looseness=15] node[below left, align=center]{$\clockTickOne = 1$ \\ $\tickTime$ \\ $\clockTickOne \assign 0$} (lpriv);

				\path (lpriv) edge[bend right] node[below right, align=center]{$\clock \geq 1$ \\ $\land 0 < \clockTickOne $  \\ $\silentaction$ \\ $\clockZeroTime \assign 0$} (lF);

				\path (lF) edge[bend right] node[below, align=center]{$\clockTickOne = 1$ \\ $\tickTime$} (lF');
				\path (lF) edge[bend left] node[align=center]{$0 < \clockTickOne < 1$ \\ $\silentaction$} (lF');
			\end{tikzpicture}
			}

		}

		\caption{Exemplifying the construction in the proof of \cref{theorem:bDE-Opacity}}
		\label{figure:example-bDEO}
	\end{figure}
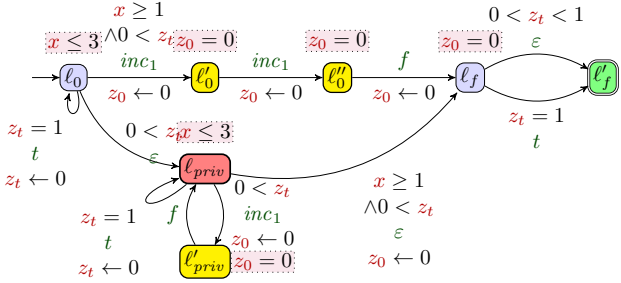
	\begin{restatable}{theorem}{theorembDEopacityPositive}\label{theorem:bDE-Opacity}
		For each $\existsWeakFull \in \{ \exists, \text{weak}, \text{full} \}$, $\existsWeakFull$-bDE-opacity can be decided in \EXPSPACE{} for discrete positive METAs,
		and in \twoEXPSPACE{} for discrete positive guarded METAs.
	\end{restatable}
	\begin{proof}
		See \cref{appendix:proof:theorem:bDE-Opacity}.
	\end{proof}

	Our last result is concerned with non-necessarily positive discrete ETAs, and makes use of operations over both regular and context-free languages.

	\begin{restatable}{theorem}{theorembDEopacity}\label{theorem:bDE-Opacity-ETAs}
		For each $\existsWeakFull \in \{ \exists, \text{weak}, \text{full} \}$, $\existsWeakFull$-bDE-opacity can be decided in \twoEXPSPACE{} for discrete ETAs.
	\end{restatable}

	Note that the case of discrete (guarded) METAs is not covered by \cref{theorem:bDE-Opacity,theorem:bDE-Opacity-ETAs} and remains open.
	\section{Opacity for integer-switching METAs}\label{section:ISMETA}

	After considering discrete METAs in the former sections, we finally consider here a (first) case in which energy variables can evolve continuously.
	A guarded META is said \emph{integer-switching (IS)} whenever all changes of energy rates are made on an (absolute) integer time.
	Formally, there exists no run $(\loc_0, \clockval_0, \enerval_0) \FlecheConcrete{(d_0, \edge_0)} (\loc_1, \clockval_1, \enerval_1) \FlecheConcrete {(d_1, \edge_1)} \cdots \FlecheConcrete{(d_{m-1}, \edge_{m-1})} (\loc_m, \clockval_m, \enerval_m)$ such that $\exists 0 \leq i \leq m-1 , \exists \enervar \in \Energies$ such that $\rates(\loc_i)(\enervar) \neq \rates(\loc_{i+1})(\enervar)$ and $\sum_{j = 0}^{i} d_j \notin \setN$.
	Note that both discrete energy updates and location changes (not modifying the rates) are allowed to occur at non-integer times, as well as any clock reset, keeping this model entirely dense-time.
	While this condition is essentially semantic, it can easily be checked on unguarded METAs by checking the region graph.
		For example, the META in \cref{figure:IS-META} is~IS.
	In addition, a guarded META is \emph{integer execution-time (iET)} whenever any final location is reached only on integer absolute times.
		For example, the META in \cref{figure:IS-META} is \emph{not} iET, due to the guard $\clock \geq 1$ to~$\locfinal$.

	We now introduce a transformation from IS-METAs to METAs that will allow us to transfer some decidability results.
	\begin{enumerate}
		\item we add a global ``ticking clock'' $\clockTickOne$ reset every time unit via action~$\tickTime$ on any location, and a global invariant $\clockTickOne \leq 1$ (not depicted in our figures), as in the proof of \cref{theorem:ET-EN-fullweak-discrete-positive-guarded-METAs};
		\item we add to any guard the additional clock constraint $0 \leq \clockTickOne < 1$, enforcing that outgoing guards are evaluated \emph{after} the discrete energy updates added by our construction;
		\item every time unit (when $\clockTickOne = 1$), we turn the location rate into a discrete update of all energy variables.
	\end{enumerate}
	This transformation is linear in~$\TA$.
		See \cref{figure:ISMETA} for an example of our construction.
	Our transformation clearly preserves the final energies for iET-IS-METAs, as well as the energy every time unit.
	This gives the following results.

	\begin{proposition}\label{proposition:IS-METAs}
		\begin{enumerate}
			\item For each $\existsWeakFull \in \{ \exists, \text{weak}, \text{full} \}$ and $\ENorETEN \in \{ \text{EN}, \text{ET-EN} \}$,
				$\existsWeakFull$-$\ENorETEN$-opacity can be decided in \twoEXPSPACE{} for positive iET-IS-METAs.
			\item For each $\existsWeakFull \in \{ \exists, \text{weak}, \text{full} \}$
			and $\ENorETEN \in \{ \text{EN}, \text{ET-EN} \}$,
			$\existsWeakFull$-$\ENorETEN$-opacity is decidable for iET-IS-ETAs.
			\item For each $\existsWeakFull \in \{ \exists, \text{weak}, \text{full} \}$, $\existsWeakFull$-DE-opacity can be decided in \EXPSPACE{} for positive iET-IS-ETAs\ifdefined\withParikhbyBlocks{}; and
			$\exists$-DE-opacity can be decided in \EXPSPACE{} for positive iET-IS-METAs\fi.
		\end{enumerate}
	\end{proposition}
	\begin{proof}
		\ifdefined\withParikhbyBlocks
			From \cref{theorem:fullweak-discrete-positive-guarded-ETAs,theorem:ET-EN-fullweak-discrete-positive-guarded-METAs,theorem:EN:discrete-ETAs,theorem:DE-Opacity,theorem:DE-Opacity-META}.
		\else
			From \cref{theorem:fullweak-discrete-positive-guarded-ETAs,theorem:ET-EN-fullweak-discrete-positive-guarded-METAs,theorem:EN:discrete-ETAs,theorem:DE-Opacity}.
		\fi
	\end{proof}
	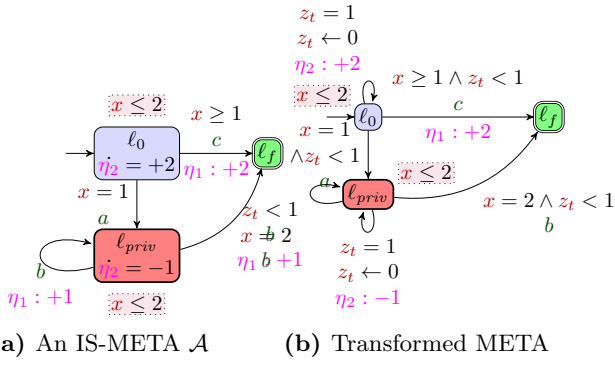
\begin{figure}[tb]
		\begin{subfigure}[b]{.35\columnwidth}
		\centering
			\scalebox{.8}{
			\begin{tikzpicture}[PTA, node distance=.8cm and 1.2cm]
				\node[location, initial] (l0) {$\locinit$ \\ $\dot{\enervari{2}} = +2$};
				\node[location, private, below =of l0, align=center] (lpriv) {$\locpriv$ \\ $\dot{\enervari{2}} = -1$ };
				\node[location, final, right =of l0] (lF) {$\locfinal$};

				\node[invariantNorth] at (l0.north) {$\clock \leq 2$};
				\node[invariantSouth] at (lpriv.south) {$\clock \leq 2$};

				\path (l0) edge node[left, align=center]{$\clock = 1$ \\ $\styleact{a}$} (lpriv);
				\path (l0) edge node[above, align=center]{$\clock \geq 1$ \\ $\styleact{c}$} node[below, align=center]{$\styleenergy{\enervar_1 : +2}$} (lF);
				\path (lpriv) edge[loop left] node[below, align=center]{$\styleact{b}$ \\ $\styleenergy{\enervar_1 : +1}$} (lpriv);
				\path (lpriv) edge[bend right] node[below right, align=center]{$\clock = 2$ \\ $\styleact{b}$} (lF);
			\end{tikzpicture}
			}

		\caption{An IS-META~$\TA$}
		\label{figure:IS-META}
		\end{subfigure}
		\hfill{}
		\begin{subfigure}[b]{.63\columnwidth}
		{\centering
			\scalebox{.8}{
			\begin{tikzpicture}[PTA, node distance=.8cm and 2.5cm]
				\node[location, initial] (l0) {$\locinit$};
				\node[location, private, below =of l0, align=center] (lpriv) {$\locpriv$};
				\node[location, final, right =of l0] (lF) {$\locfinal$};

				\node[invariantWest, yshift=1em] at (l0.west) {$\clock \leq 2$};
				\node[invariantEast, yshift=1em] at (lpriv.east) {$\clock \leq 2$};

				\path (l0) edge node[left, align=center]{$\clock = 1 $\\ $\land \clockTickOne < 1$ \\ $\styleact{a}$} (lpriv);
				\path (l0) edge[loop above] node[above left, align=center]{$\clockTickOne = 1$ \\ $\clockTickOne \assign 0$ \\ $\styleenergy{\enervar_2 : +2}$} (l0);
				\path (l0) edge node[above, align=center]{$\clock \geq 1 \land \clockTickOne < 1$ \\ $\styleact{c}$} node[below, align=center]{$\styleenergy{\enervar_1 : +2}$} (lF);

				\path (lpriv) edge[loop left] node[below left, align=center]{$\clockTickOne < 1$ \\ $\styleact{b}$ \\ $\styleenergy{\enervar_1 : +1}$} (lpriv);
				\path (lpriv) edge[loop below] node[below, align=center]{$\clockTickOne = 1$ \\ $\clockTickOne \assign 0$ \\ $\styleenergy{\enervar_2 : -1}$} (lpriv);
				\path (lpriv) edge[bend right] node[below right, align=center]{$\clock = 2 \land \clockTickOne < 1$ \\ $\styleact{b}$} (lF);
			\end{tikzpicture}
			}

		}
		\caption{Transformed META}
		\label{figure:ISMETA:transformed}
		\end{subfigure}
		\caption{Turning an IS-META into a META}
		\label{figure:ISMETA}
	\end{figure}

	Note that our transformation does not correctly encode the energy guards, hence guarded IS-METAs are left out.
	In addition, our transformation does not properly encode the energy elapse beyond the last integer time (which can also make the energy become negative in between two time units), hence  the iET assumption.
	Lifting these assumptions is the subject of ongoing works.

	\section{Conclusion}\label{section:conclusion}

	In this paper, we introduced the formalism of guarded METAs to express opacity properties over cyber-physical systems.
	While the general class is unsurprisingly undecidable, we exhibited a number of decidability results for discrete subclasses (\ie{} continuous models, where the energy can be incremented or decremented at integer or rational times); in addition, we exhibited first decidability results for models including continuous energy elapsing over multiple energy variables.
	We summarize results from \cref{section:discrete,section:DE} in \cref{table:summary-discrete,table:summary-DE}.
\begin{table*}[tb]

		{\centering\footnotesize
		\begin{tabular}{|c|c|c|}
			\hline
			\rowHeader{}
			\bfseries{}For each $\existsWeakFull \in \{ \exists, \text{weak}, \text{full} \}$ & \bfseries{}$\existsWeakFull$-EN-opacity & \bfseries{}$\existsWeakFull$-ET-EN-opacity \\
			\hline
			Discrete positive METAs & \cellDec{}\twoEXPSPACE{} (\cref{theorem:fullweak-discrete-positive-guarded-ETAs}) & \cellDec{}\twoEXPSPACE{} (\cref{theorem:ET-EN-fullweak-discrete-positive-guarded-METAs}) \\
			\hline
			Discrete positive guarded METAs & \cellDec{}\threeEXPSPACE{} (\cref{theorem:fullweak-discrete-positive-guarded-ETAs}) & \cellDec{}\threeEXPSPACE{} (\cref{theorem:ET-EN-fullweak-discrete-positive-guarded-METAs}) \\
			\hline
			Discrete (guarded) ETAs & \multicolumn{2}{c|}{\cellDec{}decidable (\cref{theorem:EN:discrete-ETAs,theorem:EN:discrete-guarded-ETAs})} \\
			\hline
			Discrete guarded METAs & \multicolumn{2}{c|}{\cellUndec{}undecidable (\cref{proposition:undecidability-2-energies})} \\
			\hline
		\end{tabular}

		}
		\caption{Decidability of (ET-)EN-opacity in discrete METAs}
		\label{table:summary-discrete}
	\end{table*}
	\begin{table*}[tb]

		{\centering\footnotesize
		\begin{tabular}{|c|c|c|c|c|}
			\hline
			\rowHeader{}
			& \multicolumn{2}{c|}{\bfseries{}DE-opacity} & \multicolumn{2}{c|}{\bfseries{}bDE-opacity} \\
			\hline
			& $\exists$ & Weak / full & $\exists$ & Weak / full \\
			\hline
			Discrete positive (guarded) ETAs & \multicolumn{2}{c|}{\cellDec{}\EXPSPACE{} (\cref{theorem:DE-Opacity})} & \multicolumn{2}{c|}{\cellDec{}\twoEXPSPACE{}* (\cref{theorem:bDE-Opacity})} \\
			\hline
			Discrete positive (guarded) METAs & \ifdefined\withParikhbyBlocks{}\cellDec{}decidable (\cref{theorem:DE-Opacity-META})\else\cellOpen{}\fi & \cellOpen{} & \multicolumn{2}{c|}{\cellDec{}\twoEXPSPACE{}* (\cref{theorem:bDE-Opacity})} \\
			\hline
			Discrete ETAs & \multicolumn{2}{c|}{\cellOpen{}} & \multicolumn{2}{c|}{\cellDec{}\twoEXPSPACE{} (\cref{theorem:bDE-Opacity-ETAs})} \\
			\hline
			Discrete METAs & \multicolumn{2}{c|}{\cellOpen{}} & \multicolumn{2}{c|}{\cellOpen{}} \\
			\hline
			Discrete guarded METAs & \multicolumn{2}{c|}{\cellUndec{}undecidable (\cref{proposition:undecidability-2-energies-DE})} & \multicolumn{2}{c|}{\cellOpen{}} \\
			\hline
		\end{tabular}

		}
		\caption{Decidability of (b)DE-opacity in discrete METAs {\scriptsize(*\EXPSPACE{} without energy guards)}}
		\label{table:summary-DE}
	\end{table*}

	Beyond the open cases in \cref{table:summary-discrete,table:summary-DE} and the lower bounds of the complexities, future works include observation of the sole execution time over guarded METAs, \ie{} where the energy, while unobservable, constrains the system.
	In addition, we will be interested in verification under energy uncertainty, demanding parametric verification techniques~\cite{BBFLMR21}.

	\begin{acks}
		We thank Engel Lefaucheux and Sarah Dépernet for preliminary discussions regarding timed opacity with energy.
		This work is partially supported by ANR BisoUS (ANR-22-CE48-0012)
		and by
		Agence de l'Innovation de Défense (AID) via Centre Interdisciplinaire Mers et Océan (CIMO) project 2025 CHRoMM.
	\end{acks}

	\newcommand{\CCIS}{Communications in Computer and Information Science}
	\newcommand{\ENTCS}{Electronic Notes in Theoretical Computer Science}
	\newcommand{\FAC}{Formal Aspects of Computing}
	\newcommand{\FundInf}{Fundamenta Informaticae}
	\newcommand{\FMSD}{Formal Methods in System Design}
	\newcommand{\IJFCS}{International Journal of Foundations of Computer Science}
	\newcommand{\IJSSE}{International Journal of Secure Software Engineering}
	\newcommand{\IPL}{Information Processing Letters}
	\newcommand{\JAIR}{Journal of Artificial Intelligence Research}
	\newcommand{\JLAP}{Journal of Logic and Algebraic Programming}
	\newcommand{\JLAMP}{Journal of Logical and Algebraic Methods in Programming} %
	\newcommand{\JLC}{Journal of Logic and Computation}
	\newcommand{\LMCS}{Logical Methods in Computer Science}
	\newcommand{\LNCS}{Lecture Notes in Computer Science}
	\newcommand{\RESS}{Reliability Engineering \& System Safety}
	\newcommand{\RTS}{Real-Time Systems}
	\newcommand{\SCP}{Science of Computer Programming}
	\newcommand{\SOSYM}{Software and Systems Modeling ({SoSyM})}
	\newcommand{\STTT}{International Journal on Software Tools for Technology Transfer}
	\newcommand{\TCS}{Theoretical Computer Science}
	\newcommand{\TOPLAS}{{ACM} Transactions on Programming Languages and Systems ({ToPLAS})}
	\newcommand{\ToPNoC}{Transactions on {P}etri Nets and Other Models of Concurrency}
	\newcommand{\TOSEM}{{ACM} Transactions on Software Engineering and Methodology ({ToSEM})}
	\newcommand{\TSE}{{IEEE} Transactions on Software Engineering}

		\bibliographystyle{alpha} %
	\bibliography{energy}
	\clearpage
	\appendix

	\section{Mathematical concepts}\label{appendix:maths}

	\begin{definition}[Ordered partition of $\setRgeqzero$]\label{definition:ordered-partition}
		We call \emph{ordered integer partition of~$\setRgeqzero$} (for short \emph{ordered partition of~$\setRgeqzero$}) a finite sequence of disjoint intervals $\subinterval_1, \dots, \subinterval_n$ such that:
		\begin{enumerate}
			\item $\bigcup_{j=1}^{n}\subinterval_j = \setRgeqzero$,
			\item all boundaries except the right boundary of $\subinterval_n$ belong to~$\setN$,
			\item each $\subinterval_j$ is non empty,
				and
			\item for all $j \in \{1, \dots ,n-1\}$, the right boundary of $\subinterval_j$ is the left boundary of $\subinterval_{j+1}$.
		\end{enumerate}
	\end{definition}

	\begin{example}
		The sequence $[0,2], (2,5), [5,10), [10, +\infty)$ is an ordered partition of~$\setRgeqzero$.
	\end{example}

	\section{Common automata concepts}\label{appendix:automata-concepts}
	\subsection{Formal definition of pushdown automata}\label{appendix:definition:pushdown}
	\begin{definition}\label{definition:pushdown}
	A \emph{pushdown automaton (PDA)} is a tuple $(\PDAStates, \Actions, \PDAstackAlphabet, \PDAstateinit, \initialStackSymbol, \PDAStatesFinal, \PDAEdges)$ where:
	\begin{enumerate}
		\item $\PDAStates$ is a finite set of states,
		\item $\Actions$ is a finite set of actions (``input alphabet''),
		\item $\PDAstackAlphabet$ is a finite set of stack symbols (``stack alphabet''),
		\item $\PDAstateinit \in \PDAStates$ is the initial state,
		\item $\initialStackSymbol \in \PDAstackAlphabet$ is the initial stack symbol,
		\item $\PDAStatesFinal \subseteq \PDAStates$ is the set of accepting states,
		\item $\PDAEdges$ is a finite set of edges $\PDAedge = (\PDAstate, \action, \stackSymbol, \PDAstate', \stackSymbol* )$ where $\PDAstate, \PDAstate' \in \PDAStates$ are the source and target states, $\action \in \Actions \cup \{ \silentaction \}$ is the transition symbol, and $\stackSymbol \in \PDAstackAlphabet$, $\stackSymbol* \in (\PDAstackAlphabet)^*$ are the symbols respectively popped from and pushed to the stack.
	\end{enumerate}
	\end{definition}
	\subsection{Region automaton}\label{appendix:regions}
We recall that the region automaton of timed automata is obtained by quotienting the set of clock valuations out by an equivalence relation $\simeq$ recalled below.
As an abuse of notation, we denote a TA $\TA = (\Actions, \LocSet, \locinit, \LocsPriv, \LocsFinal, \LabelFunc, \Clock, \invariant, \Edges)$, since the set of energy variables is empty, and all rates are~1.

Given a TA~$\TA$ and its set of clocks $\Clock$, we define $\LargestConstant : \Clock \rightarrow \setN$ the map that associates to a clock $\clock$ the greatest value to which the interpretations of $\clock$ are compared within the guards and invariants; if $\clock$ appears in no constraint, we set $\LargestConstant(\clock) = 0$.

Given $\alpha \in \setR$, we write $\intpart{\alpha}$ and $\fract{\alpha}$ for the integral and fractional parts of $\alpha$, respectively.

\begin{definition}[Equivalence relation $\simeq$ for valuations~\cite{AD94}]
	Let $\clockval$, $\clockval'$ be two clock valuations. %
	We say that $\clockval$ and $\clockval'$ are \emph{equivalent}, denoted by $\clockval \simeq \clockval'$ when, for each $\clock \in \Clock$, either $	\clockval (\clock) > \LargestConstant(\clock)$ and $\clockval' (\clock) > \LargestConstant(\clock)$ or the three following conditions hold:

	\begin{enumerate}
	\item $\intpart{\clockval (\clock)} = \intpart{\clockval' (\clock)}$;
	\item $\fract{\clockval(\clock)} = 0$ if and only if $\fract{\clockval'(\clock)} = 0$; and
	\item for each $\clocky \in \Clock$, $\fract{\clockval(\clock)} \leq  \fract{\clockval(\clocky)}$  if and only if $\fract{\clockval'(\clock)} \leq  \fract{\clockval'(\clocky)}$.
	\end{enumerate}
\end{definition}

The equivalence relation is extended to the concrete states of~$\TA$: let $\concstate = (\loc, \clockval)$ and $\concstate' = (\loc', \clockval')$ be two concrete states in~$\TA$, then
\(\concstate \simeq \concstate' \text{ if and only if } \loc = \loc' \text{ and } \clockval \simeq \clockval'. \)

The equivalence class of a valuation $\clockval$ is denoted $[ \clockval ]$ and is called a \emph{clock region}, and the equivalence class of a concrete state $\concstate = (\loc, \clockval)$ is denoted $[ \concstate ]$ and called a \emph{region} of~$\TA$.
In other words, a region is a pair made of a location and of a clock region.
Clock regions are denoted by the enumeration of the constraints defining the equivalence class. Thus, values of a clock $\clock$ that go beyond $\LargestConstant(\clock)$ are merged and described in the regions by the inequality ``$\clock > \LargestConstant(\clock)$''.

The set of regions of~$\TA$ is denoted by~$\Regions{\TA}$.
These regions are of finite number: this allows us to construct a finite ``untimed'' regular automaton, the \emph{region automaton}~$\RegionAutomaton{\TA}$.
Locations of $\RegionAutomaton{\TA}$ are regions of~$\TA$, and the transitions of $\RegionAutomaton{\TA}$ convey the reachable valuations associated with each concrete state in~$\TA$.

To formalize the construction, we need to transform discrete and time-elapsing transitions of~$\TA$ into transitions between the regions of~$\TA$. To do that, we define a ``time-successor'' relation that corresponds to time-elapsing transitions.

\begin{definition}[Time-successor relation~\cite{ALM23}]
	Let $\region = (\loc, [\clockval]), \region' = (\loc', [\clockval']) \in \Regions{\TA}$.
	We say that $\region'$ is a time-successor of $\region$ when $\region \neq \region'$, $\loc = \loc'$ and for each concrete state $(\loc, \clockval)$ in $\region$, there exists $d \in \setRgeqzero$ such that $(\loc, \clockval + d)$ is in~$\region'$ and for all $d' < d, (\loc, \clockval + d') \in \region \cup \region'$.
\end{definition}

A region $\region = (\loc, [\clockval])$ is \emph{unbounded} when, for all $\clock$ in $\Clock$ and all $\clockval' \in [\clockval]$, $\clockval'(\clock) > \LargestConstant(\clock)$.

\begin{definition}[Region automaton \cite{AD94}\label{definition:region-automaton}]
	Given a TA $\TA = (\Actions, \LocSet, \locinit, \LocsPriv, \LocsFinal, \LabelFunc, \Clock, \invariant, \Edges)$, the region automaton is the
	tuple $\RegionAutomaton{\TA} = (\Actions_{\Regions{}}, \Regions{}, \region_0, {\Regions{}}_f, \Edges_{\Regions{}})$ where
	\begin{enumerate}%
	\item $\Actions_{\Regions{}} = \Actions \cup \lbrace \silentaction \rbrace$;
	\item ${\Regions{}} = \Regions{\TA}$;
	\item $\region_0 = [\qinit]$;
	\item ${\Regions{}}_f$ is the set of regions whose first component is a final location~$\in \LocsFinal$;
	\item $\Edges_{\Regions{}}$ is defined as follows:
	\begin{itemize}%
		\item \emph{(discrete transitions)} For every $\region = (\loc, [\clockval])$ with $\loc \notin \LocsFinal$, $\region' = (\loc', [\clockval']) \in \Regions{\TA}$ and $\action \in \Actions \cup \lbrace \silentaction \rbrace$:
	\( (\region, \action, \region') \in \Edges_{\Regions{}} \text{ if } \exists \clockval'' \in [\clockval], \exists \clockval''' \in [\clockval'], \big( (\loc, \clockval'') , \edge , (\loc', \clockval''') \big) \in \TransConcrete \),
	with $\edge = (\loc, \guard, \action, \resets, \loc') \in \Edges$ for some guard~$\guard$ and some clock set~$\resets$.

	\item \emph{(delay transitions)}
		For every $\region = (\loc, [\clockval])$ with $\loc \notin \LocsFinal$, $\region' \in \Regions{\TA}$:
	\( (\region, \silentaction, \region') \in \Edges_{\Regions{}} \text{ if } \region' \text{ is a time-successor of } \region \text{ or if } \region=\region'\text{ is unbounded.}\)
	\end{itemize}%
	\end{enumerate}%
\end{definition}

A \emph{run} of~$\RegionAutomaton{\TA}$ is an alternating sequence of regions of~$\RegionAutomaton{\TA}$ and actions starting from the initial region $\regioni{0}$ and ending in a final region, of the form
$\regioni{0}, \action_0, \regioni{1}, \action_1, \ldots \regioni{n-1}, \action_{n-1}, \regioni{n}$ for some $n \in \setN$,
with $\regioni{n} \in {\Regions{}}_f$ and for each $0 \leq i \leq n-1$, $\regioni{i} \notin {\Regions{}}_f$, and $(\regioni{i}, \action_i, \regioni{i+1}) \in \Edges_{\Regions{}}$.

$\RegionAutomaton{\TA}$ is a non-deterministic finite automaton (NFA); its language is defined as usual as the set of accepting words, \ie{} as the sequence of actions ($\silentaction$~excluded) contained in a run.
The (untimed) language of~$\RegionAutomaton{\TA}$ is denoted by $\Language(\RegionAutomaton{\TA})$.

	\section{Formal definitions of opacity}\label{appendix:definitions}
	\subsection{ET-EN-opacity}\label{appendix:definitions:ETEN}
	\begin{definition}[ET-EN-opacity]\label{def:ET-EN-opacity}
		A guarded META~$\TA$ is \emph{fully ET-EN-opaque} when, for each $d \in \setRgeqzero$ and~$\enerval \in \setRgeqzero^{\cardinality{\Energies}}$,
		there exists a private run~$\varrun$ such that $\runduration{\varrun} = d$ and $\finalEnergy(\varrun) = \enerval$
		iff
		there exists a public run~$\varrun'$ such that $\runduration{\varrun'} = d$ and $\finalEnergy(\varrun') = \enerval$.

		A guarded META~$\TA$ is \emph{weakly ET-EN-opaque} when for each private run~$\varrun$, there exists a public run~$\varrun'$ such that $\runduration{\varrun} = \runduration{\varrun'}$ and $\finalEnergy(\varrun) = \finalEnergy(\varrun')$.

		A guarded META~$\TA$ is \emph{$\exists$-ET-EN-opaque} when there exist a private run~$\varrun$ and a public run~$\varrun'$ such that $\runduration{\varrun} = \runduration{\varrun'}$ and $\finalEnergy(\varrun) = \finalEnergy(\varrun')$.
	\end{definition}
	\begin{remark}
		Note that verifying \existsWeakFull-ET-opacity and \existsWeakFull-EN-opacity is not a sufficient condition for verifying \existsWeakFull-ET-EN opacity.
		Indeed, a model can be \existsWeakFull-ET-opaque and \existsWeakFull-EN-opaque but not \existsWeakFull-ET-EN opaque.
		The guarded META in \cref{figure:example-not-ET-EN} is such that:
		\begin{itemize}
			\item $\PrivDurVisit{\TA} = \PubDurVisit{\TA} = [0,10]$
			\item $\PrivEnerVisit{\TA} = \PubEnerVisit{\TA} = [0,10]$
		\end{itemize}
		It is therefore both fully-ET-opaque and fully-EN-opaque.
		However, there exists a private run such that no public run has the same execution time and the same final energy (for example, execution time of~4 and final energy at~6).
		Therefore, it is neither fully nor weakly-ET-EN-opaque.
	\end{remark}

	\begin{figure}[tb]
	\centering

			\centering
			\begin{tikzpicture}[PTA, node distance=1cm and 2cm]
				\node[location, initial] (l0) {$\locinit$};
				\node[location, right = of l0] (l1) {$\loc_1$ \\ $\styleenergy{\enervar:+1} $};
				\node[location, private, below = of l1] (lp) {$\locpriv$ \\ $\styleenergy{\enervar:-1} $};
				\node[location, final, right = of l1] (lF) {$\locfinal$};

				\node[invariantNorth] at (lp.north) {$\clock \leq 10$};
				\node[invariantNorth] at (l1.north) {$\clock \leq 10$};

				\path (l0) edge[bend right] node[ below left, align=center]{$ \clock = 0$ \\ $\styleact{a}$ \\ $\styleenergy{\enervar:+10} $} (lp);

				\path (l0) edge  node[above, align=center]{$\clock = 0$ \\ $\styleact{b}$} (l1);

				\path (l1) edge node[above, align=center] {$\styleact{c}$}(lF);

				\path (lp) edge[bend right]  node[below right, align=center] {$\styleact{c}$} (lF);

			\end{tikzpicture}

		\caption{Guarded META ET-opaque and EN-opaque but not fully-ET-EN-opaque}
		\label{figure:example-not-ET-EN}

	\end{figure}
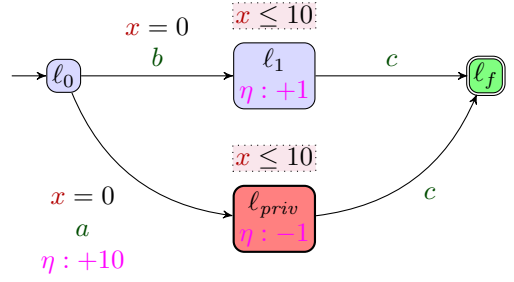

	\subsection{Formal definition of function $\EnerLevelBufferT$}\label{appendix:definition:bEL}

		We first need to embed the absolute time of each state of a run, and project onto the states, \ie{} removing the transitions.
		Given a run \(\varrun = (\loci{0}, \clockval_0, \enerval_0), (\paramd_0, \edgei{0}), (\loci{1}, \clockval_1, \enerval_1), \ldots, (\loci{n}, \clockval_n, \enerval_n)\), let $\absTimeRun(\varrun) = (\loci{0}, \clockval_0, \enerval_0, t_0), (\loci{1}, \clockval_1, \enerval_1, t_1), \ldots, (\loci{n}, \clockval_n, \enerval_n, t_n)$, where \( t_i = \sum_{j=0}^{i-1} \paramd_j \).

		We also define a function $\destutter$ that, given a non-empty sequence~$\seqrun$ of quadruples made of a location, a clock valuation, an energy valuation and an absolute time (typically the result of $\absTimeRun(\varrun)$), returns the subsequence of~$\seqrun$ where only the quadruples that modify at least one energy variable are kept.
		The formal definition is given in \cref{figure:definition:destutter}.

		\begin{figure*}
			{\centering
			\[ \destutter(\seqrun) = \left\{
			\begin{array}{ll}
				\seqrun & \mbox{if } \seqrun = (\loci{0}, \clockval_0, \enerval_0, t_0) \text{ or } \seqrun = \emptyseq\\

				\destutter \big((\loci{i}, \clockval_i, \enerval_i, t_i) , \seqrun'\big) & \mbox{if } \seqrun = (\loci{i}, \clockval_i, \enerval_i, t_i), (\loci{j}, \clockval_j, \enerval_j, t_j) , \seqrun' \text{ and }\enerval_i = \enerval_j\\

				(\loci{i}, \clockval_i, \enerval_i, t_i) , \destutter\big((\loci{j}, \clockval_j, \enerval_j, t_j) , \seqrun'\big) & \mbox{if } \seqrun = (\loci{i}, \clockval_i, \enerval_i, t_i), (\loci{j}, \clockval_j, \enerval_j, t_j) , \seqrun' \text{ and }\enerval_i \neq \enerval_j\\
			\end{array}
			\right.
			\]

			 }
			\caption{Formal definition of $\destutter$}
			\label{figure:definition:destutter}
		\end{figure*}
		Then, given a non-empty sequence~$\seqrun$ of quadruples made of a location, a clock valuation, an energy valuation and an absolute time, and given a time~$\abstime \in \setN \setminus \{ 0 \}$, we define the \emph{subsequence in the interval ending in~$\abstime$, projected onto the energy variables}, denoted by $\subseqAndProject(\seqrun, \abstime)$, as in \cref{figure:definition:subseqInterval}.

		\begin{figure*}
			{\centering
		\[ \subseqAndProject(\seqrun, \abstime) = \left\{
		\begin{array}{ll}

			(\enerval_{0}, \dots ,\enerval_{k} ) \mbox{ where } t_{k} \leq 1 < t_{k+1} & \mbox{ if } \abstime = 1\\

			(\enerval_{i}, \dots ,\enerval_{k} ) \mbox{ where } t_{i-1} < \abstime - 1 \leq t_{i}  \land t_{k} \leq \abstime < t_{k+1} & \text{ otherwise }
		\end{array}
		\right.
		\]

			 }
			\caption{Formal definition of $\subseqAndProject$, with $\seqrun = (\loci{0}, \clockval_0, \enerval_0, t_0), \ldots, (\loci{n}, \clockval_n, \enerval_n, t_n)$} %
			\label{figure:definition:subseqInterval}
		\end{figure*}

		We can finally define function $\EnerLevelBufferT$.

		\begin{definition}\label{definition:bEL}
			Let $\varrun$ be a run, and let $\abstime \in \setN \setminus \{ 0 \}$ be an absolute time.
			The sequence of energy valuations $\EnerLevelBufferT(\varrun, \abstime)$ is defined by

			\[ \EnerLevelBufferT(\varrun, \abstime) = \subseqAndProject \Big( \destutter \big(\absTimeRun(\varrun) \big) , \abstime \Big) \text{.} \]
		\end{definition}
	\section{Definitions and proofs of \cref{section:generic}}\label{appendix:proofs:section:generic}
	\subsection{Formal definitions of $\Apub{\TA}$ and $\Apriv{\TA}$}\label{appendix:def:ApubApriv}
	\begin{definition}[$\Apub{\TA}$]\label{definition:Apub}
		Let $\TA = (\Actions, \LocSet, \locinit, \LocsPriv, \LocsFinal, \LabelFunc, \Clock, \Energies, \invariant, \rates, \Edges)$ be a guarded META.
		Then, $\Apub{\TA}$ is defined as
		$(\Actions, \LocSet, \locinit, \LocsPriv, \LocsFinal, \LabelFunc, \Clock, \Energies, \invariant, \rates, \Edges')$,
		where $\Edges' = \Edges \setminus \{ (\loc, \guard, \action, \resets, \enerupdates, \loc') \mid \loc \in \locpriv \lor \loc' \in \locpriv \}$.
	\end{definition}
	\begin{definition}[$\Apriv{\TA}$]\label{definition:Apriv}
		Let $\TA = (\Actions, \LocSet, \locinit, \LocsPriv, \LocsFinal, \LabelFunc, \Clock, \Energies, \invariant, \rates, \Edges)$ be a guarded META.
		Then, $\Apriv{\TA}$ is defined as
		$(\Actions, \LocSet', \locinit', \LocsPriv', \LocsFinal', \LabelFunc', \Clock, \Energies, \invariant', \rates', \Edges')$, where
		\begin{itemize}
			\item $\LocSet' = \{ \loc^{\bar{V}} \mid \loc \in \LocSet \} \cup \{ \loc^{V} \mid \loc \in \LocSet \}$,
			\item $\locinit' = \locinit^{\bar{V}}$,
			\item $\LocsFinal' = \{ \locfinal^{V} \mid \locfinal \in \LocsFinal \}$,
			\item $\forall \loc \in \LocSet, \LabelFunc'(\loc^{V}) = \LabelFunc'(\loc^{\bar{V}}) = \LabelFunc(\loc)$,
			\item $\forall \loc \in \LocSet, \invariant'(\loc^{V}) = \invariant'(\loc^{\bar{V}}) = \invariant(\loc)$,
			\item $\forall \loc \in \LocSet, \rates'(\loc^{V}) = \rates'(\loc^{\bar{V}}) = \rates(\loc)$,
			\item $\Edges' = \big\{ (\loc_1^{\bar{V}}, \guard, \action, \resets, \enerupdates, \loc_2^{\bar{V}}) \mid (\loc_1, \guard, \action, \resets, \enerupdates, \loc_2) \in \Edges \land \loc_2 \notin \LocsPriv \big\} \cup \big\{ (\loc_1^{V}, \guard, \action, \resets, \enerupdates, \loc_2^{V}) \mid (\loc_1, \guard, \action, \resets, \enerupdates, \loc_2) \in \Edges \big\} \cup \big\{ (\loc_1^{\bar{V}}, \guard, \action, \resets, \enerupdates, \loc_2^{V}) \mid (\loc_1, \guard, \action, \resets, \enerupdates, \loc_2) \in \Edges \land \loc_2 \in \LocsPriv \big\}$.
		\end{itemize}
	\end{definition}
	\subsection{Proof of \cref{proposition:removing-guards-discrete}}\label{appendix:proof:proposition:removing-guards-discrete}
	\propRemovingGuardsDiscrete*
	\begin{proof}
		We first give a construction for a discrete positive ETA, and then generalise it to discrete positive METAs.
		Let $\maxConstantEnergy$ be the maximum constant appearing in energy guards (and invariants).
		Let us build~$\TA'$ as follows:
		\begin{enumerate}
			\item we build $\maxConstantEnergy+2$ copies of~$\TA$, where the first $\maxConstantEnergy+1$ copies encode the exact value of the energy from~0 to~$\maxConstantEnergy$, while the last one encodes $> \maxConstantEnergy$.
			\item in each copy~$i$, we redirect each transition from $\loc$ to~$\loc'$ where the energy increases by~$m$ units to location $\loc'$ of the copy $i+m$---unless $i+m > \maxConstantEnergy$, in which case we redirect to the last copy;
			\item for each transition with energy guard $\enervar \compop m$ in copy~$0 \leq i \leq \maxConstantEnergy$, we either remove the guard whenever $i \compop m$ holds, or we remove the transition otherwise (the case for the last copy is similar).
		\end{enumerate}
		This transformation clearly preserves the semantics of~$\TA$.
		This transformation extends to multiple energy variables in a straightforward manner, by multiplying the number of location by~$(\maxConstantEnergy+2)^{|\Energies|}$, which leads to an exponential blowup.
		\cref{figure:example-gMETA} illustrates this transformation.
	\end{proof}
	\begin{figure}[tb]
		\begin{subfigure}[b]{.3\textwidth}
			\centering
			\begin{tikzpicture}[PTA, node distance = 0.5cm and 2cm]

				\node[location, initial] (l0) {$\locinit$};
				\node[location, right = of l0] (l1) {$\loc_1$};
				\node[location, private, below = of l1, align=center] (lp) {$\locpriv$};
				\node[location, final, right = of l1] (lF) {$\locfinal$};

				\path (l0) edge[loop above] node[above, align=center]{ $\styleenergy{\enervar_1 \geq 1}$ \\ $\styleact{d}$ \\ $\styleenergy{\enervar_1 : +1}$} (l0);

				\path (l0) edge[bend right] node[align=center]{$\styleact{a}$} (lp);

				\path (l0) edge[bend left] node[above, align=center]{ $\styleact{c}$ \\ $\styleenergy{\enervar_1 : +1}$} (l1);

				\path (l1) edge[bend left] node[above, align=center]{ $\styleact{d}$} (l0);

				\path (l1) edge node[above, align=center]{ $\styleact{c}$} node[below, align=center]{$\styleenergy{\enervar_2 : +1}$} (lF);

				\path (lp) edge[loop below] node[below, align=center]{$\styleenergy{\enervar_1 \leq 1}$ \\ $\styleact{b}$ \\ $\styleenergy{\enervar_1 : +1}$ } (lp);
				\path (lp) edge[bend right] node[below right, align=center]{ $\styleact{b}$ \\ $\styleenergy{\enervar_2 : +1}$} (lF);
			\end{tikzpicture}

			\caption{Example with $\maxConstantEnergy=1$}

		\end{subfigure}
		\hfill{}
		\begin{subfigure}[b]{.45\textwidth}
			\centering
			\scalebox{.8}{
			\begin{tikzpicture}[PTA, node distance = 1cm and 2cm]

				\node[draw, rectangle, fill=lightgray!20, fit={(-1,0) (6,-2)}, inner sep=0.5cm, label=right:Copy~0] {};

				\node[draw, rectangle, fill=lightgray!20, fit={(-1,-3.75) (6,-5.75)}, inner sep=0.5cm, label=right:Copy~1] {};

				\node[draw, rectangle, fill=lightgray!20, fit={(-1,-7.5) (6,-9.5)}, inner sep=0.5cm, label=right:Copy $> 1$] {};

				\node[location, initial] (l0-0) {$\locinit^0$};
				\node[location, right = of l0-0] (l1-0) {$\loc_1^0$};
				\node[location, private, below = of l1-0, align=center] (lp-0) {$\locpriv^0$};
				\node[location, final, right = of l1-0] (lF-0) {$\locfinal^0$};

				\node[location, below =of l0-0, yshift=-7em] (l0-1) {$\locinit^1$};
				\node[location, right = of l0-1] (l1-1) {$\loc_1^1$};
				\node[location, private, below = of l1-1, align=center] (lp-1) {$\locpriv^1$};
				\node[location, final, right = of l1-1] (lF-1) {$\locfinal^1$};

				\node[location, below = of l0-1, yshift=-7em] (l0-2) {$\locinit^2$};
				\node[location, right = of l0-2] (l1-2) {$\loc_1^2$};
				\node[location, private, below = of l1-2, align=center] (lp-2) {$\locpriv^2$};
				\node[location, final, right = of l1-2] (lF-2) {$\locfinal^2$};

				\path (l0-0) edge node[right, align=center]{$\styleact{a}$} (lp-0);
				\path (l0-1) edge node[right, align=center]{$\styleact{a}$} (lp-1);
				\path (l0-2) edge[bend right] node[right, align=center]{$\styleact{a}$} (lp-2);

				\path (l0-0) edge[bend right] node[above, align=center]{ $\styleact{c}$ \\ $\styleenergy{\enervar_1 : +1}$} (l1-1);
				\path (l0-1) edge node[above, align=center]{ $\styleact{c}$ \\ $\styleenergy{\enervar_1 : +1}$} (l1-2);
				\path (l0-2) edge[bend left] node[above, align=center]{ $\styleact{c}$ \\ $\styleenergy{\enervar_1 : +1}$} (l1-2);

				\path (l1-0) edge node[above, align=center]{ $\styleact{d}$} (l0-0);
				\path (l1-1) edge node[above, align=center]{ $\styleact{d}$} (l0-1);
				\path (l1-2) edge[bend left] node[above, align=center]{ $\styleact{d}$} (l0-2);

				\path (l1-0) edge[bend left]  node[right, align=center]{$\styleact{c}$ \\$\styleenergy{\enervar_2 : +1}$} (lF-1);
				\path (l1-1) edge[bend left]  node[right, align=center]{$\styleact{c}$ \\$\styleenergy{\enervar_2 : +1}$} (lF-2);
				\path (l1-2) edge  node[below, align=center]{$\styleact{c}$ \\$\styleenergy{\enervar_2 : +1}$} (lF-2);

				\path (lp-0) edge[bend left] node[below left , align=center, yshift = -2em, xshift = 1em]{$\styleact{b}$ \\ $\styleenergy{\enervar_1 : +1}$ } (lp-1);
				\path (lp-1) edge[bend left] node[below left, align=center, yshift = -2em, xshift = 1em]{ $\styleact{b}$ \\ $\styleenergy{\enervar_1 : +1}$ } (lp-2);

				\path (lp-0) edge node[above, align=center]{ $\styleact{b}$ \\ $\styleenergy{\enervar_2  +1}$} (lF-1);
				\path (lp-1) edge node[above, align=center]{ $\styleact{b}$ \\ $\styleenergy{\enervar_2 : +1}$} (lF-2);
				\path (lp-2) edge node[below, align=center]{ $\styleact{b}$ \\ $\styleenergy{\enervar_2 : +1}$} (lF-2);

				\path (l0-1) edge node[left, align=center]{ $\styleact{d}$ \\ $\styleenergy{\enervar_1 : +1}$} (l0-2);
				\path (l0-2) edge[loop below] node[below, align=center]{$\styleact{d}$ \\ $\styleenergy{\enervar_1 : +1}$} (l0-2);

			\end{tikzpicture}
			}

			\caption{Transformation}

		\end{subfigure}
		\caption{Exemplifying the energy guards removal}
		\label{figure:example-gMETA}
	\end{figure}
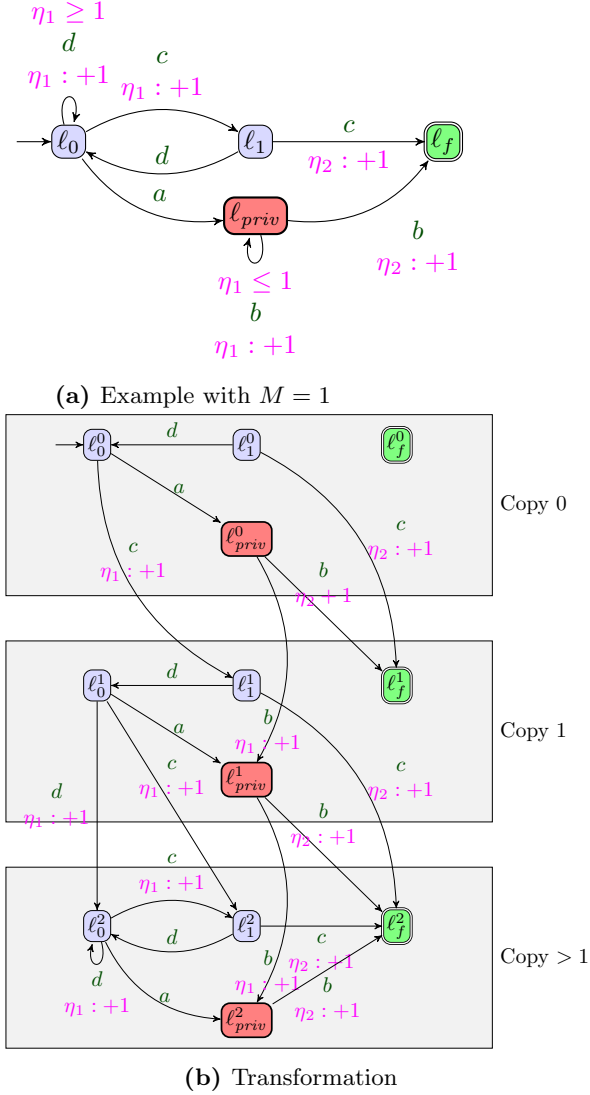
	\section{Proofs of \cref{section:discrete}}
	\subsection{Proof of \cref{proposition:undecidability-2-energies}}\label{appendix:proof:proposition:undecidability-2-energies}
\subsubsection{Two-counter machines}\label{ss:2CM}
Our subsequent undecidability proof works by reduction from
the halting problem for 2-counter machines.

We borrow the following material from~\cite{ALR22}.
A deterministic 2-counter machine (``2CM'')~\cite{Minsky67} has two non-negative counters $\Counter_1$ and~$\Counter_2$, a finite number of states and a finite number of transitions, which can be of the form (for $\counterGenericIndex \in \{ 1, 2\}$):
	\begin{enumerate}
		\item ``when in state~$\cms_i$, increment~$\Counter_\counterGenericIndex$ and go to~$\cms_j$''; or
		\item ``when in state~$\cms_i$, if $\Counter_\counterGenericIndex = 0$ then go to~$\cms_k$, otherwise decrement $\Counter_\counterGenericIndex$ and go to~$\cms_j$''.
	\end{enumerate}

The 2CM starts in state~$\cms_0$ with the counters set to~0.
The machine follows a deterministic transition function, meaning for each combination of state and counter conditions, there is exactly one action to take.
The \emph{halting problem} consists in deciding whether some distinguished state called \cmshalt{} can be reached or not.
This problem is known to be undecidable~\cite{Minsky67}.
\subsubsection{Undecidability of (constant-time) reachability in discrete guarded METAs}\label{appendix:lemma:undecidability-META-2-1}
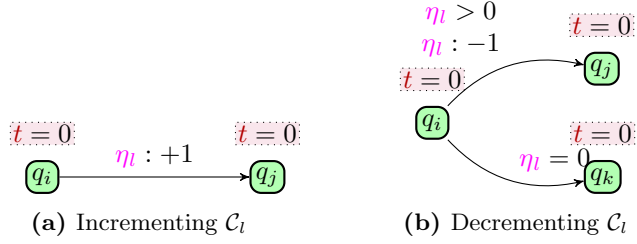
\begin{figure*}[tb]
	{\centering
	\begin{subfigure}[b]{.45\columnwidth}
	\centering
	\scalebox{\generalFigsScaleFactor}{
			\begin{tikzpicture}[PTA, node distance=2.5cm]
			\node[location2CM] (li) {$\cmspta_i$};
			\node[location2CM, right=of li] (lj) {$\cmspta_j$};

			\node[invariantNorth] at (li.north) {$\clockCMt = 0$};
			\node[invariantNorth] at (lj.north) {$\clockCMt = 0$};

			\path (li) edge node[above, align=center]{$\styleenergy{\enervar_\counterGenericIndex} : + 1$}  (lj);
			\end{tikzpicture}
	}

	\caption{Incrementing~$\Counter_\counterGenericIndex$}
	\label{figure:2CM:increment}
	\end{subfigure}
	\hspace{3em}
	\begin{subfigure}[b]{.45\columnwidth}
	\centering
	\scalebox{\generalFigsScaleFactor}{
			\begin{tikzpicture}[PTA, node distance=2.5cm]
			\node[location2CM] (li) {$\cmspta_i$};
			\node[location2CM, above right=of li, yshift=-1.5cm] (lj) {$\cmspta_j$};
			\node[location2CM, below right=of li, yshift=+1.5cm] (lk) {$\cmspta_k$};

			\node[invariantNorth] at (li.north) {$\clockCMt = 0$};
			\node[invariantNorth] at (lj.north) {$\clockCMt = 0$};
			\node[invariantNorth] at (lk.north) {$\clockCMt = 0$};

			\path (li) edge[bend left] node[align=center]{$\styleenergy{\enervar_\counterGenericIndex} > 0$ \\ $\styleenergy{\enervar_\counterGenericIndex} : - 1$}  (lj);
			\path (li) edge[bend right] node[align=center]{$\styleenergy{\enervar_\counterGenericIndex} = 0$}  (lk);
			\end{tikzpicture}
	}

	\caption{Decrementing~$\Counter_\counterGenericIndex$}
	\label{figure:2CM:decrement}
	\end{subfigure}

	}
	\caption{Encoding a 2CM as a discrete guarded META}
	\label{figure:2CM}
\end{figure*}

We first prove an intermediate lemma; the construction is straightforward, and similar to the construction in the proof of \cite[Proposition~1]{BDFP04}, that was given in a slightly different context (updatable timed automata).

\begin{lemma}\label{lemma:undecidability-META-2-1}
	Reachability in discrete METAs is undecidable, with two energy variables and a single clock.
\end{lemma}
\begin{proof}
	We reduce from the halting problem of 2-counter machines, which is undecidable~\cite{Minsky67}.
	Given a 2CM~$\calM$, we encode it as a discrete guarded META~$\TA$ over two energy variables $\{\enervar_1, \enervar_2 \}$ and a single clock~$\clockCMt$.

	Each state $\cms_i$ of the machine is encoded as a location of the discrete guarded META, which we call~$\cmspta_i$.
	The encoding is such that energy $\enervar_\counterGenericIndex$ directly encodes the value of $\Counter_\counterGenericIndex$.
	The entire computation is done over 0-time, \ie{} the unique clock only serves as a global invariant $\clockCMt = 0$.
	See \cref{figure:2CM:increment,figure:2CM:decrement} for the META fragment encoding increment, and zero-test and decrement, respectively.

	Obviously, $\TA$ can reach~$\lochalt$ iff $\calM$ can reach~$\cmshalt$.
	Since reachability of~$\cmshalt$ in~$\calM$ is undecidable, then reachability of~$\lochalt$ in~$\TA$ is undecidable.
\end{proof}
\begin{corollary}\label{corollary:undecidability-constant-time-META-2-1}
	Constant-time reachability in discrete METAs is undecidable, with two energy variables and a single clock.
\end{corollary}
\begin{proof}
	From the fact the construction in the proof of \cref{lemma:undecidability-META-2-1} works in 0-time.
	Any other constant can be achieved by delaying the start of the simulation of the 2-counter machine, thanks to the (unique) clock.
\end{proof}
\subsubsection{Opacity is harder than reachability for discrete METAs}\label{appendix:hardness:opacity-reachability-META}
\begin{figure*}[tb]
	{\centering
	\begin{tikzpicture}[PTA]

		\node[location, initial] (l0) at (-2, 0) {$\locinit$};
		\node[location] (lf) at (+1.8, 0) {$\locfinal$};
		\node[cloud, cloud puffs=15.7, cloud ignores aspect, minimum width=5cm, minimum height=2cm, align=center, draw] (cloud) at (0cm, 0cm) {$\TA$};

		\node[location, private] (lpriv) at (0, -2) {$\locpriv$};
		\node[location] (l1) at (+4, 0) {$\loci{1}$};
		\node[location, final] (lf') at (+6, -1.2) {$\locfinal'$};

		\path
		(l0) edge[bend right] (lpriv)
		(lf) edge (l1)
		(lpriv) edge[out=0,in=200] (lf')
		(l1) edge[bend left] (lf')
		(l1) edge[loop above] node[align=center] {$\styleenergy{\enervari{i}} : + 1$} (l1)
		(l1) edge[loop below] node[align=center] {$\styleenergy{\enervari{i}} : - 1$} (l1)
		(lpriv) edge[loop above] node[right, align=center] {$\styleenergy{\enervari{i}} :+ 1$} (lpriv)
		;

	\end{tikzpicture}

	}
	\caption{Solving reachability using EN-opacity for discrete METAs}
	\label{figure:EN-opacity-harder-than-reachability-discrete-ETAs}
\end{figure*}
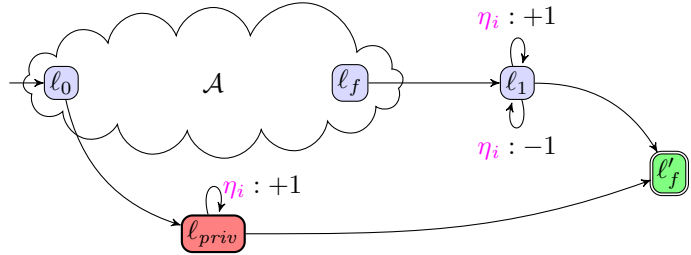
\begin{proposition}\label{proposition:EN-opacity-harder-reachability-discrete-METAs}
	For each $\existsWeakFull \in \{ \exists, \text{weak}, \text{full} \}$,
	reachability in discrete METAs can be solved using $\existsWeakFull$-EN-opacity.
\end{proposition}
\begin{proof}
	Let $\TA$ be a discrete META with initial location $\locinit$; let $\locfinal$ be one of its locations.
	We let $\TA'$ be the discrete META extending~$\TA$ as follows:
	\begin{itemize}
		\item we add a fresh location $\locpriv$ reachable from~$\locinit$ without guard;
		\item we add a fresh location $\loci{1}$ reachable from~$\locfinal$ without guard;
		\item we add a fresh location $\locfinal'$ reachable from~$\loci{1}$ and from~$\locpriv$ without guard;
		\item for each energy variable $\enervari{i}$ ($1 \leq i \leq |\Energies|$), we add a self-loop over~$\locpriv$ incrementing~$\enervari{i}$;
		\item for each energy variable $\enervari{i}$, we add two self-loops over~$\loci{1}$, one incrementing and one decrementing~$\enervari{i}$.
	\end{itemize}
	$\{ \locpriv \}$ is the private location set of~$\TA'$, while $\{ \locfinal' \}$ is its final location set.
	\cref{figure:EN-opacity-harder-than-reachability-discrete-ETAs} exemplifies this construction, where $\enervari{i} : + 1$ denotes $|\Energies|$ distinct self-loops incrementing one energy variable each.

	First note that $\PrivEnerVisit{\TA} = \setN^\Energies$, \ie{} the final energy can be any integer value for each energy variable.
	Furthermore, if $\locfinal$ is reachable, then $\PubEnerVisit{\TA} = \setN^\Energies$ (note that the energy can be arbitrary when reaching~$\locfinal$, but the self-loops over~$\loci{1}$ ensure that any value, smaller or larger, can be reached); in that case, all $\exists$, weak and full EN-opacity hold.
	Conversely, if $\locfinal$ is not reachable, then $\PubEnerVisit{\TA} = \emptyset$; in that case, none of the opacity definitions hold in~$\TA'$.

	As a consequence, for each $\existsWeakFull \in \{ \exists, \text{weak}, \text{full} \}$, $\existsWeakFull$-EN-opacity holds in~$\TA'$ iff $\locfinal$ is reachable in~$\TA$.
\end{proof}
\begin{figure*}[tb]
	{\centering
	\begin{tikzpicture}[PTA]

		\node[location, initial] (l0) at (-2, 0) {$\locinit$};
		\node[location] (lf) at (+1.8, 0) {$\locfinal$};
		\node[cloud, cloud puffs=15.7, cloud ignores aspect, minimum width=5cm, minimum height=2cm, align=center, draw] (cloud) at (0cm, 0cm) {$\TA$};

		\node[location, private] (lpriv) at (0, -2) {$\locpriv$};
		\node[location] (l1) at (+4, 0) {$\loci{1}$};
		\node[location] (l2) at (+4, +2) {$\loci{2}$};
		\node[location, final] (lf') at (+6, -1.2) {$\locfinal'$};

		\path
		(l0) edge[bend right] (lpriv)
		(l0) edge[bend angle=40, bend left] node[align=center]{$\clock > \maxConstantTReachability$} (l1)
		(l0) edge[bend left] node[align=center]{$\clock < \maxConstantTReachability$ \\ $\clock \assign 0$} (l2)
		(lf) edge (l1)
		(lpriv) edge[out=0,in=200] (lf')
		(l1) edge[bend left] (lf')
		(l2) edge[bend left] node[align=center]{$\clock = 0$} (lf')
		(l1) edge[loop above] node[align=center] {$\enervari{i} : + 1$} (l1)
		(l1) edge[loop below] node[align=center] {$\enervari{i} : - 1$} (l1)
		(l2) edge[loop above] node[align=center] {$\enervari{i} : + 1$} (l2)
		(lpriv) edge[loop above] node[right, align=center] {$\enervari{i} : + 1$} (lpriv)
		;

	\end{tikzpicture}

	}
	\caption{Solving reachability in constant time~$\maxConstantTReachability$ using ET-EN-opacity for discrete METAs}
	\label{figure:ETEN-opacity-harder-than-reachability-discrete-ETAs}
\end{figure*}
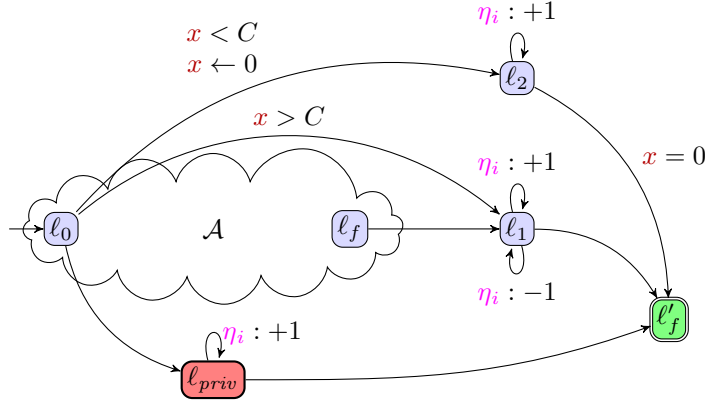
\begin{proposition}\label{proposition:exists-ETEN-opacity-harder-reachability-discrete-METAs}
	Reachability in discrete METAs can be solved using $\exists$-ET-EN-opacity.
\end{proposition}
\begin{proof}
	One can reuse the construction from the proof of \cref{proposition:EN-opacity-harder-reachability-discrete-METAs} (given in \cref{figure:EN-opacity-harder-than-reachability-discrete-ETAs}).
	(The self-loops over $\loci{1}$ are actually not needed here.)
	Due to the absence of invariant in~$\locpriv$, $\locfinal'$ can be reached for any duration and energy valuation via $\locpriv$; then, there is at least one duration and one energy valuation reaching~$\locfinal'$ without visiting~$\locpriv$ (via $\locfinal$ and~$\loci{1}$) iff $\locfinal$ is reachable in~$\TA$.
	As a consequence, $\exists$-ET-EN-opacity holds in~$\TA'$ iff $\locfinal$ is reachable in~$\TA$.
\end{proof}
\begin{proposition}\label{proposition:ETEN-opacity-harder-reachability-discrete-METAs}
	Constant-time reachability in discrete METAs can be solved using weak-ET-EN-opacity (resp.\ full-ET-EN-opacity).
\end{proposition}
\begin{proof}
	Let $\TA$ be a discrete META with initial location $\locinit$; let $\locfinal$ be one of its locations.
	We let $\TA'$ be the discrete META extending~$\TA$ as follows:
	\begin{itemize}
		\item we add a fresh location $\locpriv$ reachable from~$\locinit$ without guard;
		\item we add a fresh location $\loci{1}$ reachable from~$\locfinal$ without guard;
		\item we add a fresh location $\locfinal'$ reachable from~$\loci{1}$ and from~$\locpriv$ without guard;
		\item we add a fresh location $\loci{2}$ reachable from~$\locinit$, guarded by $\clock < \maxConstantTReachability$ and resetting~$\clock$ (where $\clock$ is an arbitrary clock from~$\TA$);
		\item we add a transition from~$\locinit$ to~$\loci{1}$ guarded by $\clock > \maxConstantTReachability$;
		\item we add an urgent transition from~$\loci{2}$ to~$\locfinal'$ (\ie{} guarded by $\clock = 0$);
		\item for each energy variable $\enervari{i}$, we add a self-loop over~$\locpriv$ (resp.\ over~$\loci{2}$) incrementing~$\enervari{i}$;
		\item for each energy variable $\enervari{i}$, we add two self-loops over~$\loci{1}$, one incrementing and one decrementing~$\enervari{i}$.
	\end{itemize}
	$\{ \locpriv \}$ is the private location set of~$\TA'$, while $\{ \locfinal' \}$ is its final location set.
	\cref{figure:ETEN-opacity-harder-than-reachability-discrete-ETAs} exemplifies this construction.

	Let $\maxConstantTReachability \in \setN$.
	First note that any energy and any duration can reach~$\locfinal'$ via~$\locpriv$.
	Second, note that any duration $> \maxConstantTReachability$ (resp.\ $< \maxConstantTReachability$) and any energy can reach~$\locfinal'$ without visiting~$\locpriv$, via~$\loci{1}$ (resp.\ via~$\loci{2}$).
	Third, any energy for duration $= \maxConstantTReachability$ can reach~$\locfinal'$ via~$\locfinal$ and~$\loci{1}$ without visiting~$\locpriv$, iff $\locfinal$ is reachable in time~$\maxConstantTReachability$.

	From this construction, it follows that weak-ET-EN-opacity and full-ET-EN-opacity hold in~$\TA'$ iff $\locfinal$ is reachable in~$\TA$ in time~$\maxConstantTReachability$.
\end{proof}
\subsubsection{Proof of \cref{proposition:undecidability-2-energies}}

We can finally prove \cref{proposition:undecidability-2-energies}.

\undecTwoEnergies*
\begin{proof}
	Recall that:
	\begin{itemize}
		\item For each $\existsWeakFull \in \{ \exists, \text{weak}, \text{full} \}$, $\existsWeakFull$-EN-opacity is harder than reachability in discrete METAs (\cref{proposition:EN-opacity-harder-reachability-discrete-METAs});
		\item $\exists$-ET-EN-opacity is harder than reachability in discrete METAs (\cref{proposition:exists-ETEN-opacity-harder-reachability-discrete-METAs});
		\item Weak-ET-EN-opacity and full-ET-EN-opacity are harder than constant-time reachability in discrete METAs (\cref{proposition:ETEN-opacity-harder-reachability-discrete-METAs});
		\item Reachability is undecidable in discrete guarded METAs with two energy variables and a single clock (\cref{lemma:undecidability-META-2-1});
			and
		\item Constant-time reachability is undecidable in discrete guarded METAs with two energy variables and a single clock (\cref{corollary:undecidability-constant-time-META-2-1}).
	\end{itemize}

	As a consequence, for each $\existsWeakFull \in \{ \exists, \text{weak}, \text{full} \}$, and $\ENorETEN \in \{ \text{EN}, \text{ET-EN} \}$, $\existsWeakFull$-$\ENorETEN$-opacity is undecidable in discrete guarded METAs with 2~energy variables and 1~clock.
\end{proof}
\subsection{Proof of \cref{theorem:EN:discrete-guarded-ETAs}}\label{appendix:theorem:EN:discrete-guarded-ETAs}

A first attempt to prove \cref{theorem:EN:discrete-guarded-ETAs} would be to follow the reasoning from \cref{proposition:removing-guards-discrete}, by using copies of the guarded META, each copy encoding the current (discrete) value of the energy.
However, this leads to the following problem: the last copy (encoding a value of energy greater than the maximum constant~$\maxConstantEnergy$ used in energy guards) does not keep track of the exact energy value---this prevents decreasing the energy to the appropriate copy.
We will therefore combine two ideas, namely from the proofs of \cref{proposition:removing-guards-discrete,theorem:EN:discrete-ETAs}.

\theoremENDiscreteGuardedETAs*
\begin{proof}[Proof (sketch)]
	To circumvent the issue in the proof of \cref{proposition:removing-guards-discrete} losing the value of the energy after it exceeds the maximum constant~$\maxConstantEnergy$ used in energy guards, we will use a stack.
	To avoid using two stacks (in addition to the stack used in the proof of \cref{theorem:EN:discrete-ETAs}), we use the \emph{same} stack, in addition to the copy mechanism, which we use this time in the region automaton.

	We embed the energy guards into the actions, \ie{} each energy guard becomes a fresh action, potentially associated to the original action (if any).
	The region automaton is therefore equivalent to the system without any energy guard---but we will reintroduce them subsequently.
	Once the region automaton is built, we first copy $\maxConstantEnergy + 2$ times this automaton, in the line of the mechanism in the proof of \cref{proposition:removing-guards-discrete}.
	We then evaluate the transitions when needed: an energy ``hard-coded'' as an action in one of the first $\maxConstantEnergy+1$ copies can now either be replaced with a silent action (or with the original action) if the current copy satisfies the energy guard, or its transition completely deleted otherwise. Note that if a decrement is found in the copy encoding the energy 0, then the action must be removed, as our systems doesn't allow energies below zero.
	In the last copy of the region automaton (encoding $> \maxConstantEnergy$), we use the stack:
	when performing increments or decrements (actions $\tickIncr$ and~$\tickDecr$ in the region automata), we start pushing symbols or popping symbols, in the line of the proof of \cref{theorem:EN:discrete-ETAs}.
	Finally, we modify the last items from the proof of \cref{theorem:EN:discrete-ETAs} as follows:
	in addition to the self-loops over accepting regions~$\qfinal$, we then add for each such region a transition labelled via~$\action$ to the copy encoding an immediately lower energy level.
	Only regions $\qfinal'$ in the copy encoding energy~0 are made accepting.
	This ensures that, for a value of energy greater than the maximum constant, first the stack is emptied with actions~$\action$; only after the stack is emptied, the automaton has to go down to the level~0, using an action~$\action$ each time a level is traversed.
	The number of $\action$ correctly encodes the final energy level.

	\cref{figure:example-theorem:EN:discrete-guarded-ETAs-d-g-ETA} exemplifies our construction.
\end{proof}

	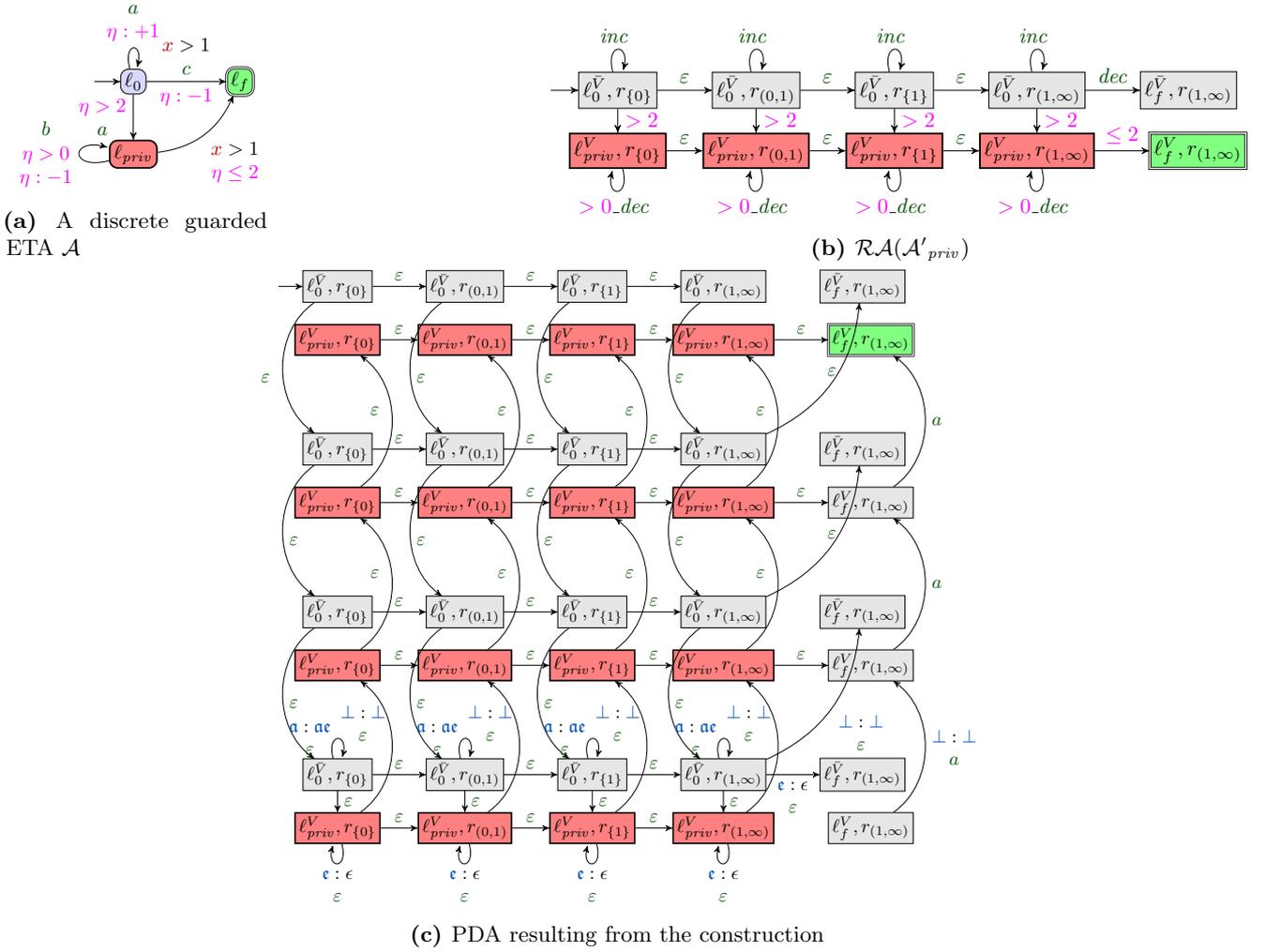
\begin{figure*}[tb]
		\begin{subfigure}[b]{.45\columnwidth}
			\centering
			\scalebox{.8}{
				\begin{tikzpicture}[PTA, node distance=.8cm and 1.4cm]
					\node[location, initial] (l0) {$\locinit$};
					\node[location, private, below =of l0] (lpriv) {$\locpriv$};
					\node[location, final, right =of l0] (lF) {$\locfinal$};

					\path (l0) edge[loop above] node[above, align=center]{$\styleact{a}$ \\ $\styleenergy{\enervar : +1}$} (l0);
					\path (l0) edge node[left, align=center]{$\styleenergy{\enervar > 2}$ \\ $\styleact{a}$} (lpriv);
					\path (l0) edge node[above, align=center]{$\clock > 1$ \\ $\styleact{c}$} node[below, align=center]{$\styleenergy{\enervar : -1}$} (lF);
					\path (lpriv) edge[loop left] node[left, align=center]{ $\styleact{b}$ \\ $\styleenergy{\enervar > 0}$ \\ $\styleenergy{\enervar : -1}$} (lpriv);
					\path (lpriv) edge[bend right] node[below right, align=center]{$\clock > 1$ \\ $\styleenergy{\enervar \leq 2}$} (lF);
				\end{tikzpicture}
			}
			
			\caption{A discrete guarded ETA~$\TA$}
			\label{figure:example-d-g-ETA}
		\end{subfigure}
		\hfill{}
		\begin{subfigure}[b]{\linewidth}
			\centering
			\scalebox{.85}{
				\begin{tikzpicture}[NFA, node distance=.4cm and .9cm]
					\node[state, initial] (l0s0) {$\locinit^{\bar{V}},\regioni{\{0\}}$};
					\node[state, private, below = of l0s0] (lps0) {$\locpriv^V,\regioni{\{0\}}$};
					\node[state,  right = of l0s0] (l0s1) {$\locinit^{\bar{V}},\regioni{(0,1)}$};
					\node[state, private, below = of l0s1] (lps1) {$\locpriv^V,\regioni{(0,1)}$};
					\node[state, right = of l0s1, xshift=0em] (l0s2) {$\locinit^{\bar{V}},\regioni{\{1\}}$};
					\node[state, private, below = of l0s2] (lps2) {$\locpriv^V,\regioni{\{1\}}$};
					\node[state, right = of l0s2] (l0s3) {$\locinit^{\bar{V}},\regioni{(1,\infty)}$};
					\node[state, private, below = of l0s3] (lps3) {$\locpriv^V,\regioni{(1,\infty)}$};
					\node[state, right = of l0s3] (lFs3) {$\locfinal^{\bar{V}},\regioni{(1,\infty)}$};
					\node[state, final, right = of lps3] (lFps3) {$\locfinal^V,\regioni{(1,\infty)}$};
					
					\path (l0s0) edge node[right, align=center]{$\styleenergy{>2}$} (lps0);
					\path (l0s1) edge node[right, align=center]{$\styleenergy{>2}$} (lps1);
					\path (l0s2) edge node[right, align=center]{$\styleenergy{>2}$} (lps2);
					\path (l0s3) edge node[right, align=center]{$\styleenergy{>2}$} (lps3);
					
					\path (l0s0) edge[loop above] node[above, align=center]{ $\styleact{\tickIncr}$ } (l0s0);
					\path (l0s1) edge[loop above] node[above, align=center]{ $\styleact{\tickIncr}$ } (l0s1);
					\path (l0s2) edge[loop above] node[above, align=center]{ $\styleact{\tickIncr}$ } (l0s2);
					\path (l0s3) edge[loop above] node[above, align=center]{ $\styleact{\tickIncr}$ } (l0s3);

					\path (lps0) edge[loop below] node[below, align=center]{ $\styleenergy{>0}\_\styleact{\tickDecr}$ } (lps0);
					\path (lps1) edge[loop below] node[below, align=center]{ $\styleenergy{>0}\_\styleact{\tickDecr}$ } (lps1);
					\path (lps2) edge[loop below] node[below, align=center]{ $\styleenergy{>0}\_\styleact{\tickDecr}$ } (lps2);
					\path (lps3) edge[loop below] node[below, align=center]{ $\styleenergy{>0}\_\styleact{\tickDecr}$ } (lps3);
					
					\path (l0s3) edge node[above, align=center]{$\styleact{\tickDecr}$} (lFs3);
					\path (lps3) edge node[above, align=center]{$\styleenergy{\leq 2}$} (lFps3);
					
					\path (l0s0) edge node[above, align=center]{$\silentaction$} (l0s1);
					\path (l0s1) edge node[above, align=center]{$\silentaction$} (l0s2);
					\path (l0s2) edge node[above, align=center]{$\silentaction$} (l0s3);	
					
					\path (lps0) edge node[above, align=center]{$\silentaction$} (lps1);
					\path (lps1) edge node[above, align=center]{$\silentaction$} (lps2);
					\path (lps2) edge node[above, align=center]{$\silentaction$} (lps3);		

				\end{tikzpicture}
			}
			
			\caption{$\RegionAutomaton{\Apriv{\TA'}}$}
			\label{figure:example-d-g-ETA-Apriv}
		\end{subfigure}

			\begin{subfigure}[b]{\linewidth}
			\centering
			\scalebox{.75}{
				\begin{tikzpicture}[NFA, node distance=.4cm and 1cm]
					\node[state, initial] (l0s0-0) {$\locinit^{\bar{V}},\regioni{\{0\}}$};
					\node[state, private, below = of l0s0-0] (lps0-0) {$\locpriv^V,\regioni{\{0\}}$};
					\node[state,  right = of l0s0-0] (l0s1-0) {$\locinit^{\bar{V}},\regioni{(0,1)}$};
					\node[state, private, below = of l0s1-0] (lps1-0) {$\locpriv^V,\regioni{(0,1)}$};
					\node[state, right = of l0s1-0, xshift=0em] (l0s2-0) {$\locinit^{\bar{V}},\regioni{\{1\}}$};
					\node[state, private, below = of l0s2-0] (lps2-0) {$\locpriv^V,\regioni{\{1\}}$};
					\node[state, right = of l0s2-0] (l0s3-0) {$\locinit^{\bar{V}},\regioni{(1,\infty)}$};
					\node[state, private, below = of l0s3-0] (lps3-0) {$\locpriv^V,\regioni{(1,\infty)}$};
					\node[state, right = of l0s3-0] (lFs3-0) {$\locfinal^{\bar{V}},\regioni{(1,\infty)}$};
					\node[state, final, right = of lps3-0] (lFps3-0) {$\locfinal^V,\regioni{(1,\infty)}$};
					
					\node[state, below = of lps0-0, yshift = -3em] (l0s0-1) {$\locinit^{\bar{V}},\regioni{\{0\}}$};
					\node[state, private, below = of l0s0-1] (lps0-1) {$\locpriv^V,\regioni{\{0\}}$};
					\node[state,  right = of l0s0-1] (l0s1-1) {$\locinit^{\bar{V}},\regioni{(0,1)}$};
					\node[state, private, below = of l0s1-1] (lps1-1) {$\locpriv^V,\regioni{(0,1)}$};
					\node[state, right = of l0s1-1, xshift=0em] (l0s2-1) {$\locinit^{\bar{V}},\regioni{\{1\}}$};
					\node[state, private, below = of l0s2-1] (lps2-1) {$\locpriv^V,\regioni{\{1\}}$};
					\node[state, right = of l0s2-1] (l0s3-1) {$\locinit^{\bar{V}},\regioni{(1,\infty)}$};
					\node[state, private, below = of l0s3-1] (lps3-1) {$\locpriv^V,\regioni{(1,\infty)}$};
					\node[state, right = of l0s3-1] (lFs3-1) {$\locfinal^{\bar{V}},\regioni{(1,\infty)}$};
					\node[state, right = of lps3-1] (lFps3-1) {$\locfinal^V,\regioni{(1,\infty)}$};
					
					\node[state, below = of lps0-1, yshift = -3em] (l0s0-2) {$\locinit^{\bar{V}},\regioni{\{0\}}$};
					\node[state, private, below = of l0s0-2] (lps0-2) {$\locpriv^V,\regioni{\{0\}}$};
					\node[state,  right = of l0s0-2] (l0s1-2) {$\locinit^{\bar{V}},\regioni{(0,1)}$};
					\node[state, private, below = of l0s1-2] (lps1-2) {$\locpriv^V,\regioni{(0,1)}$};
					\node[state, right = of l0s1-2, xshift=0em] (l0s2-2) {$\locinit^{\bar{V}},\regioni{\{1\}}$};
					\node[state, private, below = of l0s2-2] (lps2-2) {$\locpriv^V,\regioni{\{1\}}$};
					\node[state, right = of l0s2-2] (l0s3-2) {$\locinit^{\bar{V}},\regioni{(1,\infty)}$};
					\node[state, private, below = of l0s3-2] (lps3-2) {$\locpriv^V,\regioni{(1,\infty)}$};
					\node[state, right = of l0s3-2] (lFs3-2) {$\locfinal^{\bar{V}},\regioni{(1,\infty)}$};
					\node[state, right = of lps3-2] (lFps3-2) {$\locfinal^V,\regioni{(1,\infty)}$};
					
					\node[state, below = of lps0-2, yshift = -3em] (l0s0-3) {$\locinit^{\bar{V}},\regioni{\{0\}}$};
					\node[state, private, below = of l0s0-3] (lps0-3) {$\locpriv^V,\regioni{\{0\}}$};
					\node[state,  right = of l0s0-3] (l0s1-3) {$\locinit^{\bar{V}},\regioni{(0,1)}$};
					\node[state, private, below = of l0s1-3] (lps1-3) {$\locpriv^V,\regioni{(0,1)}$};
					\node[state, right = of l0s1-3, xshift=0em] (l0s2-3) {$\locinit^{\bar{V}},\regioni{\{1\}}$};
					\node[state, private, below = of l0s2-3] (lps2-3) {$\locpriv^V,\regioni{\{1\}}$};
					\node[state, right = of l0s2-3] (l0s3-3) {$\locinit^{\bar{V}},\regioni{(1,\infty)}$};
					\node[state, private, below = of l0s3-3] (lps3-3) {$\locpriv^V,\regioni{(1,\infty)}$};
					\node[state, right = of l0s3-3] (lFs3-3) {$\locfinal^{\bar{V}},\regioni{(1,\infty)}$};
					\node[state, right = of lps3-3] (lFps3-3) {$\locfinal^V,\regioni{(1,\infty)}$};

					\path (l0s0-0) edge[bend right=55] node[below left, align=center]{ $\silentaction$ } (l0s0-1);
					\path (l0s1-0) edge[bend right=55] node[below right, align=center]{ $\silentaction$ } (l0s1-1);
					\path (l0s2-0) edge[bend right=55] node[below right, align=center]{ $\silentaction$ } (l0s2-1);
					\path (l0s3-0) edge[bend right=55] node[below right, align=center]{ $\silentaction$ } (l0s3-1);
					
					\path (l0s0-1) edge[bend right=55] node[below right, align=center]{ $\silentaction$ } (l0s0-2);
					\path (l0s1-1) edge[bend right=55] node[below right, align=center]{ $\silentaction$ } (l0s1-2);
					\path (l0s2-1) edge[bend right=55] node[below right, align=center]{ $\silentaction$ } (l0s2-2);
					\path (l0s3-1) edge[bend right=55] node[below right, align=center]{ $\silentaction$ } (l0s3-2);
					
					\path (l0s0-2) edge[bend right=55] node[below right, align=center]{ $\silentaction$ } (l0s0-3);
					\path (l0s1-2) edge[bend right=55] node[below right, align=center]{ $\silentaction$ } (l0s1-3);
					\path (l0s2-2) edge[bend right=55] node[below right, align=center]{ $\silentaction$ } (l0s2-3);
					\path (l0s3-2) edge[bend right=55] node[below right, align=center]{ $\silentaction$ } (l0s3-3);

					\path (lps3-0) edge node[above, align=center]{$\silentaction$} (lFps3-0);
					
					\path (l0s0-0) edge node[above, align=center]{$\silentaction$} (l0s1-0);
					\path (l0s1-0) edge node[above, align=center]{$\silentaction$} (l0s2-0);
					\path (l0s2-0) edge node[above, align=center]{$\silentaction$} (l0s3-0);	
					
					\path (lps0-0) edge node[above, align=center]{$\silentaction$} (lps1-0);
					\path (lps1-0) edge node[above, align=center]{$\silentaction$} (lps2-0);
					\path (lps2-0) edge node[above, align=center]{$\silentaction$} (lps3-0);

					\path (lps0-1) edge[bend right=55] node[above left, align=center]{ $\silentaction$ } (lps0-0);
					\path (lps1-1) edge[bend right=55] node[above left, align=center]{ $\silentaction$ } (lps1-0);
					\path (lps2-1) edge[bend right=55] node[above left, align=center]{ $\silentaction$ } (lps2-0);
					\path (lps3-1) edge[bend right=55] node[above left, align=center]{ $\silentaction$ } (lps3-0);
					
					\path (l0s3-1) edge[bend right] node[above, align=center]{$\silentaction$} (lFs3-0);
					\path (lps3-1) edge node[above, align=center]{$\silentaction$} (lFps3-1);
					
					\path (l0s0-1) edge node[above, align=center]{$\silentaction$} (l0s1-1);
					\path (l0s1-1) edge node[above, align=center]{$\silentaction$} (l0s2-1);
					\path (l0s2-1) edge node[above, align=center]{$\silentaction$} (l0s3-1);	
					
					\path (lps0-1) edge node[above, align=center]{$\silentaction$} (lps1-1);
					\path (lps1-1) edge node[above, align=center]{$\silentaction$} (lps2-1);
					\path (lps2-1) edge node[above, align=center]{$\silentaction$} (lps3-1);

					\path (lps0-2) edge[bend right=55] node[above left, align=center]{ $\silentaction$ } (lps0-1);
					\path (lps1-2) edge[bend right=55] node[above left, align=center]{ $\silentaction$ } (lps1-1);
					\path (lps2-2) edge[bend right=55] node[above left, align=center]{ $\silentaction$ } (lps2-1);
					\path (lps3-2) edge[bend right=55] node[above left, align=center]{ $\silentaction$ } (lps3-1);
					
					\path (l0s3-2) edge[bend right] node[above, align=center]{$\silentaction$} (lFs3-1);
					\path (lps3-2) edge node[above, align=center]{$\silentaction$} (lFps3-2);
					
					\path (l0s0-2) edge node[above, align=center]{$\silentaction$} (l0s1-2);
					\path (l0s1-2) edge node[above, align=center]{$\silentaction$} (l0s2-2);
					\path (l0s2-2) edge node[above, align=center]{$\silentaction$} (l0s3-2);	
					
					\path (lps0-2) edge node[above, align=center]{$\silentaction$} (lps1-2);
					\path (lps1-2) edge node[above, align=center]{$\silentaction$} (lps2-2);
					\path (lps2-2) edge node[above, align=center]{$\silentaction$} (lps3-2);

					\path (l0s0-3) edge node[right, align=center]{$\silentaction$} (lps0-3);
					\path (l0s1-3) edge node[right, align=center]{$\silentaction$} (lps1-3);
					\path (l0s2-3) edge node[right, align=center]{$\silentaction$} (lps2-3);
					\path (l0s3-3) edge node[right, align=center]{$\silentaction$} (lps3-3);
					
					\path (l0s0-3) edge[loop above] node[left, align=center]{ $\stackSymbol : \stackSymbol \PDAstackEnergy$ \\ $\silentaction$} (l0s0-3);
					\path (l0s1-3) edge[loop above] node[left, align=center]{ $\stackSymbol : \stackSymbol \PDAstackEnergy$ \\ $\silentaction$} (l0s1-3);
					\path (l0s2-3) edge[loop above] node[left, align=center]{ $\stackSymbol : \stackSymbol \PDAstackEnergy$ \\ $\silentaction$} (l0s2-3);
					\path (l0s3-3) edge[loop above] node[left, align=center]{ $\stackSymbol : \stackSymbol \PDAstackEnergy$ \\ $\silentaction$} (l0s3-3);

					\path (lps0-3) edge[bend right=55] node[above left, align=center]{ $\initialStackSymbol : \initialStackSymbol$ \\ $\silentaction$ } (lps0-2);
					\path (lps1-3) edge[bend right=55] node[above left, align=center]{ $\initialStackSymbol : \initialStackSymbol$ \\ $\silentaction$ } (lps1-2);
					\path (lps2-3) edge[bend right=55] node[above left, align=center]{ $\initialStackSymbol : \initialStackSymbol$ \\ $\silentaction$ } (lps2-2);
					\path (lps3-3) edge[bend right=55] node[above left, align=center]{ $\initialStackSymbol : \initialStackSymbol$ \\ $\silentaction$ } (lps3-2);
					
					\path (lps0-3) edge[loop below] node[below, align=center]{ $\PDAstackEnergy : \emptyword$ \\ $\silentaction$ } (lps0-3);
					\path (lps1-3) edge[loop below] node[below, align=center]{ $\PDAstackEnergy : \emptyword$ \\ $\silentaction$ } (lps1-3);
					\path (lps2-3) edge[loop below] node[below, align=center]{ $\PDAstackEnergy : \emptyword$ \\ $\silentaction$ } (lps2-3);
					\path (lps3-3) edge[loop below] node[below, align=center]{ $\PDAstackEnergy : \emptyword$ \\ $\silentaction$ } (lps3-3);
					
					\path (l0s3-3) edge node[below, align=center]{$\PDAstackEnergy : \emptyseq$ \\ $\silentaction$} (lFs3-3);
					\path (l0s3-3) edge[bend right] node[below right, align=center]{$\initialStackSymbol : \initialStackSymbol$ \\ $\silentaction$} (lFs3-2);

					\path (l0s0-3) edge node[above, align=center]{$\silentaction$} (l0s1-3);
					\path (l0s1-3) edge node[above, align=center]{$\silentaction$} (l0s2-3);
					\path (l0s2-3) edge node[above, align=center]{$\silentaction$} (l0s3-3);	
					
					\path (lps0-3) edge node[above, align=center]{$\silentaction$} (lps1-3);
					\path (lps1-3) edge node[above, align=center]{$\silentaction$} (lps2-3);
					\path (lps2-3) edge node[above, align=center]{$\silentaction$} (lps3-3);

					\path (lFps3-3) edge[bend right=55] node[right, align=center]{ $\initialStackSymbol : \initialStackSymbol$ \\ $\styleact{\action}$ } (lFps3-2);
					\path (lFps3-2) edge[bend right=55] node[right, align=center]{ $\styleact{\action}$ } (lFps3-1);
					\path (lFps3-1) edge[bend right=55] node[right, align=center]{ $\styleact{\action}$ } (lFps3-0);

				\end{tikzpicture}
			}
			
			\caption{PDA resulting from the construction}
			\label{figure:example-PDA-theorem:EN:discrete-guarded-ETAs}
		\end{subfigure}
			
		\caption{Exemplifying the construction of \cref{theorem:EN:discrete-guarded-ETAs}}
		\label{figure:example-theorem:EN:discrete-guarded-ETAs-d-g-ETA}
	\end{figure*}
	\section{Proofs of \cref{section:DE}}
	\subsection{Proof of \cref{proposition:undecidability-2-energies-DE}}\label{appendix:proof:proposition:undecidability-2-energies-DE}
	\undecTwoEnergiesDE*
	\begin{proof}
		The proof if a straightforward adaptation of the proof of \cref{proposition:undecidability-2-energies}.
		First note that, for a system running in time $\leq 1$, DE-opacity is equivalent to EN-opacity, since observing the energy every time unit is equivalent to observe energy only at the end of the execution.
		The proof then follows by choosing $\maxConstantTReachability = 1$ in the proof of \cref{proposition:ETEN-opacity-harder-reachability-discrete-METAs,proposition:undecidability-2-energies}.
	\end{proof}
	\subsection{Proof of \cref{theorem:DE-Opacity}}\label{appendix:proof:theorem:DE-Opacity}
	\theoremDEopacity*
	\begin{proof}
		Let $\TA$ be a discrete positive ETA.
		We construct $\Language(\RegionAutomaton{\Apriv{\TA''}})$ and $\Language(\RegionAutomaton{\Apub{\TA''}})$ the same way as in the proof of \cref{theorem:ET-EN-fullweak-discrete-positive-guarded-METAs}, except that we replace the last action $\tickLastTimeRational$ with~$\tickTime$;
			this way, we encode the fact that we observe the last time unit, without making a difference between energy changes occurring in 0-time before the last observable time, or in the preceding open interval.
		Verifying weak DE-opacity amounts to checking $\Language(\RegionAutomaton{\Apriv{\TA''}}) \subseteq \Language(\RegionAutomaton{\Apub{\TA''}})$, which is a language inclusion problem in non-deterministic finite automata (over 2~symbols), which is \PSPACE-complete~\cite{KMT17}.
		Full DE-opacity is similar.
		$\exists$-DE-opacity is solved by checking the non-emptiness of the intersection, which can be verified in polynomial time~\cite{KMT17}.
		The complexity remains unchanged for positive \emph{guarded} ETAs since the transformation in the proof of \cref{proposition:removing-guards-discrete} is linear for~ETAs.
	\end{proof}
	\subsection{Details on the proof of \cref{theorem:DE-Opacity-META}}\label{appendix:proof:theorem:DE-Opacity-META}

	We give additional details regarding the formalisation of our construction of ``Parikh by block'' automaton.
	Given an NFA~$\TA$ with alphabet $\Actions$, and given $\Actions' \subseteq \Actions$, the construction $\ParikbB(\TA)$ will construct an NFA (technically an extension of NFAs with two labels, one being an action as usual, and the other one being a semilinear set), encoding the number of actions of each symbol in $\Actions \setminus \Actions'$ in between two consecutive actions in~$\Actions'$.
	\begin{enumerate}
			\item Add to a list $L$ the initial state of~$\TA$ and every state that has an incoming transition with some action~$\tickTime$, for $\tickTime \in \Actions'$.
			\item For all $\q \in L$, and for all state $\q''$ such that
	there exists in~$\TA$ a transition $(\q',\tickTime,\q'')$,
			compute the Parikh's image of the automaton where
			\begin{itemize}
				\item All transitions with action~$\tickTime$, for $\tickTime \in \Actions'$, are removed,
				\item $\q$ is the initial state,
				\item %
					$\q'$ is the only final state.
			\end{itemize}
			Let $P_\tickTime(\q, \q'')$ denote this semilinear set. %

			\item Finally build a new automaton $\ParikbB(\TA)$ such that:
			\begin{itemize}
				\item The initial state and final states are the same as the original~$\TA$,
				\item For all non-empty set $P_\tickTime(\q, \q'), \tickTime \in \Actions'$,
					add states $\q$ and~$\q'$, and add a transition from $\q$ to~$\q'$ with $\tickTime$ as action, and with the semilinear $P_\tickTime(\q,\q')$ as label.
			\end{itemize}
		\end{enumerate}

		\theoremDEopacityMETA*

		\begin{proof}
		We begin by constructing the Parikh by block automaton of each NFA over $\Actions' = \{ \tickTime , \tickFin \}$, \ie{} we build $\ParikbB(\RegionAutomaton{\Apriv{\TA'}}) $ and $\ParikbB(\RegionAutomaton{\Apub{\TA'}})$.
		These structures encode exactly the number of increments between two time ticks.

		Now construct the synchronised product of $\ParikbB(\RegionAutomaton{\Apriv{\TA'}}) $ and $\ParikbB(\RegionAutomaton{\Apub{\TA'}})$ over $\{ \tickTime, \tickFin \}$, as follows:
		in the resulting product, each transition has the original synchronised label (in $\{ \tickTime, \tickFin \}$), while the other component is the pair of semilinear sets, coming from each of the two synchronised automata.
		Then, for each synchronised transition, we verify the emptiness of the intersection between these two semilinear sets.
		Because the Parikh image of a regular language is a semilinear set~\cite{Parikh66}, this can be checked using Presburger arithmetic~\cite{BHK17}.
		After removing transitions with such an empty intersection, verifying $\exists$-DE-opacity simply consists in exhibiting an accepting run in that structure.
		\end{proof}

		Finally recall that we exemplify the construction in \cref{example-Parikh-by-blocks}.

	\subsection{Proof of \cref{theorem:bDE-Opacity}}\label{appendix:proof:theorem:bDE-Opacity}
	\theorembDEopacityPositive*
	\begin{proof}
		We use inclusion and intersection problems over non-deterministic finite automata, with a modified construction.
		Given a discrete positive META~$\TA$, we transform it into~$\TA'$ as follows:
		\begin{enumerate}
			\item we add a global ``ticking clock'' $\clockTickOne$ reset every time unit via action~$\tickTime$ as in the proof of \cref{theorem:ET-EN-fullweak-discrete-positive-guarded-METAs};

			\item we eliminate discrete increments %
				$\neq 1$ by splitting them into consecutive
					$+1$
				updates in 0-time, using an extra clock~$\clockZeroTime$;
			\item we add after each update (or series of updates in 0-time) a last transition in 0-time via a fresh action~$\tickFin$;
			\item we relabel each transition of the form $\enervari{i}:+1$ %
				with action $\tickIncr_i$%
				, and make any other action silent;
			\item for each final location~$\locfinal$, we add a new location~$\locfinal'$ (the only ones to be final) reachable from $\locfinal$ via action~$\tickTime$ if $\locfinal$ is reachable on an integer time, and via~$\silentaction$ otherwise.
		\end{enumerate}
		We exemplify our transformation (applied to \cref{figure:example-DiscPosGuarETA}) in \cref{figure:example-bDEO}.

		Now, there is a one-to-one correspondence between the buffered discrete energy observation of the runs of~$\TA$, and the untimed language of~$\TA'$.
		For example, given the private run~$\varrun_1 = (\locinit, 0, 0) \FlecheConcrete{(0.8, \styleact{a})} (\locpriv, 0.8, 0) \FlecheConcrete{(0.3, \styleact{b})} (\locpriv, 1.1, 1) \FlecheConcrete{(0, \styleact{b})} (\locpriv, 1.1, 2) \FlecheConcrete{(0, \styleact{b})} (\locfinal, 1.1, 2) $
		of~$\TA$ in \cref{figure:example-DiscPosGuarETA},
		the untimed word will be $\tickTime \styleact{\tickIncr_1} \tickFin \styleact{\tickIncr_1} \tickFin$.
		The public run $\varrun_2 = (\locinit, 0, 0) \FlecheConcrete{(1.1, \styleact{c})} (\locfinal, 1.1, 2) $ gives in contrast $\tickTime \styleact{\tickIncr_1} \styleact{\tickIncr_1} \tickFin$.

		Solving $\existsWeakFull$-bDE-opacity therefore amounts to checking inclusion and intersection over the NFAs $\RegionAutomaton{\Apub{\TA'}}$ and $\RegionAutomaton{\Apriv{\TA'}}$.

		Again, the complexity for discrete positive \emph{guarded} METAs follows from \cref{proposition:removing-guards-discrete} and the fact that the transformation in the proof of \cref{proposition:removing-guards-discrete} for discrete positive guarded METAs leads to an exponential blowup of the META.
	\end{proof}
	\subsection{Proof of \cref{theorem:bDE-Opacity-ETAs}}\label{appendix:proof:theorem:bDE-Opacity-ETAs}
	\theorembDEopacity*
	\begin{proof}
		Given a discrete ETA~$\TA$, we build $\RegionAutomaton{\Apub{\TA'}}$ and $\RegionAutomaton{\Apriv{\TA'}}$ the same way as in \cref{theorem:bDE-Opacity} (see \cref{appendix:proof:theorem:bDE-Opacity}), where increments are encoded by action~$\tickIncr$ and decrements by~$\tickDecr$.
		These regular languages recognize some words that are not valid, typically leading to the energy dropping below~0, \ie{} not in $\LanguageIncrDecr = \{\omega = (\tickDecr+\tickIncr+\tickFin+\tickTime)^* \mid $ for each prefix $ \omega'$ of $\omega, \#_{\tickDecr}(\omega') \leq \#_{\tickIncr}(\omega')\}$.
		This language is context-free: see \cref{figure:PDA-dec-inc} for a pushdown automaton recognizing the language~$\LanguageIncrDecr$.

		\begin{figure}[tb]
				\begin{minipage}{.35\columnwidth}
				\centering
				\begin{tikzpicture}[NFA, node distance=2.5cm]
					\node[state, initial, final] (l0) {$\q$};

					\path (l0) edge[loop right] node[right, align=center]{$\stackSymbol : \stackSymbol$ \\ $\tickIncr$, $\tickTime$, $\tickFin$} (l0);
					\path (l0) edge[loop above] node[above, align=center]{$\stackSymbol : \stackSymbol \PDAstackEnergy$ \\ $\tickIncr$} (l0);
					\path (l0) edge[loop below] node[below, align=center]{$\PDAstackEnergy : \emptyword$ \\ $\tickDecr$} (l0);
				\end{tikzpicture}
				\end{minipage}
				\hfill{}
				\begin{minipage}{.5\columnwidth}
					$\PDAstackAlphabet = \{ \initialStackSymbol , \PDAstackEnergy \}$

					$\stackSymbol$ denotes $\initialStackSymbol$ or $\PDAstackEnergy$
				\end{minipage}

			\caption{Non-deterministic pushdown automaton recognizing the language $\LanguageIncrDecr$}
			\label{figure:PDA-dec-inc}
		\end{figure}
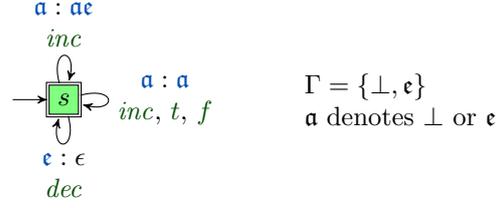

		To verify $\exists$-bDE-Opacity, verifying $ ((\RegionAutomaton{\Apub{\TA'}} \cap \RegionAutomaton{\Apriv{\TA'}} ) \cap \LanguageIncrDecr) $ emptiness is sufficient: the intersection of both region automata is quadratic, but determinising it adds an exponential factor.
		Intersection between a regular language and a context-free language expressed by a PDA is polynomial.
		This gives a \twoEXPSPACE{} procedure.

		For weak and full bDE-opacity, verifying the inclusions $(\RegionAutomaton{\Apub{\TA'}} \cap \LanguageIncrDecr) \subseteq (\RegionAutomaton{\Apriv{\TA'}} \cap \LanguageIncrDecr)$ and $(\RegionAutomaton{\Apriv{\TA'}} \cap \LanguageIncrDecr) \subseteq (\RegionAutomaton{\Apub{\TA'}} \cap \LanguageIncrDecr)$ cannot be done directly, as the inclusion of two context-free languages is not decidable in general \cite{CM14}.
		However, in our case, the only non-regular language is~$\LanguageIncrDecr$.
		Therefore $((\RegionAutomaton{\Apub{\TA'}} \cap (\neg \RegionAutomaton{\Apriv{\TA'}})  \cap \LanguageIncrDecr)$ and $((\RegionAutomaton{\Apriv{\TA'}} \cap (\neg \RegionAutomaton{\Apub{\TA'}})  \cap \LanguageIncrDecr)$ are computable, and the emptiness of both languages gives us the inclusion that we seek.
		Computing these intersections can be done in polynomial time.
		The automaton to be intersected with~$\LanguageIncrDecr$ is doubly exponential in the original~ETA (one exponential for the region automaton, and another one for the determinisation).
		This gives a \twoEXPSPACE{} procedure.
	\end{proof}
\end{document}